\documentclass[a4paper,11pt]{article}
\usepackage{fullpage}
\usepackage{libertine}

\usepackage{microtype}
\usepackage{graphicx}
\usepackage{booktabs} % for professional tables

\usepackage{subcaption}

\usepackage{algorithm} % For algorithms
\usepackage{algorithmic}
\usepackage{bm}
\usepackage{thm-restate}
\usepackage{amssymb,amsmath,amsfonts,amstext,amsthm}
\usepackage{thmtools} % for restate the theorems
\usepackage{mathtools} % for coloneqq:=

\usepackage{arydshln}

\usepackage[dvipsnames]{xcolor}         % colors
\definecolor{pku-red}{RGB}{139,0,18}

\usepackage{url}
\renewcommand\url[1]{#1}

\usepackage[backref, colorlinks,citecolor=blue,linkcolor=pku-red,urlcolor=pku-red,bookmarks=true]{hyperref}
\usepackage[nameinlink]{cleveref}
\usepackage{enumitem}

\usepackage{pgfplots}
\pgfplotsset{compat=newest}
\usepgfplotslibrary{ternary}
\usepackage{tikz}
\usetikzlibrary{patterns}

\usepackage{bbm}
\usepackage{natbib}%
\setcitestyle{sort,authoryear,square,citesep={;},aysep={,},yysep={;}}

\newcommand{\calC}{\mathcal{C}}

\newcommand{\calH}{\mathcal{H}}

\newcommand{\calK}{\mathcal{K}}

\newcommand{\calN}{\mathcal{N}}

\newcommand{\calU}{\mathcal{U}}

\newcommand{\bbE}{\mathbb{E}}

\newcommand{\bbP}{\mathbb{P}}

\newcommand{\bbR}{\mathbb{R}}

\newcommand{\vb}{\bm{b}}

\newcommand{\dd}{\bm{d}}

\newcommand{\xx}{\bm{x}}
\newcommand{\yy}{\bm{y}}
\newcommand{\zz}{\bm{z}}

\newcommand{\bv}{\bm{v}}

\newcommand{\bx}{\bm{x}}

\newcommand{\blambda}{\boldsymbol{\lambda}}
\newcommand{\bmu}{\boldsymbol{\mu}}

\newcommand{\bxi}{\boldsymbol{\xi}}
\newcommand{\brho}{\boldsymbol{\rho}}

\newcommand{\rmC}{\mathrm{C}}
\newcommand{\rmIP}{\mathrm{IP}}
\newcommand{\rmCP}{\mathrm{CP}}
\newcommand{\rmHP}{\mathrm{HP}}

\newcommand{\utilityC}{\Pi^\rmC_k}
\newcommand{\utilityIP}{\Pi^\rmIP_k}
\newcommand{\utilityCP}{\Pi^\rmCP_k}
\newcommand{\utilityHP}{\Pi^\rmHP_k}
\newcommand{\EXPutility}{U_k}
\newcommand{\EXPutilityCP}{U^\rmCP_k}
\newcommand{\EXPutilityHP}{U^\rmHP_k}
\newcommand{\EXPcost}{G_k}
\newcommand{\EXPcostCP}{G^\rmCP_k}
\newcommand{\EXPcostHP}{G^\rmHP_k}
\newcommand{\EXPvalue}{V_k}
\newcommand{\EXPvalueCP}{V^\rmCP_k}
\newcommand{\EXPvalueHP}{V^\rmHP_k}
\newcommand{\EXPSumUtility}{U^{\rm{S}}_k}
\newcommand{\EXPSumCost}{G^{\rm{S}}_k}

\newcommand{\wt}{\widetilde}

\DeclareMathOperator{\argmax}{argmax}

\DeclareMathOperator{\supp}{Supp}

\crefname{assumption}{assumption}{assumptions}
\Crefname{assumption}{Assumption}{Assumptions}

\newtheorem{proposition}{Proposition}[section]
\newtheorem{lemma}[proposition]{Lemma}

\newtheorem*{theorem*}{Theorem}
\newtheorem{corollary}[proposition]{Corollary}

\newtheorem{definition}{Definition}

\newtheorem{assumption}{Assumption}
\newtheorem{claim}{Claim}
\newtheorem{example}{Example}

\usepackage{mdframed}
\newmdenv[skipabove=2mm, skipbelow=2mm, backgroundcolor=black!10]{boxedtext*}
\newmdenv[skipabove=2mm, skipbelow=2mm, rightline=true, leftline=true]{boxedtext}
\newcommand{\email}[1]{\href{mailto:#1}{#1}}

\newcommand\blfootnote[1]{%
  \begingroup
  \renewcommand\thefootnote{}\footnote{#1}%
  \addtocounter{footnote}{-1}%
  \endgroup
}

\Crefname{ALC@unique}{Line}{Lines}
\newcounter{myalg}
\AtBeginEnvironment{algorithmic}{\refstepcounter{myalg}}
\makeatletter
\@addtoreset{ALC@unique}{myalg}
\makeatother

\title{Coordinated Dynamic Bidding in Repeated Second-Price Auctions with Budgets}

\author{
\textbf{Yurong Chen}$^{1,*}$, 
\textbf{Qian Wang}$^{1,*}$\blfootnote{Equal contribution.}, 
\textbf{Zhijian Duan}$^{1}$,
\textbf{Haoran Sun}$^{1}$,\\
\textbf{Zhaohua Chen}$^{1}$,
\textbf{Xiang Yan}$^{2}$,
\textbf{Xiaotie Deng}$^{1}$
\\
$^{1}$Peking University
$^{2}$Huawei TCS Lab
\\
\email{\{chenyurong, charlie, zjduan\}@pku.edu.cn},~
\email{sunhaoran0301@stu.pku.edu.cn},\\
\email{chenzhaohua@pku.edu.cn},~
\email{xyansjtu@163.com},~
\email{xiaotie@pku.edu.cn}
}

\date{}

\begin{document}
\maketitle

\begin{abstract}
  In online ad markets, 
  a rising number of advertisers
  are employing bidding agencies to participate in ad auctions. 
  These agencies are specialized in designing online 
  algorithms and bidding on behalf of their clients. 
  Typically, an agency usually has information on multiple advertisers, so she can potentially coordinate bids to help her clients achieve higher utilities than those under independent bidding. 
  
In this paper,  we study coordinated online bidding algorithms in repeated second-price auctions with budgets. 
  We propose algorithms that guarantee every client a higher utility than the best she can get under independent bidding. 
 We show that these algorithms achieve maximal coalition welfare and discuss bidders' incentives to misreport their budgets, in symmetric cases. Our proofs combine the techniques of online learning and equilibrium analysis,  overcoming the difficulty of competing with a multi-dimensional benchmark. The performance of our algorithms is further evaluated by experiments on both synthetic and real data.  To the best of our knowledge, 
   we are the first to consider bidder coordination in online repeated auctions with constraints. 
\end{abstract}

% Paper body
% \pagestyle{headings}
% \pagenumbering{arabic}
\section{Introduction}

Online advertising has been the main source of revenue for important Internet companies~\citep{iab2021interactive}. A transaction in the market usually goes as follows:
upon the arrival of an ad query, an auction is held by the advertising platform among all advertisers, and the winner gets the opportunity to display her ad. Thus, advertisers engaging in such repeated auctions are confronted with an online decision problem. That is, they could not 
%predict
know the values of incoming queries, but have to decide their bids based on historical information to maximize their accumulated utilities. 

Recently, more and more advertisers are delegating their bidding tasks to specialized market agencies~\citep{Decarolis20marketing}. Advertisers only need to report their objectives and constraints within a given number of rounds. Typically, an advertiser aims to maximize her total utility subject to a budget constraint that limits her total payment.
With the requests received, these agencies will design corresponding online bidding algorithms and bid on behalf of these advertisers.
A bidding agency usually has information of all her clients, which leaves her a chance to coordinate bids to improve every client's utility. 

However, most existing online bidding algorithms 
focus on bidding for one bidder in the campaign~\citep{han2020optimal, golrezaei2021bidding, balseiro2022best, chen2022dynamic, feng2022online, wang2023learning}. The fundamental question of bidder coordination in online settings remains unanswered: can we design a coordinated online bidding strategy, such that each 
client advertiser enjoys a higher utility compared to the best they can get under independent bidding? 

 In this paper, we study bidder coordination 
 in repeated second-price auctions with budgets. The above question can be modeled as an online multi-objective optimization:  to show that an algorithm has the desired performance, one needs to compare each bidder's utility  with her corresponding benchmark utility. The main difficulty, 
compared to the problem of designing individual online bidding algorithms, is how we deal with the interplay of bidders in dynamic repeated settings. 
Since auctions are multi-player games, 
a bidder's utility is influenced by others' bids.

Indeed, one way to overcome the problem is to assume that members who gain more than their respective benchmark utility are able to compensate other members with money. 
In this way, optimizing the welfare of the whole coalition is sufficient~\citep{bachrach2011cooperative, allouah2017auctions}.  However, we argue that monetary transfer is unrealistic in the ad auction scenario. 
In reality, advertisers' first goal is usually to enlarge the influence of their products.
They participate in auctions 
% would enter the coalition aim 
to win more impressions for their ads, rather than simply conducting 
% money-to-money 
financial investments.

Therefore, we focus on coordinated online bidding algorithms without monetary transfer. 
We propose an algorithm, 
\textbf{H}ybrid coordinated \textbf{P}acing (HP), under which each bidder's expected utility is theoretically guaranteed to outperform her corresponding benchmark utility asymptotically. In designing HP, an important observation we follow is that the problem of coordinated online bidding strategy design can be reduced to two design tasks: the selection rule inside the coalition and the budget management strategies. A simple application of this principle, \textbf{C}oordinated \textbf{P}acing (CP), is another algorithm proposed in the paper. It already has the desired utility guarantees when bidders in the coalition share the same value distribution and the same amount of budget (i.e., symmetric case). 

In proving our theorems and overcoming the aforementioned difficulty, we combine the techniques of online learning and equilibrium analysis. We first show that these algorithms' performance converges. The problem then reduces to the comparison of utilities between the benchmark and the equilibrium these algorithms converge to. The analysis of equilibria, which are defined by non-linear complementarity problems and have no closed-form solutions, allows us to show that our algorithms maximize coalition welfare in the symmetric case. We believe all these techniques and results will be of generality and important use in future studies. Finally, we run experiments on both synthetic and real data to further illustrate how our algorithms improve every member's utility through coordination. 

\subsection{Related Work}

In recent years, individual online bidding problems with budget constraints in repeated auctions have been widely discussed in literature, promoted by the surging of auto-bidding services in the industry~\citep{aggarwal2019autobidding}. 
Theoretically, this problem is closely related to contextual bandits with knapsacks 
(CBwK)
~\citep{badanidiyuru2014resourceful, agrawal2016efficient, agrawal2016linear, sivakumar2022smoothed} as well as its specifications, e.g., online allocation problems~\citep{li2021online, balseiro2022best}. 
For the dynamic bidding problem, 
\citet{balseiro2019learning} set the foundation of pacing in repeated second-price auctions, by which a bidder would shade values by a constant factor as her bids. 
They proposed an optimal online bidding algorithm that adaptively adjusts the pacing parameter.
Subsequent works have extended the above results in multiple ways~\citep{golrezaei2021bidding,gaitonde2022budget, celli2022parity, gaitonde2022budget}. 
Our work considers the coordinated bidding scenario, in which bidders could cooperate to increase everyone's utility. 

\citet{graham1987collusive} studied coordinated strategies in second-price auctions and first proposed to select only one member in the coalition as the representative to bid in the auction. 
\citet{pesendorfer2000study, leyton2002bidding, aoyagi2003bid, skrzypacz2004tacit, romano2022power} also consider models that do not allow
% forbid
monetary transfer in the coalition. 
Among these researches, \citet{Decarolis20marketing, romano2022power} studied the computation problem of optimizing each bidder's utility in GSP and VCG auctions. 
In comparison, we consider the coordination problem in online repeated second-price auctions with budgets, which is one of the most practical scenarios in current ad markets.  

Technically, 
for the budget management strategy of each individual bidder, we adopt the adaptive pacing strategy introduced by \citet{balseiro2019learning}, which is based on dual gradient descent.
While algorithms and analysis based on duality are standard in online convex optimization~\citep{nemirovski2009robust, li2021online, golrezaei2021bidding, feng2022online}, due to the introduction of inner dual variables, our analysis of HP does not follow the standard dual analysis. 
The idea of selecting a representative to bid is from \citet{graham1987collusive} and confirmed to be applicable to the online setting in this paper. Moreover, it is important and nontrivial to design a fair adaptive selection rule, so as to make sure that everyone has fairly equal chances to bid:
While CP can be regarded as a natural application of individual pacing to coordinated bidding, it can only guarantee performance in symmetric cases. \citet{zhou2017mirror} also studied the case when all players learn simultaneously in repeated games. They established the last-iterate convergence of mirror descent strategies when the one-shot game satisfies the property of variational stability. However, their results do not apply directly to the settings when players' actions across rounds satisfy certain budget constraints. 
Moreover, our analysis includes an equilibrium analysis that compares per-person utilities in two equilibria.

Other works 
utilize the multi-agent reinforcement learning framework to design algorithms
\citep{JinSLGWZ18, guan21multi, chao22a, tan22learning}. They usually need to explore beforehand to get enough samples. Compared to them, 
the algorithms we propose do not need samples. They consider accumulated utilities and achieve asymptotically better performance than the benchmark.

\subsection{Paper Organization}
In \Cref{sec: model}, we introduce our model and benchmark. In \Cref{sec: cp}, we confirm the idea of selecting representatives to bid, and present the results of CP. \Cref{sec: hp} presents the results of HP. In \Cref{sec: discussion}, we show other properties our algorithms possess in symmetric cases. \Cref{sec: experiments} show experiment results of our algorithms. \Cref{sec: conclusion} concludes.

\section{Model and Benchmark}
\label{sec: model}
In this work, we consider the scenario in which $N$ bidders participate in $T$ rounds of repeated second-price auctions. 
In each round $t = 1, \ldots, T$, there is an available ad slot, auctioned by the advertising platform. 
There are $K$ among $N$ bidders who form a coalition $\calK = \{1, \cdots, K\}$. 
The value bidder $k\in \calK$ perceives for the ad slot in round $t$ is denoted by $v_{k, t}$ and is assumed to be i.i.d. sampled from a distribution $F_k$. 
We assume that $F_k$ has a bounded density function $f_k$, with a support over $[0, \bar{v}_k] \subseteq \bbR_+$. 
Each bidder $k\in \calK$ has a budget constraint $B_k$, which limits the sum of her payments throughout the period,
and we denote bidder $k$'s target expenditure rate by $\rho_k \coloneqq B_k/T\in (0, \bar{v}_k)$. 
As usual, we use bold symbols for vectors (or matrices), e.g., using $\bv_t$ without subscript $k$ to denote the vector $(v_{1, t}, \ldots, v_{K, t})$, and $\bv$ to denote $(\bv_{1}, \ldots, \bv_{T})$; the same goes for $\brho$ as well as other variables to be defined.

For $k \in \calK$, denote bidder $k$'s bid at round $t$ by $b_{k, t}$. For bidders outside the coalition, we assume their highest bid in round $t$, denoted by $d^O_t$, is sampled i.i.d. from a distribution $H$. This modeling follows the standard mean-field approximation \citep{iyer2014mean} when the number of bidders outside the coalition is large. Similar to $F_k$, $H$ is also assumed to have a bounded density function $h$.
Thus, the 
highest 
competing bid that bidder $k$ faces is:
\[
d_{k, t} \coloneqq \max \left\{\max_{i\in \calK: i\ne k}b_{i, t}, \ d^O_t\right\}.
\]
Let $x_{k, t} \coloneqq \bm{1}\left\{b_{k, t} \geq d_{k, t}\right\} \in \{0,1\}$ indicate whether bidder $k$ wins the ad slot in round $t$. 
We denote by $u_{k, t} \coloneqq x_{k, t}(v_{k, t} - d_{k, t})$ bidder $k$'s utility in round $t$ and by $z_{k, t} \coloneqq x_{k, t}d_{k, t}$ her corresponding expenditure for a second-price auction.

Now we formally model the coordination among bidders in $\calK$, which can also be viewed as a third-party agency bidding on behalf of all bidders in $\calK$, given their budget and value information.
We denote by $\calH^t_{\calK}$ the history available to the coalition before posting the bids in round $t$, defined as:
\[\calH^t_{\calK} \coloneqq \left\{\left\{\bv_{\tau}, \xx_{\tau}, \zz_{\tau}\right\}_{\tau=1}^{t-1}, \bv_t\right\}.\] 
% With the notation, 
A coordinated strategy maps $\calH^t_{\calK}$ to a (possibly random) bid vector $\vb_t$ for each $t$. 
We use $\calC$ to denote the set of all coordinated strategies, 
strategy $\rmC \in \calC$ feasible if it guarantees that for each bidder, her cumulative expenditures never exceed her budget for any realizations of values and competing bids, i.e., $\forall\bv, \dd^O$,
\[
    \sum^T_{t = 1} z^{\rmC}_{k,t} = \sum^T_{t = 1} \bm{1}\left\{b^{\rmC}_{k, t} \geq d^{\rmC}_{k, t}\right\} d^{\rmC}_{k, t} \leq B_k, \forall k\in \calK.
\]
We denote by $\Pi^{\rmC}_k$ bidder $k$'s expected payoff under coordinated algorithm $\rmC \in \calC$:
\begin{align*}
\Pi^{\rmC}_k &\coloneqq \bbE^{\rmC}_{\bv, \dd^O}\left[\sum^T_{t=1}u^{\rmC}_{k,t}\right] \\ 
&= \bbE^{\rmC}_{\bv, \dd^O}\left[\sum^T_{t=1} \bm{1}\left\{b^{\rmC}_{k, t} \geq d^{\rmC}_{k, t}\right\} \left(v_{k, t} - d^{\rmC}_{k, t}\right) \right],
\end{align*}
 where the expectation is taken with respect to the randomness of algorithm $\rmC$, the values of all coalition members $\bv$, and the highest bids outside the coalition $\dd^O$.

\subsection{Optimal Individual Budget Management Algorithm and Benchmark} 

A desirable coordinated bidding algorithm should benefit every participant.
That is, all participants should gain more utility than what they could obtain when they bid independently. 
This is also a requirement for stable coordination, as a bidder should have left the coalition if she gained less in it. A natural measurement is 
% to consider
bidders' utilities when everyone uses optimal individual bidding algorithms independently. 

As for individual bidding algorithms, it is known that adaptive pacing strategy \citep{balseiro2019learning} is asymptotically optimal in both stochastic and adversarial environments. The main idea of this strategy is to maintain a shading parameter $\lambda$, and bid $v / (1+\lambda)$ when the true value is $v$. The algorithm updates $\lambda$ according to each round's cost and tries to keep the bidder's average expenditure per auction close to the target expenditure rate. By doing so, the algorithm shares the risk as well as the opportunity among all rounds. 
For the completeness of our work, we state the individual adaptive pacing algorithm in \Cref{sec: ip}.

When all bidders follow adaptive pacing simultaneously, their shading parameters converge, assuming that the multi-valued function of their expected expenditures per round is strongly monotone (defined below). At the same time, their utilities also converge to the corresponding 
% market
equilibrium
utilities, 
which we regard as our benchmark.

Specifically,  consider $K$ bidders shade bids according to a profile of parameters $\blambda \in \bbR^K_+$. Denote bidder $k$'s expected expenditure per auction under $\blambda$ by
\begin{align}
\label{eqn: ip exp cost definition}
    \EXPcost(\blambda)\coloneqq  \bbE_{\bv, d^O}\left[\bm{1}\{v_k\geq (1 + \lambda_k) d_k\}\cdot d_k\right],
\end{align}
and bidder $k$'s expected utility over $T$ rounds by 
\begin{align*}
    \EXPutility(\blambda) \coloneqq T\cdot \bbE_{\bv, d^O}\left[\bm{1}\left\{v_k \geq \left(1+\lambda_k\right)d_k\right\}\cdot \left(v_k - d_k\right)\right],
\end{align*}
where we omit the subscript $t$ since the values and the highest bids are independent across rounds. 
We say function $\bm{F}: \mathbb{R}^K \rightarrow \mathbb{R}^K$ is $\gamma$-strongly monotone over a set $\mathcal{U} \subset \mathbb{R}^K$ if
\[\left(\blambda-\blambda^{\prime}\right)^{\top}\left(\bm{F}(\blambda^{\prime})-\bm{F}(\blambda)\right) \geq \gamma\left\|\blambda-\blambda^{\prime}\right\|_2^2, ~\forall \blambda, \blambda^{\prime} \in \mathcal{U}. \]

\begin{assumption}\label{asm: strongly-monotone}
There exists a constant $\gamma>0$ such that $\bm{G}=(G_k)^K_{k=1}$ is a $\gamma$-strongly monotone function over $\prod_{k=1}^K[0,\bar{v}_k/\rho_k]$.
\end{assumption}

We call the special coordinated algorithm, the simultaneous adoption of individual adaptive pacing, IP for short.
When \Cref{asm: strongly-monotone} holds, parameters under 
IP
converge to a fixed vector $\blambda^*$, defined by the following complementarity conditions:
\begin{equation}
\label{eqn: lambda star definition}
    \lambda^*_k\geq 0 \perp G_k(\blambda^*)\leq \rho_k, \quad \forall k\in \calK,
\end{equation}
where $\perp$ means that at least one condition holds with equality. The existence and uniqueness of $\blambda^*$ are guaranteed by the result of \citet{facchinei2003finite}. Moreover, its long-run average performance 
% under simultaneous adoption of IP 
converges to the expected utilities achieved with vector $\blambda^*$, i.e.,
\begin{align}
    \lim_{T\to \infty} \frac{1}{T}\left(\utilityIP - \EXPutility(\blambda^*) \right) = 0, \quad \forall k\in \calK.
\end{align}

We consider $\bm U(\blambda^*)$ as our benchmark. 
Notably, our goal is to design coordinated algorithms with better performance for every bidder inside the coalition, i.e., finding an algorithm $\rmC$ such that 
\[\liminf_{T \to \infty} \frac{1}{T}\left(\utilityC - U_k(\blambda^*)\right) \geq 0, \quad \forall k\in \calK.\]

\section{First Steps for Coordinated Strategy Design}
\label{sec: cp}

In this section, we present two important elements used in our algorithms. 
We first show that we only need to consider a class of ``bid rotation'' strategies to design satisfying coordinated bidding algorithms. 
We then use a simple algorithm 
to illustrate the non-triviality of the problem and the importance of maintaining competition inside the coalition. 

In a one-shot second-price auction where bidders have no constraint, as is revealed, the best collusive strategy is to select a representative to bid truthfully and let the other bidders inside the coalition bid zero. 
Such a strategy decreases the second-highest bid, alleviates the inner competition, and notably, improves the representative's utility without making others worse~\citep{graham1987collusive}. 
When it comes to the multi-round scenario with budget constraints, it is no longer optimal for bidders to bid truthfully. In fact, they tend to underbid to control their expenditures, making sure not to exceed their budgets~\citep{aggarwal2019autobidding}. However, we show that the idea of selecting representatives still works. 

\begin{definition}
	A coordinated bidding strategy is a \textit{bid rotation}, if it lets at most one bidder bid non-zero in each round.
\end{definition}

\Cref{lemma: bid rotation strategy} indicates that there always exists a best coordinated bidding algorithm that is a bid rotation. Therefore, it suffices to consider bid rotation strategies only. 
\begin{restatable}{lemma}{lemmabidrotationstrategy}
\label{lemma: bid rotation strategy}
    For any feasible coordinated strategy $\mathrm{C}$, there exists a feasible bid rotation strategy $\mathrm{C}^{\rm{P}}$ such that every member's expected utility under algorithm $\mathrm{C}^{\rm{P}}$ is at least as good as that under algorithm $\mathrm{C}$.
\end{restatable}

Hereafter, we divide the algorithm design into two parts: designing the budget control strategies in the real auction and the selection rule inside the coalition. 
We use adaptive pacing for the budget management part, and the main difficulty lies in designing the rule for selecting the representative bidder each round. 

One 
%simple and 
natural idea is to choose the member with the highest bid to compete outside the coalition. 
Such an idea motivates coordinated pacing (CP,
Algorithm \ref{alg: coordinated pacing}). 
In CP, \Cref{alg:CP:select} shows the selection step inside the coalition: the one with the highest bid becomes the winner. 
In \Cref{alg:CP:bid}, the winner posts the same bid as in the coalition, while others post zero in the real auction. Finally, the pacing parameters are updated with the adaptive pacing technique. We note that the step size can be set in a more general way, with details in the appendix. Here, we use $1/\sqrt{T}$ to simplify the description.

\begin{algorithm}[tb]
   \caption{Coordinated Pacing (CP)}
   \label{alg: coordinated pacing}
\begin{algorithmic}[1]
\setcounter{ALC@unique}{0}
   \STATE {\bfseries Input:} $\epsilon = 1/\sqrt{T}$, $\bar{\xi}_k\geq\bar{v}_k/\rho_k$ for all $k\in\calK$. 
   \STATE Select an initial multiplier $\bxi_0\in [0,\bar{\xi}_k]^K$, and set the remaining budget of agent $k$ to $\wt{B}_{k,1}=B_k=\rho_kT$.
   \FOR{$t=1$ {\bfseries to} $T$}
   \STATE Observe the realization of $\bv_{t}$. 
   \STATE Select $k^*\in \argmax_{k \in \calK}\min\{v_{k,t} / (1+\xi_{k,t}), \tilde{B}_{k,t}\}$ (break ties arbitrarily). \label{alg:CP:select}
   \STATE For each $k$, post a bid \label{alg:CP:bid}
   \STATE ~~~~$b_{k, t} = \min \left\{\bm{1}\left\{k=k^*\right\} {v_{k,t}} / {(1+\xi_{k,t})}, \tilde{B}_{k,t}\right\}$.
   \STATE Observe the the expenditures $\zz_{t}$.
   \STATE For each $k$, update the multiplier by 
   \STATE ~~~~$\xi_{k,t+1} = P_{[0,\bar{\xi}_k]}\left(\xi_{k,t}-\epsilon(\rho_k-z_{k,t})\right)$.
   \STATE For each $k$, update the remaining budget by 
   \STATE ~~~~$\wt{B}_{k,t+1} = \wt{B}_{k,t}-z_{k,t}$. 
   \ENDFOR
\end{algorithmic}
\end{algorithm}

Now we give our theoretical guarantees for the performance of CP. 
First, we make assumptions on the function of bidders' expected expenditures per round. 
In the case of CP, denote the expected expenditure of bidder $k$ under CP
to be $\EXPcostCP$, which is
\[\EXPcostCP(\bxi)\coloneqq \bbE_{\bv,d^O}\left[\bm{1}\left\{v_k\geq (1+\xi_k)d_k\right\}d^O\right].\]
\begin{assumption}
\label{asm: cp}
 $\bm{G}^\rmCP$ is strongly monotone in $[0,\bar{\xi}_k]^K$.
\end{assumption}

Denote the expected utility of bidder $k$ under CP 
by $\Pi^{\rmCP}_k$. In the symmetric case, we show that the performance of CP
is better than the benchmark $\EXPutility(\blambda^*)$ asymptotically:
\begin{restatable}{theorem}{thmsymmetricdominanceofcp}
\label{thm: symmetric dominance of cp}
 Suppose that \Cref{asm: cp} holds. 
When every bidder has the same value distribution and budget, we have
\begin{align}
\label{eqn: symmetric dominance of cp}
    \liminf_{T\rightarrow\infty}\frac{1}{T}(\utilityCP - \EXPutility(\blambda^*))> 0, ~~\forall{k \in \calK}.
\end{align}
\end{restatable}

The proof of the theorem is given in \Cref{app: CP}. To prove Theorem \ref{thm: symmetric dominance of cp}, we first show that $\bxi$ converges to some equilibrium $\bxi^*$ and, therefore, 
the algorithm's performance converges to the corresponding expected utilities. 
The problem then reduces to the comparison of equilibrium utilities. Concretely, we show that $\bxi^*\leq \blambda^*$. This means that as an entirety, bidders in coalition $\calK$ are more competitive than when they independently bid, and therefore the coalition wins more items. However, the competition inside the coalition also becomes fiercer. 

Despite the good performance in the symmetric case, CP does not outperform the benchmark in all instances. In some cases, strong bidders may become stronger, and weak bidders become weaker. 
To see this, we consider a special case with two bidders forming a coalition. Bidder $1$'s value is always higher than the bidder $2$'s, and the outside competing bid $d^{O}$ is always $0$. For IP, which converges to our benchmark, bidder $1$'s expected expenditure per round is not zero. Therefore her shading parameter is non-zero, so long as her budget is sufficiently small. There is a chance that her bid is lower than bidder $2$'s, enabling bidder $2$ to compete and win the item. 
In CP, however,  bidder $1$ keeps being selected as the representative, and bidder $2$'s expected utility will be lower than hers in the benchmark. A detailed experimental example is given in \Cref{sec: experiments}. 

To avoid such cases and outperform the benchmark in all instances, one way is to independently maintain competition inside the coalition, giving each member a fair chance to become the representative. The ideas of bid rotation and maintaining the competition are important in designing HP, 
which we will show next.

\section{Hybrid Coordinated Pacing}
\label{sec: hp}

Inspired by CP (Algorithm \ref{alg: coordinated pacing}), we further propose a hybrid coordinated pacing (HP) algorithm shown 
 in Algorithm \ref{alg: hybrid coordinated pacing}.
 HP ensures that every coalition member can achieve a higher utility than the benchmark.

\begin{algorithm}[tb]
   \caption{Hybrid Coordinated Pacing (HP)}
   \label{alg: hybrid coordinated pacing}
\begin{algorithmic}[1]
\setcounter{ALC@unique}{0}
   \STATE {\bfseries Input:} $\epsilon = 1/\sqrt{T}$, $\bar{\mu}_k\geq \bar{v}_k/\rho_k$ for all $k\in \calK$. 
   \STATE Select an initial multiplier $\bmu_1 \in [0, \bar{\mu}_k]^K$, let the initial pseudo multiplier be $\blambda_1 = \bmu_1$, and set the remaining budget of agent $k$ to $\wt{B}_{k,1}=B_k=\rho_kT$.
   \FOR{$T=1$ {\bfseries to} $t$}
   \STATE Observe the realization of $\bv_{t}$.
   \STATE Let $b^{I}_{k, t} = \min \left\{v_{k, t}/(1 + \lambda_{k, t}), \tilde{B}_{k,t}\right\}$ and $d^{I}_{k, t} = \max_{i\in \calK: i\neq k} b^{I}_{i, t}$ for each $k$.
   \STATE Select $k^* \in \arg \max_k b^{I}_{k, t}$ (breaking ties arbitrarily). \label{line: internal election}
   \STATE For each $k$, post a bid
   \STATE ~~~~$b^{O}_{k, t} = \min \left\{\bm{1}\left\{k = k^*\right\} v_{k,t}/(1+\mu_{k, t}), \tilde{B}_{k,t}\right\}$.
   \STATE Observe the allocations $\bx_{t}$ and the expenditures $\zz_{t}$.
   \STATE For each $k$, update the pseudo multiplier by 
   \STATE ~~~~$\lambda_{k,t+1} = P_{[0,\bar{\mu}_k]}\left(\lambda_{k,t} - \epsilon\left(\rho_k - z'_{k, t}\right)\right), $ \label{line: lambda update}
   \STATE where $z'_{k, t} = \bm{1}\left\{b^{I}_{k, t} \geq z_{k, t} \right\}\max \left(z_{k,t}, x_{k,t}d^{I}_{k, t}\right).$ 
   \STATE For each $k$, update the multiplier by 
   \STATE ~~~~$\mu_{k,t+1} = P_{[0,\lambda_{k,t+1}]}\left(\mu_{k,t} - \epsilon\left(\rho_k-z_{k,t}\right)\right).$ \label{line: mu smaller than lambda}
   \STATE For each $k$, update the remaining budget by 
   \STATE ~~~~$\wt{B}_{k,t+1} = \wt{B}_{k,t}-z_{k,t}$.
   \ENDFOR
\end{algorithmic}
\end{algorithm}

More specifically, HP draws a clear line between the internal election and external competition. 
In addition to the multiplier $\mu_k$ used to compete with bidders outside the coalition, 
HP maintains for each bidder $k$ a pseudo multiplier $\lambda_k$ for the internal election. 
The coalition first holds an internal election to select a representative and 
recommends her 
to compete with other bidders outside the coalition.
When the value is $v_k$ (and the budget is still sufficient), bidder $k$ bids $v_k/(1+\lambda_k)$ in the internal election, and bids $v_k/(1+\mu_k)$ in the real auction if she wins. 
To make a good distinction, we label the two ranking bids as $b^I_{k}$ and $b^O_{k}$ with superscripts. 

By updating $\lambda_k$ and $\mu_k$ separately, we are able to make sure that every member gets at least the same opportunities to bid outside the coalition and win the item, as they do in IP. To achieve this, we let the inner dual variables update exactly as the dual variables in IP. Therefore, in the updating rule of $\mu_k$, we require that the inner variable always be greater than the outer variable, i.e., $\mu_{k,t}\leq \lambda_{k,t}$, so that if an agent wins at a round in IP, she will also win at the same round in HP. In this way, HP can observe every necessary expenditure $z_{t}$ that will be used to update the dual variables in IP.

We denote by $\EXPcostHP(\mu_k, \blambda)$ bidder $k$'s expected expenditure per round when all coalition members shade bids according to $\bmu$ and $\blambda$ (ignoring their budget constraints):
\begin{equation}
    \begin{aligned}
    \label{eqn: hp exp cost definition}
        \EXPcostHP(\mu_k, \blambda) =\bbE_{\bv, d^O}\left[\bm{1}\left\{\frac{v_k}{1 + \lambda_k} \geq d^I_k \land \frac{v_k}{1 + \mu_k} \geq d^O \right\}d^O\right], 
    \end{aligned} 
\end{equation}
and we denote by $\EXPutilityHP(\mu_k, \blambda)$ her corresponding expected utility in this case:
% \small
\begin{equation*}
    \begin{aligned}
    % \label{eqn: hp exp utility definition}
        \EXPutilityHP(\mu_k, \blambda) = T\cdot \bbE_{\bv, d^O}\left[\bm{1}\left\{\frac{v_k}{1 + \lambda_k} \geq d^I_k \land \frac{v_k}{1 + \mu_k} \geq d^O \right\}(v_k - d^O)\right].
    \end{aligned}
\end{equation*}

We now make the following assumptions.

\begin{assumption}\label{asm: hp}
\begin{enumerate}
    \item There exists a constant $\gamma>0$ such that $\bm G = (G_k)_{k = 1}^K$ is a $\gamma$-strongly monotone function over $\prod_{k=1}^K[0,\bar{\mu}_k]$.
    \item There exists $\underline{G}'>0$ such that $\partial \EXPcostHP/\partial \mu_k < - \underline{G}'$ over $\prod_{k=1}^{K+1}[0,\bar{\mu}_k]$ for every bidder $k\in \calK$.
\end{enumerate}
\end{assumption}

\Cref{asm: hp} regularizes the behavior of $\EXPcostHP$ with two assumptions on $\blambda$ and $\mu_k$, respectively.
The first one directly follows from \Cref{asm: strongly-monotone}, which guarantees the convergence of $\blambda$.
Interestingly, we only require $\bm{G}$ instead of $\bm{G}^{\rm HP}$ to be $\gamma$-strongly monotone 
% function
with respect to $\blambda$, which differs from that \Cref{thm: symmetric dominance of cp} requires the strong monotonicity of $\bm{G}^\rmCP$. 
The second one assures the strict monotonicity of $\EXPcostHP$ and $\EXPutilityHP$ with respect to $\mu_k$. Intuitively, for a fixed $\blambda$, if both $\EXPcostHP$ and $\EXPutilityHP$ strictly decrease in $\mu_k$, bidder $k$ tends to decrease $\mu_k$ as long as her budget constraint is not binding, so as to maximize her utility. Therefore, there exists a unique equilibrium $\mu^*_k$ such that either $\mu^*_k = 0$ or any $\mu_k$ lower than $\mu_k^*$ would break the constraint, and we will show that the dynamically updated sequence of $\mu_{k,t}$ converges to such an equilibrium under HP.
We further provide a sufficient condition for the second assumption to hold in \Cref{lem: lipschitz hp}.

\begin{restatable}{theorem}{thmdominanceofhp}
\label{thm: dominance of hp}
Suppose that \Cref{asm: hp} holds. We have
\begin{align}
\label{eqn: dominance of hp}
    \liminf_{T\rightarrow\infty}\frac{1}{T}(\utilityHP - \EXPutility(\blambda^*)) \geq 0, \forall{k \in \calK},
\end{align}
and the equality holds for at most one bidder.
\end{restatable}

\Cref{thm: dominance of hp} establishes that with HP,
 all but one bidder in the coalition
 % a coalition of size $K$
 can get higher utilities, compared to  $\EXPutility(\blambda^*)$ in an asymptotic sense. Unlike CP,
 this performance guarantee holds in a general setting where bidders may have different budget constraints or value distributions. It is worth noticing that, both HP, CP and IP have the same convergence rates. Therefore, HP and CP are also comparable to IP, taking the convergence rates into consideration. 

To prove \Cref{thm: dominance of hp}, we first show the convergence of both the multipliers $\bmu$ and the pseudo multipliers $\blambda$. Observe that the update of $\lambda_{k}$ is equivalent to 
\begin{equation}
\label{eqn: lambda update equivalent}
    \lambda_{k, t+1} = P_{[0,\bar{\lambda}_k]}\left(\lambda_{k, t} - \epsilon\left(\rho_k - \bm{1}\left\{b^{I}_{k,t} \geq d_{k, t}\right\}d_{k, t}\right)\right),
\end{equation}
which coincides with the subgradient descent scheme of an individual adaptive pacing strategy. Note that the Euclidean projection operator in Line \ref{line: mu smaller than lambda} assures that for all $t$, $$\mu_{k,t} \leq \lambda_{k, t} \Longrightarrow \frac{v_{k, t}}{1 + \lambda_{k, t}} \leq \frac{v_{k,t}}{1+\mu_{k, t}}.$$ Thus by bidding $v_{k,t}/(1+\mu_{k, t})$, we can always discern $\bm{1}\left\{v_{k, t}/(1 + \lambda_{k, t}) \geq d^O_{t} \right\}$. 
This provides enough information to let the sequence of pseudo multipliers exactly behave as if bidders are independently pacing, and converge to the same $\blambda^*$ defined by \eqref{eqn: lambda star definition}.

At the same time, we show that the sequence of multipliers $\bmu_t$ converges to a unique $\bmu^*$, defined by the complementarity conditions,
\begin{equation}
\label{eqn: mu star definition}
\mu^*_k\geq 0 \perp \EXPcostHP(\mu^*_k, \blambda^*)\leq \rho_k, \forall k\in \calK.
\end{equation}

Consequently, we show that the expected utility of bidder $k$ under HP
converges to $\EXPutilityHP(\mu^*_k, \blambda^*)$. Finally, the proof of \Cref{thm: dominance of hp} is concluded by a comparison between $\EXPutilityHP(\mu^*_k, \blambda^*)$ and $\EXPutility(\blambda^*)$. 

One may consider whether the inequality can strictly hold for everyone as in \Cref{thm: symmetric dominance of cp}. Here we provide a concrete example in which a bidder gains no utility boost.

\begin{example}
\label{eg: at most one}
We consider two bidders with no budget constraints. 
Bidder $1$'s value $v_1$ is uniformly distributed over $[0, 2]$, while bidder $2$' value $v_2$ is uniformly distributed over $[0, 1]\cup[4, 5]$. 
The highest competing bid of other bidders $d^O$ follows a uniform distribution over $[1, 5]$. One can verify that this instance satisfies \Cref{asm: hp}.

When bidding independently, the optimal strategy for each bidder is to bid truthfully with respect to second-price auctions, i.e., $\lambda_1 = \lambda_2 = 0$.
Then bidder 1 wins an auction only if $v_2\in [0, 1]$ and $v_1\geq d^O \geq 1$, which means her winning payment is always equal to $d^O$. 
Therefore, bidder 1 cannot benefit from forming a coalition with bidder 2. 
\end{example}

\section{Other Properties of Proposed Algorithms}
\label{sec: discussion}

In this section, we discuss the properties of CP and HP in the symmetric case. Recall that by showing the performance of the algorithms converges, we are able to study their utility properties by analyzing the equilibria they converge to. 
In what follows, we will assume the corresponding assumptions hold (i.e., \Cref{asm: strongly-monotone,asm: cp,asm: hp}), and all the results hold the in an asymptotic sense. 
We note that since these equilibria are defined as solutions to non-linear complementarity problems (NCP), and do not have closed forms, they can have various possibilities and are hard to analyze. Therefore, we leave the exploitation of their properties in asymmetric cases as an open question. 

\subsection{Coalition Welfare Maximization}   
We first show that CP and HP
are the best coordinated algorithms in the sense of coalition welfare maximization and Pareto optimality.
Specifically, we consider the hindsight maximal coalition welfare, given a realization of $(\bv;\dd^O)$, the values and highest bids outside the coalition:
 \begin{equation*}
 \begin{aligned}
 \pi^{\mathrm{H}}\left(\bv; \dd^O\right) \coloneqq \max_{\bx_{k} \in\{0,1\}^{K\times T}} & \sum_{t=1}^{T} \sum_{k\in \calK} x_{k, t}(v_{k, t} - d^O_{t}), \\
 \text { s.t. } & \sum_{t=1}^{T} x_{k, t} d^O_{ t} \leq T \rho_{k}, \forall k \in \calK.
 \end{aligned}
 \end{equation*}
We further say a bidding algorithm $\rmC$ is Pareto optimal in the symmetric case, if there does not exist another coordinated bidding algorithm, such that in some symmetric instance, every bidder's expected utility is no worse than that under $\rmC$, and at least one bidder's expected utility is strictly better.  
\begin{restatable}{theorem}{lmmparetooptimalityofhp}
\label{lmm: pareto optimality of hp}
When every bidder in the coalition has the same value distribution and budget, CP  and  HP maximize the total expected utility of the coalition asymptotically,
\[
\liminf_{T\rightarrow\infty}\frac{1}{T}(\sum_{k\in\calK}\Pi^{\rmCP/\rmHP} - \bbE[\pi^{\rm{H}}])\geq 0. 
\]
Therefore, they are Pareto optimal. 
\end{restatable}

A key of the proof is that the sum of expected utilities CP and HP converge to, $\sum_{k\in\calK}\EXPutilityCP(\bxi^*)$ and $\sum_{k\in\calK}\EXPutilityHP(\mu_k^*, \blambda^*)$, have the same form as an upper bound of $\bbE[\pi^{\rm{H}}]$ in the symmetric case. 
% Therefore, t
The simplicity of CP allows us to use it as a surrogate to prove the properties of HP. 

\subsection{Truthfulness on Misreporting Budgets} 
In this section, we will study how a bidder's obtained value and utility at equilibrium change if she misreports her budget, with others truthfully behaving. We assume the values are publicly known and only consider the possibility that a bidder may under-report her budget, as over-reporting may let her total cost exceed her budget. 

We directly consider $\bm{G}$ as a function of $\brho$: $\bm{G}(\brho)=\bm{G}(\blambda^*)$ where $\blambda^*$ is defined by \eqref{eqn: lambda star definition} with respect to $\brho$. $\bm{G}^{\rmCP}(\brho), \bm{U}^{\rmCP}(\brho)$ and $\bm{V}(\brho)$ are defined similarly, where 
\begin{align*}
V_k(\brho) =V_k(\blambda^*)=T\cdot
\bbE\left[\bm{1}\left\{v_k\geq d_k\right\}v_k\right]
\end{align*}
is bidder $k$'s expected obtained value under $\blambda^*$ in IP.  
\Cref{thm:misreport-budget-value-cp} shows that CP is incentive compatible. 
\begin{restatable}{theorem}{thmmisreportbudgetvaluecp}
	\label{thm:misreport-budget-value-cp}
	Under CP, bidder $k$'s expected utility $U^{\rmCP}_{k}$ will not increase,  if she under-reports her budget and others report budgets truthfully. 
\end{restatable}

Contrary to \Cref{lmm: pareto optimality of hp}, CP and HP have different forms when bidders deviate. Recall that IP can be regarded as a special coordinated bidding algorithm, and the pseudo parameters of HP converge to the same equilibrium as IP. Therefore, we study IP first. 

\begin{restatable}{theorem}{thmmisreportbudgetvalueip}
	\label{thm:misreport-budget-value-ip}
	Under IP, bidder $k$'s expected obtained value $V_{k}$  will not increase if she under-reports her budget and others report budgets truthfully. 
\end{restatable}

To prove \Cref{thm:misreport-budget-value-cp} and \ref{thm:misreport-budget-value-ip}, we compare the equilibrium the algorithm converges to when bidder $k$ misreports with that when she truthfully behaves. 
Note that the equilibria the algorithms converge to are defined by NCPs, which are generally difficult to study. We analyze all 8 possible solutions of the corresponding NCPs. By applying the definition of NCP and strong monotonicity, we show that bidder $k$'s average expected obtained value at the equilibrium does not increase by misreporting.

\Cref{thm:misreport-budget-value-ip} shows that IP has good incentive compatibility on obtained values. However, when it comes to utilities, there are chances that the decrease in cost outweighs that in value, and a bidder's utility may overall increase. An experimental example is shown in the appendix. 

Moreover, as the analysis of IP is complicated enough, we make the incentive problem of HP an open problem. 
Compared with the intricacy of HP, 
CP has relative succinct structure and is easy to analyze.

\section{Experiments}
\label{sec: experiments}

It is known that the expected performance of IP
converges to the benchmark equilibrium \citep{balseiro2019learning}.
In this section, we run IP together with CP and HP
to demonstrate their performance. 
We first use an example to provide useful insights. 
Afterward, we conduct experiments on real-world  and synthetic data, as a supplement to our theoretical results.

\subsection{An Asymmetric Example}
\label{sec: asymmetric example}
Consider the example where there are only two bidders in the coalition.
The valuation of bidder $1$ is sampled from the uniform distribution in $[0,1]$, $\calU[0, 1]$, with probability $p$, and from $\calU[1, 1 + \eta]$ otherwise. 
$p, \eta \in (0, 1)$ are two small constants.
The valuation of bidder $2$ is sampled from $U[0, 1]$.
We consider the extreme case where $d_t^O$ is always zero. 
As a result, the representative bidder in CP or HP has zero expenditure. 
For both CP and HP, 
We set the initial multipliers and their upper bounds as $0$ and $\bar{\mu}_i = 3.0$, respectively.
We also set $\rho_i = \bar{v}_i / \bar{\mu}_i$ for each bidder.
The strong monotonicity assumption can be validated experimentally by computing the minimal marginal monotone parameter within small grids. 
When the grid width is $0.003$, the value is $0.035$.
We use a step size $\epsilon = 0.1T^{-0.5}$. 

We demonstrate the empirical simulation of the example in \Cref{fig:counterexample}. 
The experiment is run $1,000$ times, each time with $20,000$ rounds. The average performance is presented.
We can see from the figure that by using CP, bidder $2$ receives lower utility than using IP.
This result indicates that CP does not outperform the benchmark in this instance. 
Intuitively, in the extreme case where $p\rightarrow 0$, it always holds that $v_{1,t}>v_{2,t}$. 
Since bidder $1$'s expenditure per round is always zero, she can set her parameter as $0$. 
Therefore, bidder $1$ always wins the inner election, leaving bidder $2$ no chance to win any item.
In comparison, bidder $2$ is more likely to win in IP and HP.
The reason is that in IP, bidder $1$ has a positive payment, so she has to lower her bid if her budget is small. This enables bidder $2$ to win when bidder $1$'s bid is low enough.
In HP, the opportunities for competing outside the coalition are allocated the same way as in IP. Therefore, both bidders get chances to win, but their expenditures are lower than those under IP. 

\begin{figure}
    \centering
    \includegraphics[width=0.97\textwidth]{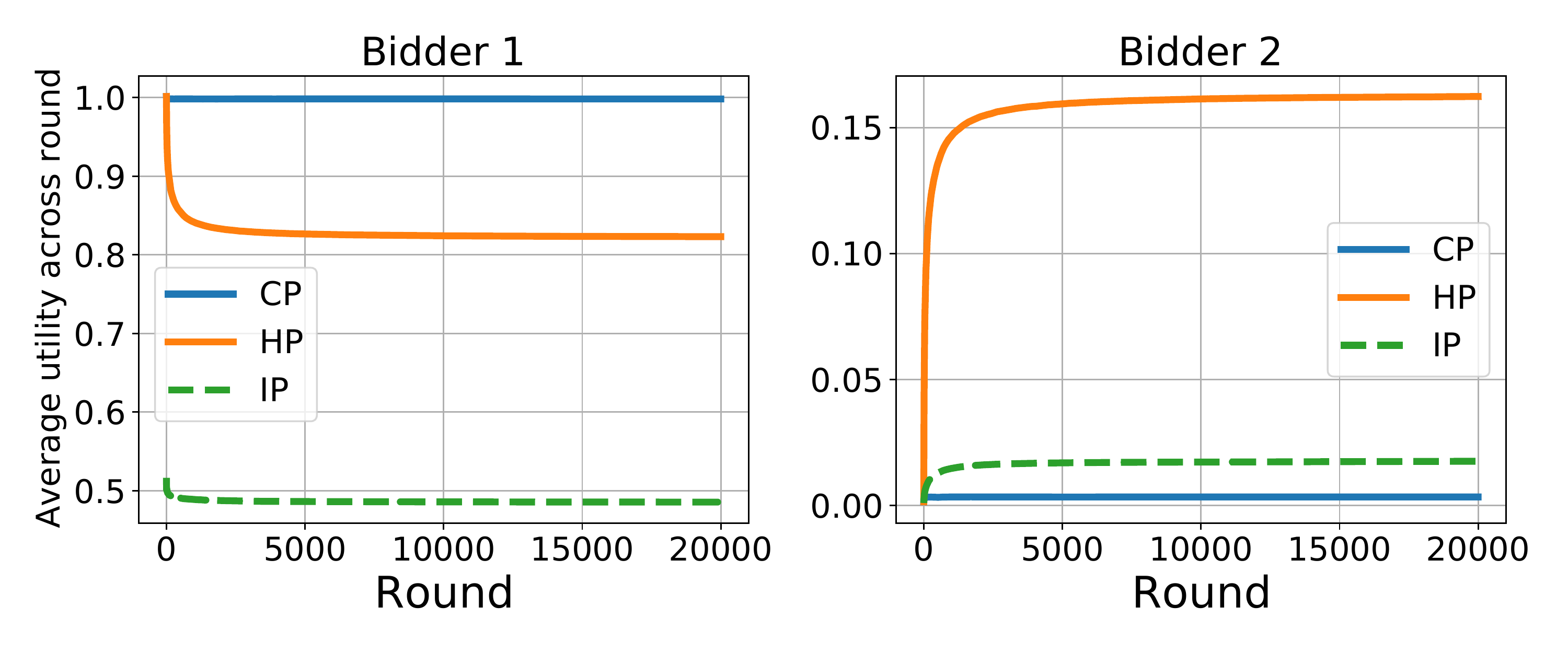}
    \caption{The utility convergence results of an asymmetric example.
    The y-axes are the average utilities across rounds ($\Pi_k/T$).
    CP does not outperform IP concerning bidder $2$'s utility.
    }
    \label{fig:counterexample}
\end{figure}

\subsection{Experiments on Real Data and Synthetic Data}

We present experiments on real-world data to evaluate the performance of our proposed algorithms, and justify its robustness through experiments on various synthetic data.

\begin{figure*}[htbp]
	\begin{subfigure}{0.498\textwidth}
	\centering
	\includegraphics[width = \linewidth]{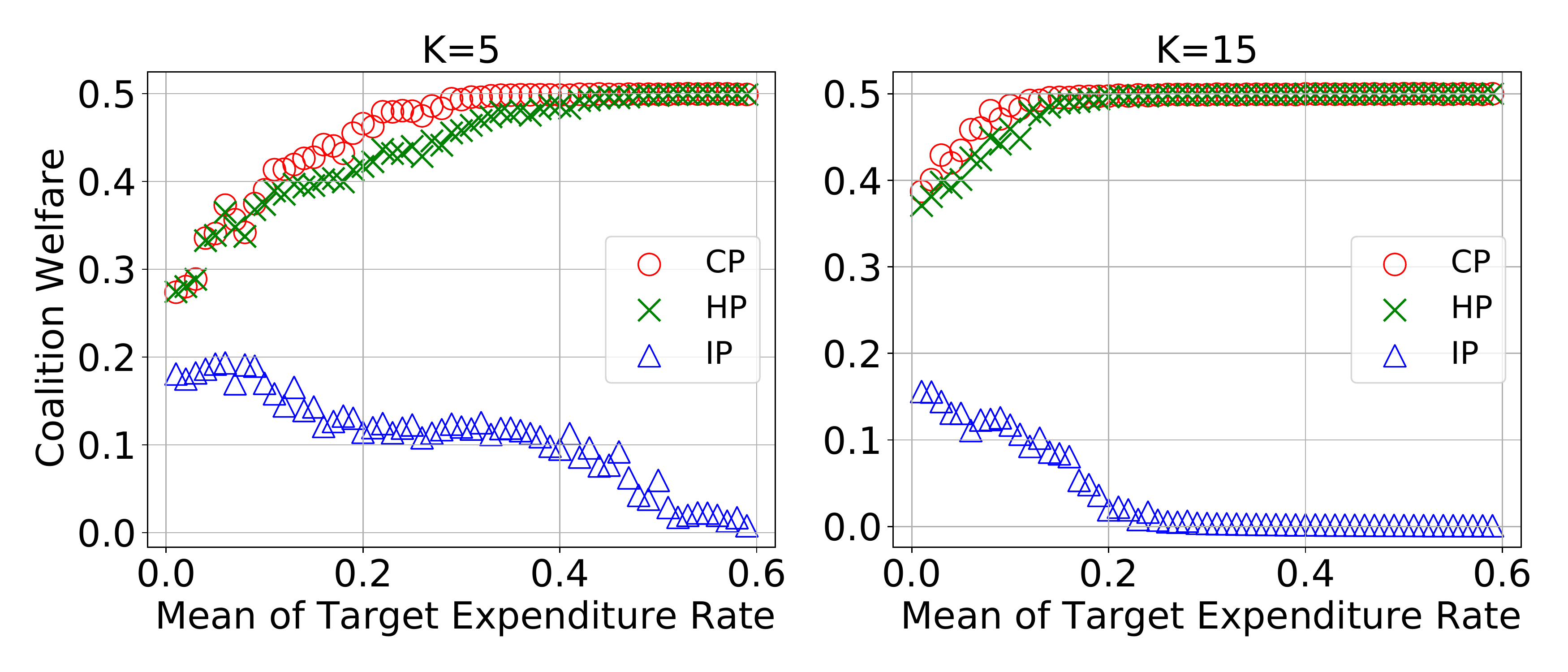}
	\caption{Coalition welfare in real-data experiments.}
	\label{fig: coalition welfare-ipinyou}
	\end{subfigure}
	\begin{subfigure}{0.498\textwidth}
	\centering
	\includegraphics[width = \linewidth]{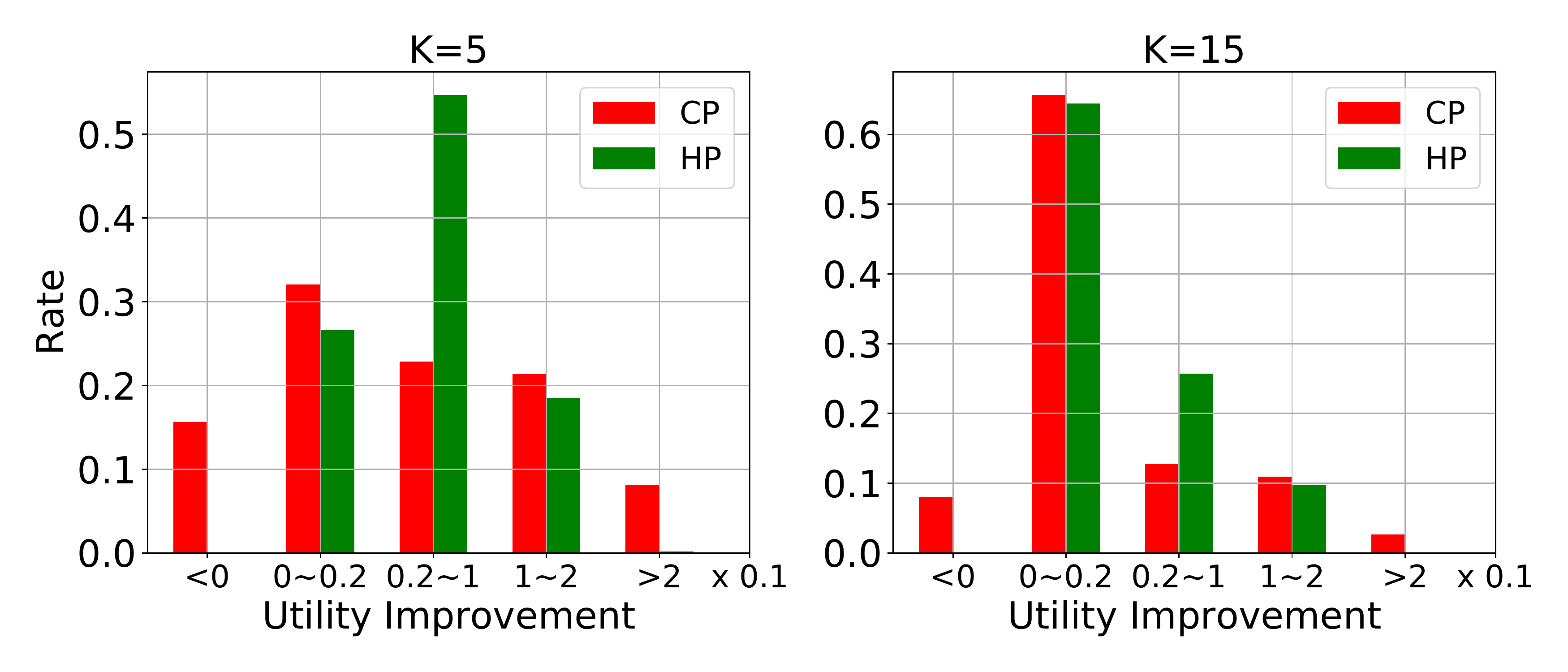}
	\caption{Utility improvements in real-data experiments. }
	\label{fig: utility improvement-ipinyou}
	\end{subfigure}
		\begin{subfigure}{0.498\textwidth}
	\centering
	\includegraphics[width = \linewidth]{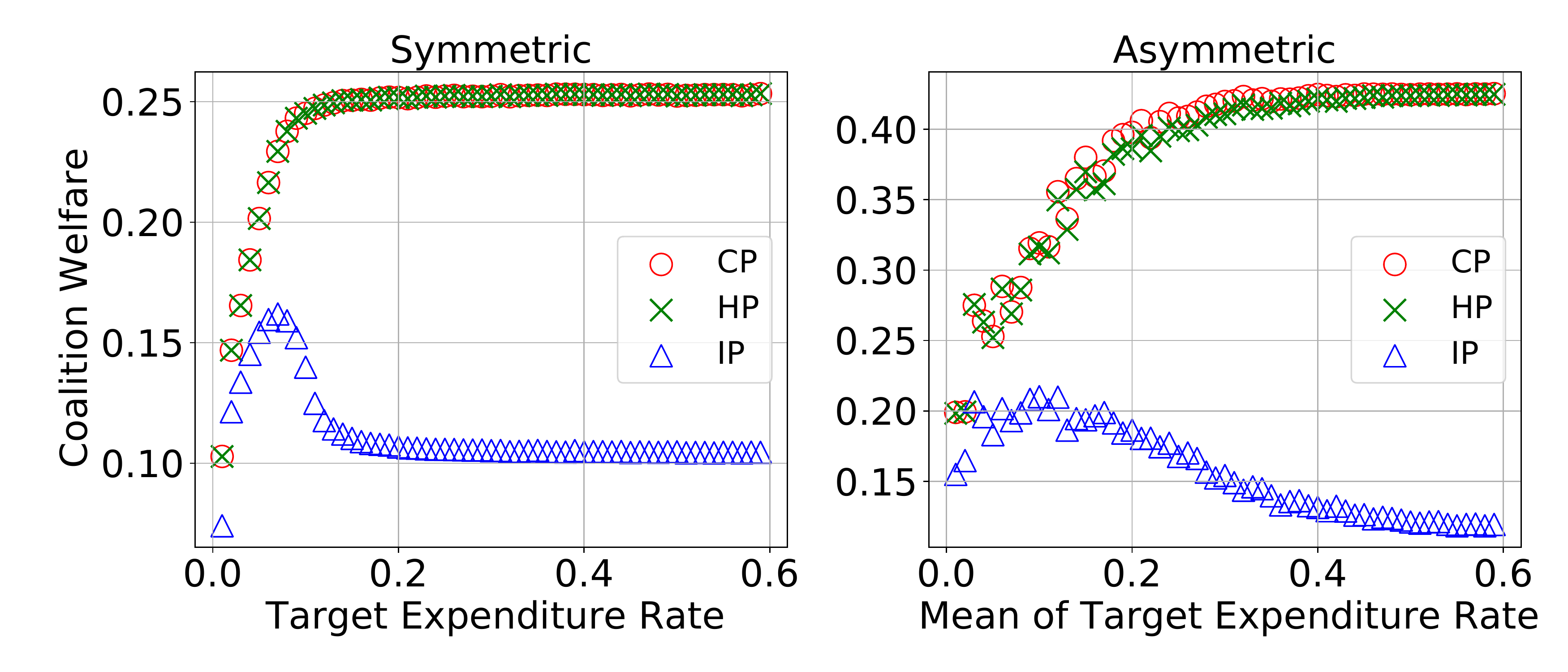}
	\caption{Coalition welfare in simulation experiments. }
	\label{fig: coalition welfare-simulation}
	\end{subfigure}
		\begin{subfigure}{0.498\textwidth}
	\centering
	\includegraphics[width = \linewidth]{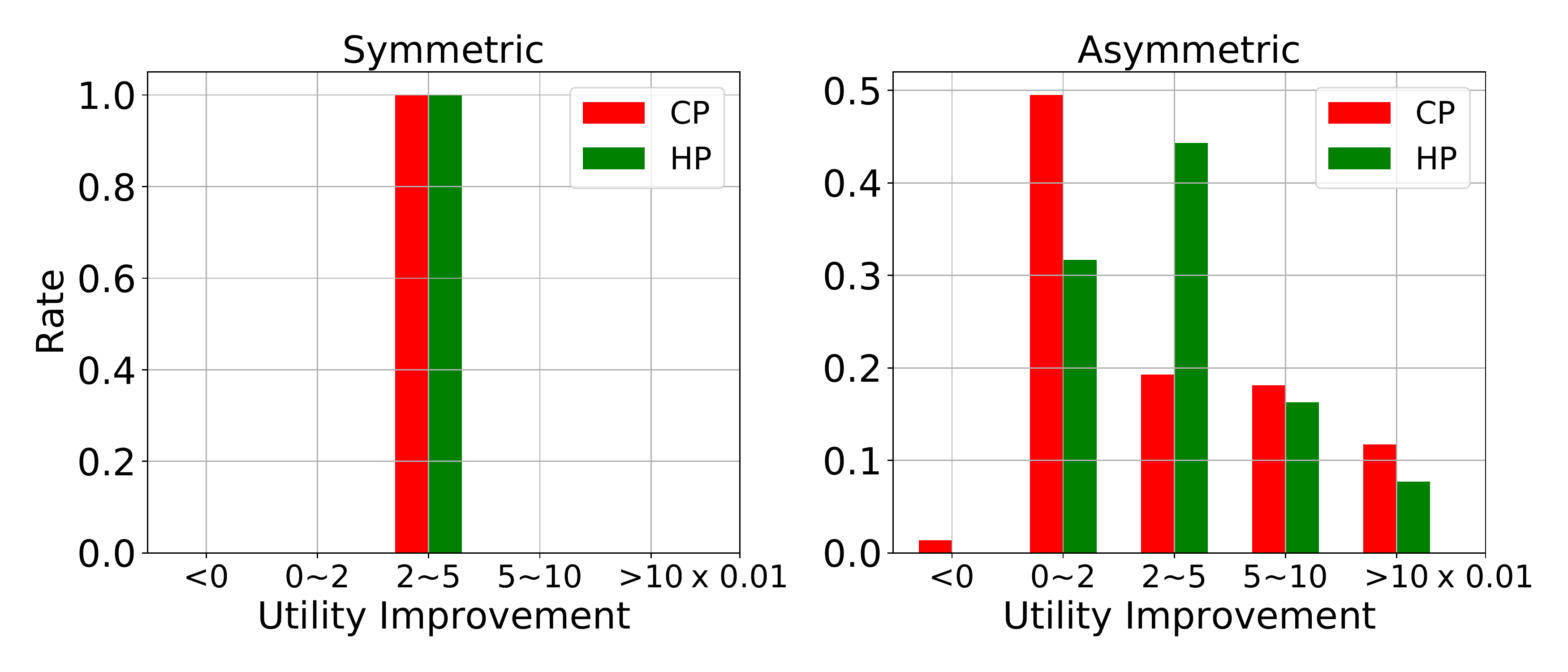}
	\caption{Utility improvements in simulation experiments. }
	\label{fig: utility improvement-simulation}
	\end{subfigure}
	\caption{Performance of coordinated pacing (CP) and hybrid coordinated pacing (HP) against individual pacing (IP) in real-data and simulated data: \Cref{fig: coalition welfare-ipinyou} and \ref{fig: coalition welfare-simulation} show the changes of coalition welfare ($\sum_{k\in\calK}\Pi_k/T$) with target expenditure ratio varies. \Cref{fig: coalition welfare-ipinyou} shows the result on iPinYou real data and \Cref{fig: coalition welfare-simulation} on the simulated data. Each point is an average of $20$ repeated experiments.
	\Cref{fig: utility improvement-ipinyou} and \ref{fig: utility improvement-simulation} the utility increase of coalition members increase after transferring to CP/HP from IP, captured by ($\Pi^{\rmCP/\rmHP}_k/T - \Pi^\rmIP_k/T$) when the target expenditure rate ranges from 0.1 to 0.3. \Cref{fig: utility improvement-ipinyou} shows the results on real data, while \Cref{fig: utility improvement-simulation} on simulated data. }
        \label{fig: experimental results}
\end{figure*}

\paragraph{Experiments Set-Up.}

The real-world data we use, the iPinYou dataset
\footnote{
\url{https://contest.ipinyou.com}
}
~\citep{liao2014ipinyou, zhang2014real}, was provided for a real-time bidding competition organized by iPinYou.
We use the testing data of Season $2$, which contains $2,521,630$ records from 
%the perspective of 
$5$ different bidders. 
Each piece of data contains the information of the user, the ad slot, and the bid, from which we mainly use the bidding (winning) price, paying price, and advertiser ID.

We first linearly scale all the bidding prices into $[0, 1]$.
To make up $K$ coordinated bidders, we randomly divide the historical bidding records of each of the 5 raw bidders into $K/5$ parts. 
Subsequently, we obtain $K$ bidding records.
We sample the valuation of each coalition member by drawing with replacement from her corresponding bidding record.
Besides, we sample the highest bid outside the coalition, $d_t^O$, from $\calN(0.5, 0.2)$.
As for the budget, we first select $\bar{\rho}$, the mean of the target expenditure ratio, to be $0.01\alpha$ with $\alpha \in \{1,\dots,100\}$. 
Then we sample each coalition member's target expenditure ratio i.i.d. from $\calN(\bar\rho, 0.1)$.
We set the initial multipliers and their upper bounds as $1.0$ and $20.0$, respectively.
We take step size $\epsilon = 0.1T^{-0.5}$ and 
%we set 
$T=20,000$.
At last, we let $K \in \{5, 15\}$, and run $20$ experiments for each setting with different random seeds.

We also manually design different distributions to verify the superiority of our algorithms. Due to page limits, we show only one typical example. We set $K=5$ and compare the coalition welfare and bidders' utility improvements in both symmetric and asymmetric cases, respectively. More results can be found in \Cref{sec: app-exp}. 
For the symmetric case, we sample each bidder's value i.i.d. from 
$\calN(0.5, 0.2)$ and let their common target expenditure ratio be $0.01\alpha$ for $\alpha \in \{1, \cdots, 100\}$. 
For the asymmetric case, we sample each bidder's value from five different distributions, and each member's target expenditure ratio is selected the same as in real data experiments. For both cases, the highest bid outside the coalition is also sampled i.i.d. from $\calN(0.5, 0.2)$. We note that the experiments of different distributions of the highest bid outside the coalition were conducted. The results are shown in Appendix \ref{sec: app-exp}.

\paragraph{Experimental Results.}

\Cref{fig: coalition welfare-ipinyou} and \ref{fig: utility improvement-ipinyou} show algorithms' performance in real data. 
We can see that 
Both CP and HP significantly improve coalition welfare compared with IP. The improvement increases with each member's budget.
Such a result attributes to the reduction of inner competition in CP and HP. 
When bidder's budgets are large, the inner competition becomes fiercer under independent bidding. That is also why in IP, the coalition welfare tends to be 0 as people's budgets increase.
Moreover, 
CP is slightly higher than HP in coalition welfare, whereas the method brings greater utility variance, and cannot guarantee that everyone gets better than IP. This coincides with our analysis in \Cref{sec: cp} that CP makes the entire coalition more competitive. 
A further observation is that the utility improvements of HP and CP decrease with the increase of $K$, the coalition size.  
The reason is that the number of rounds is fixed, and so does the number of items to be auctioned.
Therefore, each person's share decreases as more people join the coalition. 

The simulation results are shown in \Cref{fig: coalition welfare-simulation} and \ref{fig: utility improvement-simulation}. 
In the symmetric case, we see that CP and HP are identical in both coalition welfare and utility improvement. This is consistent with the concrete forms of the algorithms. 
The asymmetric case reflects similar characteristics with real data.  
That is, HP and CP both improve total welfare significantly. In HP, all the participants receive higher utilities than in IP. In CP, the utilities of the vast majority can be improved.
All the experimental results demonstrate the benefits of reducing the inner competition among members in the coalition.

Recall that the theoretic results predict that the averages of multipliers will converge. In the experiments, the variances of the last 100 rounds of the averages of multipliers are around $10^{-4}$, indicating the convergence of multipliers. Further details show that they converge to the point predicted by our theorems, which are shown in \Cref{sec: app-exp}. Therefore, while we did not verify the assumptions made in previous sections, the experimental results further confirm the robustness of our algorithms. 
\section{Conclusion}
\label{sec: conclusion}
In this paper, we study designing online coordinated bidding algorithms in repeated second-price auctions with budgets. We propose algorithms that guarantee every coalition member a better utility than the best they can get under independent bidding. We analyze the game-theoretic properties of the algorithms like coalition welfare maximization and bidders' incentives to misreport their budget constraints. Experiment results on both synthetic data and real-world data show the effectiveness and robustness of our proposed algorithms. 

There are quite a few open problems we can ask about. The first is whether we can relax the assumptions made in this paper, e.g., the strong monotonicity of bidders' expected expenditure functions. Second, in this paper, we assume that the highest bid outside the coalition is drawn from a stationary distribution. Another future work is to assume that agents outside the coalition also use IP to bid and study whether CP and HP can still improve every coalition member's utilities, compared to their independent pacing. As more bidders are taken into consideration, we can also ask about the fairness of the algorithms, before which an appropriate concept of fairness needs to be defined. One can also consider more complicated game-theoretic scenarios, e.g., taking the incentives of the bidding agency into account. 

\newpage
%\bibliographystyle{plainnat}
%\bibliography{Arxiv/ref}

%%%%%%%%%%%%%%%%%%%%%%%%%%%%%%%%%%%%%%%%%%%%%%%%%%%%%%%%%%%%%%%%%%%%%%%%%%%%%%%
%%%%%%%%%%%%%%%%%%%%%%%%%%%%%%%%%%%%%%%%%%%%%%%%%%%%%%%%%%%%%%%%%%%%%%%%%%%%%%%
% Appendix
\newpage
\appendix
% \onecolumn
\section{Individual Adaptive Pacing}
\label{sec: ip}

For the completeness of our work, we present adaptive pacing~\citep{balseiro2019learning}, an optimal individual bidding algorithm, in Algorithm \ref{alg: individual pacing}. We omit the subscript $k$ to simplify the notation. 

\begin{algorithm}[htbp!]
   \caption{Individual Adaptive Pacing (IP)}
   \label{alg: individual pacing}
\begin{algorithmic}[1]
\setcounter{ALC@unique}{0}
   \STATE {\bfseries Input:} $\epsilon = 1/\sqrt{T}$, $\bar{\lambda} \geq \bar{v}/\rho$.
   \STATE Select an initial multiplier $\lambda_1\in [0,\bar{\lambda}]$, and set the remaining budget of agent $k$ to $\wt{B}_{1}=B=\rho T$.
   \FOR{$t=1$ {\bfseries to} $T$}
   \STATE Observe the realization of $v_{t}$.
   \STATE Post a bid $b_{t} = \min \left\{\frac{v_{t}}{1+\lambda_{t}}, \wt{B}_{t}\right\}$. 
   \STATE Update the multiplier by $\lambda_{t+1} = P_{[0,\bar{\lambda}]}\left(\lambda_{t}-\epsilon(\rho-z_{t})\right)$.
   \STATE Update the remaining budget by $\wt{B}_{t+1} = \wt{B}_{t}-z_{t}$. 
   \ENDFOR
\end{algorithmic}
\end{algorithm}
\section{Missing Proofs in Section~\ref{sec: cp}}
\label{app: CP}

\subsection{Proof of \Cref{lemma: bid rotation strategy}}

\lemmabidrotationstrategy* 

\begin{proof}
Let the bidding profile produced by strategy $\mathrm{C}$ be $(\boldsymbol{b}_1, \boldsymbol{b}_2, \cdots, \boldsymbol{b}_T)$. Now randomly choose $k\in \mathop{\arg \max}\boldsymbol{b}_t$. Define new strategy $\mathrm{C}^{P}$ such that $b^{\mathrm{P}}_{k,t} = b_{k,t}$, and $b^{\rm{P}}_{k',t} = 0$ for all $k' \neq k$. 
This rule guarantees $\forall t \in [T]$, $\boldsymbol{b}^{\rm{P}}_t$ has at most one non-zero component, which means it is a bid rotation strategy. 

For each round $t$, denote $I_t = \mathop{\arg \max}\boldsymbol{b}_t$.  if $|I_t|=1$, the bidder with the highest bid produced by $\rm{C}^{P}$ remains the same with that of $\rm{C}$. The second-highest bid under $\rm{C}^{P}$ is no more than that under $\rm{C}$. Therefore a coalition member will win the same item and gain the same value as in  $\rm{C}$, but pay no more than in $\rm{C}$, leading to more utility. If $|I_t| >1$, each bidder will pay at least the price as she does in $\rm{C}$. This means $\rm{C}^{P}$ is feasible. Since the case $|I_t|>1$ has measure zero, every member's expected utility under $\rm{C}^{P}$ is at least as good as that under $\rm{C}$. 
\end{proof}

\subsection{Details on the Choice of Step Sizes}
In the algorithms presented in the main paper, we use $\epsilon = 1/\sqrt{T}$ as our step size. 
Generally, the results in \Cref{sec: cp} hold when $\epsilon$ satisfies the following conditions:
 \begin{enumerate}
        \item $\bar{v}\epsilon<1$ and $2\gamma\epsilon<1$; 
        \item $\lim_{T\rightarrow \infty}\epsilon = 0$ and $\lim_{T\rightarrow \infty}T\epsilon = \infty$.
 \end{enumerate}
One can check it when reading the following proofs. 

\subsection{Proof of \Cref{thm: symmetric dominance of cp}}

\thmsymmetricdominanceofcp*

To show that the performance of Algorithm \ref{alg: coordinated pacing} outperforms the benchmark in symmetric cases, we first prove that it converges to some equilibrium defined below. With the convergence results, we reduce the problem of comparing the asymptotic performance of our algorithms with players' benchmark utility, to that comparing the two equilibrium utilities. 

Specifically, fix the pacing variables to be $\xi_1, \cdots, \xi_K$. Let $d_k\coloneqq \max\left\{d^I_k, d^O\right\}$, where $d^I_k = \max_{i\in\calK: i\ne k} \frac{v_i}{1+\xi_i}$. We can compute  bidder $k$'s expected expenditure per auction under Algorithm \ref{alg: coordinated pacing} to be 
\[\EXPcostCP(\bxi)\coloneqq \bbE_{\bv,d^O}\left[\bm{1}\left\{v_k\geq (1+\xi_k)d_k\right\}d^O\right], \]
and her corresponding expected utility function is 
\begin{align}
    \EXPutilityCP(\bxi) \coloneqq T\left(\bbE_{\bv,d^O}\left[\bm{1}\left\{v_k \geq (1+\xi_k)d_k\right\}\left(v_k - d^O\right)\right]\right).  
\end{align}

Since $\EXPcostCP$ and $\EXPutilityCP$ are both decreasing in $\xi_k$, consider the equilibrium $\bxi^*$ defined by the following complementarity conditions:
\begin{equation}
\label{eqn: xi definition}
\xi^*_k\geq 0 \perp \EXPcostCP(\bxi^*)\leq \rho_k, \forall k \in \calK.
\end{equation}
Since $\bm{G}^\rmCP$ is strongly monotone, $\bxi^*$ exists and is unique.

First, with the help of the pacing strategy, we have the performance of Algorithm \ref{alg: coordinated pacing} converges to $\EXPutilityCP(\bxi^*)$, in arbitrary problem instances. 
\begin{restatable}{lemma}{lmmconvergenceofCP}
\label{lmm: convergence of CP}
Suppose that $\bm{G}^\rmCP$ is strongly monotone. We have
\begin{align}
\label{eqn: convergence of cp}
\liminf_{T\rightarrow\infty}\frac{1}{T}(\utilityCP - \EXPutilityCP(\bxi^*))\geq 0, ~\forall{k \in \calK}. 
\end{align}
\end{restatable} 

Therefore, the problem reduces to the comparison of equilibrium utilities,  $\EXPutility(\blambda^*)$ and $\EXPutilityCP(\bxi^*)$. In the symmetric cases, we first show that everyone's pacing parameter is lower compared to $\blambda^*$. 
We can utilize this property, to show that the sets related to the indicator function in $\EXPutilityCP$ are supersets of that in $\EXPutility$. By further showing that the utilities on those sets are no less in the former than the latter, we finished the comparison of two equilibria. 

Specifically, to prove \eqref{eqn: symmetric dominance of cp}, by \Cref{lmm: convergence of CP}, it suffices to show that 
\begin{align}
\label{eqn: comparison of cp equilibrium} \EXPutilityCP(\bxi^*) > \EXPutility(\blambda^*).  
\end{align}
By symmetry, for all $k\in \calK$, we have 
\begin{align*}
\lambda^*_k=\lambda^*_1,~&\text{and}~ \xi^*_k = \xi^*_1,\\
\EXPcost(\blambda^*)=G_1(\blambda^*)~&\text{and}~\EXPcostCP(\bxi^*)=G^{\rm{CP}}_1(\bxi^*),\\
\EXPutility(\blambda^*)=U_1(\blambda^*)~&\text{and}~\EXPutilityCP(\bxi^*)=U^{\rm{CP}}_1(\bxi^*). 
\end{align*}
Without loss of generality, we assume $T=1$, then
\begin{align*}
\EXPutility(\blambda^*)=\bbE\left[\bm{1}\left\{\left(\forall i\in \calK: i\ne k\right)\left(v_k\geq v_i\right) \wedge \left(v_k\geq(1+\lambda^*_1)d^O\right)\right\}(v_k-d_k)\right],\\
\EXPutilityCP(\bxi^*)=\bbE\left[\bm{1}\left\{\left(\forall i\in \calK: i\ne k\right)\left(v_k\geq v_i\right) \wedge \left(v_k\geq(1+\xi^*_1)d^O\right)\right\}(v_k-d^O)\right].
\end{align*}

We first show the following lemmas. Notice that in the symmetric setting, we can without loss of generality consider the case $\xi_k=\xi_1$ and write $\EXPcostCP(\bxi)$ as the function of $\xi_1$. The same holds for other functions as well.
\begin{restatable}{lemma}{lmmstrictlydecreasingvalue}
\label{lmm: strictly decreasing value}
Suppose $\bm{G}^\rmCP(\bxi)$ is strongly monotone. When every bidder has the same value distribution and budget, $\EXPcostCP(\xi_1)$ is strictly decreasing on $[0,\bar{\xi}]$. Moreover, the bidders' expected value function $\EXPvalueCP(\xi_1)$ under Algorithm \ref{alg: coordinated pacing} is strictly decreasing on $[0,\bar{\xi}]$:
\[
\EXPvalueCP(\xi_1)\coloneqq\bbE\left[\bm{1}\left\{\left(\forall i\in \calK: i\ne k\right)\left(v_k\geq v_i\right) \wedge \left(v_k\geq(1+\xi_1)d^O\right)\right\}v_k\right]. 
\]
\end{restatable}

\begin{restatable}{lemma}{lmmrangeofxi}
\label{lmm: range of xi}
$\xi^*_k < \bar{\xi}_k$ for $k\in\calK$. 
\end{restatable}

\begin{restatable}{lemma}{lmmrelationofcostfunctions}
\label{lmm: relation of cost functions}
For any $\xi_1\in [0,\bar{\xi})$, we have
\[
\EXPcost(\xi_1) > \EXPcostCP(\xi_1). 
\]
\end{restatable}

Therefore when $\blambda^*>0$, we must have $\bxi^*<\blambda^*$, and $\EXPcostCP(\xi^*_1)\leq \EXPcost(\lambda^*_1)=\rho_1$. 
Similarly, we can define bidder $k$'s expected value function with respect to the benchmark, and it coincides with $\EXPvalueCP$. 
Moreover, since $\EXPcostCP$ and $\EXPvalueCP$ are strictly decreasing, we have 
\begin{align*}
\EXPutility(\lambda^*_1) &= \EXPvalueCP(\lambda^*_1)-\EXPcost(\lambda^*_1)\\ 
&< \EXPvalueCP(\xi^*_1) -\EXPcostCP(\xi^*_1) = \EXPutilityCP(\xi^*_1). 
\end{align*}

When $\blambda^*=0$, we have $\EXPcostCP(0)<\EXPcost(0)\leq \rho_k$, which means $\bxi^*=0$. Therefore
\[
\EXPutility(\lambda^*_1)=\EXPutility(0) = \EXPvalueCP(0)-\EXPcost(0)<\EXPvalueCP(0)-\EXPcostCP(0)=\EXPutility(\xi^*_1). 
\]
This completes the proof.

\subsection{Missing Proofs of Lemmas in \Cref{thm: symmetric dominance of cp}}

In the following proofs of the lemmas, Denote the support of $
\bm{F} \coloneqq (F_1,\dots,F_K,H)$ to be $\supp(\bm{F})$. That is, for any $(\bv,d^O)\in \supp(\bm{F})\subseteq \mathbb{R}^{K+1}$, $f_i(v_i)>0$, for any $i \in \calK$ and $h(d^O)>0$. When we talk about sets of $(\bv, d^O)$, we consider the subsets of  $\supp(\bm{F})$ without claiming explicitly.  For any subsets of $\{v_i\}^K_{i=1}\cup\{d^O\}$, we denote their supports similarly. For example,  $\supp(F_i, F_j) \coloneqq \left\{(\bv,d^O) \in \mathbb{R}^{K+1}: f_i(v_i)>0\wedge f_j(v_j)>0\right\}$. 

Due to the convergence of the pacing strategies, the proof of \Cref{lmm: convergence of CP} is analogous to that in \citet{balseiro2019learning}, which we shall omit here.

\lmmstrictlydecreasingvalue* 
\begin{proof}
We first show that $\EXPcostCP$ is strictly decreasing. For any $\xi_1 < \xi^\prime_1$, 
\begin{align*}
&\left(\xi_1-\xi_1^\prime\right)\left(\EXPcostCP(\xi_1)-\EXPcostCP(\xi'_1)\right) \\
= &\left(\xi_1-\xi_1^\prime\right) \bbE\left[\bm{1}\left\{\left(\forall i\in \calK: i\ne k\right)\left(v_k\geq v_i\right) \wedge \left((1+\xi'_1)d^O > v_k\geq(1+\xi_1)d^O\right)\right\}d^O\right]\\
=&1/K (\xi_1\mathbf{1}-\xi'_1\mathbf{1})\left(\bm{G}^{\rmCP}(\xi_1\mathbf{1})-\bm{G}^{\rmCP}(\xi'_1\mathbf{1})\right) <0. 
\end{align*}

Moreover, we have 
\begin{claim}
\label{lmm: non zero measure}
For $\xi_1, \xi'_1 \in [0,\bar{\xi}]$ such that $\xi_1 <\xi'_1$,  
% \[\left\{(1+\xi'_1)d^O > v_k\geq(1+\xi_1)d^O\right\}\]
\[
\left\{\left(\forall i\in \calK: i\ne k\right)\left(v_k \geq v_i\right)\wedge \left((1+\xi'_1)d^O > v_k\geq(1+\xi_1)d^O\right)\right\}
\]
has non-zero measure.
\end{claim}

Then 
\begin{align*}
\EXPvalueCP(\xi_1) - \EXPvalueCP(\xi'_1)  =
\bbE\left[\bm{1}\left\{\left(\forall i\in \calK: i\ne k\right)\left(v_k\geq v_i\right) \wedge \left((1+\xi'_1)d^O > v_k\geq(1+\xi_1)d^O\right)\right\}v_k\right]>0.\\
\end{align*}
This completes the proof. 
\end{proof}

\lmmrangeofxi*

\begin{proof}

$\xi^*_k < \bar{\xi}_k$ because
\begin{align*}
    \EXPcostCP\left(\bar{\xi}_k, \bxi^*_{-k}\right) = \bbE \left[\bm{1}\left\{\frac{v_k}{1+\bar{\xi}_k}\geq d_k\right\}d^O\right] \overset{(\text a)}{\leq} \bbE \left[\bm{1}\left\{\frac{v_k}{1+\bar{v}_k/\rho_k}\geq d_k\right\}d^O\right] \leq \frac{v_k}{1+\bar{v}_k/\rho_k} < \rho_k,
\end{align*}
where (a) follows from $\bar{\xi}_k \geq \bar{v}_k/\rho_k$. 

\end{proof}

\lmmrelationofcostfunctions*
\begin{proof}
Since 
\[
g(\xi_1)\coloneqq \EXPcost(\xi_1) - \EXPcostCP(\xi_1)=\int_{x}\bar{F}(x) \int^{\frac{x}{1+\xi_1}}_{0}H(y)\,dy\,dL \geq 0,
\]
where $\bar{F}(x)=1-F(x)$ and $L$ is the cumulative probability function of $d^I_k\coloneqq \max_{i\in\calK: i \ne k} v_i$. 

Since by \Cref{lmm: non zero measure} and the identical independence of bidders' distributions, we have, for any $\xi_1 \in [0,\bar{\xi})$, 
\[
\left\{d^I_{k} \geq (1+\xi_1)d^O\right\} ~~~\text{and}~~~\left\{v_k \geq d^I_k\right\}
\]
have non-zero measure, and therefore $\bbE\left[\bm{1}\left\{v_k \geq d^I_k\right\}\bm{1}\left\{d^I_k\geq (1+\xi_1)d^O\right\}d^I_k/(1+\xi_1)\right]>0$. 

Note that for any $\xi_1 \in [0,\bar{\xi})$,
\[
\frac{dg}{d\xi_1}=-\int_{x}\frac{x}{(1+\xi_1)^2}\bar{F}(x)H(\frac{x}{1+\xi_1})\,dx\,dL = -\frac{1}{1+\xi_1}\bbE\left[\bm{1}\left\{v_k \geq d^I_k\right\}\bm{1}\left\{d^I_k\geq (1+\xi_1)d^O\right\}d^I_k/(1+\xi_1)\right]<0. 
\]
We have $\EXPcost(\xi_1)>\EXPcostCP(\xi_1)$. 
This completes the proof. 

\end{proof}

\subsection{Other properties of CP in General Cases}

In symmetric cases,  $\bxi^*\leq \blambda^*$. Here, we show this holds in general cases under mild conditions.

\begin{restatable}{lemma}{lmmdecreasingparameter}
\label{lmm: decreasing parameter}
If it holds that for any $k\in\calK$, 
\[
\left\{\frac{v_k}{1+\xi_k^*}\geq d^O\right\} \subsetneq \left\{\frac{v_k}{1+\lambda_k^*}\geq d^O\right\}.
\]
then, 
\begin{align}
    \blambda^*\geq \bxi^*. 
\end{align}
\end{restatable}
For example, the ``if'' condition holds when there exists $c>0$, such that $H(d^O)\ne 0$, for $d^O \in (0,c)$. As $H$ is the cdf of $d^O$ and is an increasing function, this is equivalent to saying  $H(d^O)>0$ for all $d^O>0$. 

\begin{proof}
Writing down the specific formulation of $\EXPcost$ and $\EXPcostCP$ respectively, we have
\begin{align*}
    \EXPcost(\blambda) &= \bbE\left[\bm{1}\left\{\frac{v_k}{1+\lambda_k} \geq \max\left\{\max_{i\ne k}\frac{v_i}{1+\lambda_i}, d^O\right\}\right\}\max\left\{\max_{i\ne k}\frac{v_i}{1+\lambda_i}, d^O\right\} \right]\\
    &=\bbE\left[\prod_{i\ne k}\bm{1}\left\{\frac{v_k}{1+\lambda_k}\geq \frac{v_i}{1+\lambda_i}\right\}\bm{1}\left\{\frac{v_k}{1+\lambda_k}\geq d^O\right\}\max\left\{\max_{i\ne k}\frac{v_i}{1+\lambda_i}, d^O\right\}\right],\\
    \EXPcostCP(\bxi) &= \bbE\left[\bm{1}\left\{\frac{v_k}{1+\xi_k} \geq \max\left\{\max_{i\ne k}\frac{v_i}{1+\xi_i}, d^O\right\}\right\}d^O \right]\\
    &=\bbE\left[\prod_{i\ne k}\bm{1}\left\{\frac{v_k}{1+\xi_k}\geq \frac{v_i}{1+\xi_i}\right\}\bm{1}\left\{\frac{v_k}{1+\xi_k}\geq d^O\right\} d^O\right].\\
\end{align*}

Now suppose for the sake of contradiction, $\xi^*_i > \lambda^*_i$ for $i\in \{1,\dots, s\}$, for some $s\leq K$, and $\xi^*_i \leq \lambda^*_i$ for $i> s$. 
Consider $k=1$, and we have $G^{\rmCP}_1(\bxi^*) = \rho_1$. Since $\xi^*_i \leq \lambda^*_i$ for $i>s$, we have,
\begin{align}
\left\{\frac{v_1}{1+\xi^*_1}\geq \frac{v_i}{1+\xi^*_i}\right\} \subseteq \left\{\frac{v_1}{1+\lambda^*_1}\geq \frac{v_i}{1+\lambda^*_i}\right\}, \forall~i>s;
\end{align}
Since $\xi^*_1 > \lambda^*_1$, we have,
\begin{align}
\left\{\frac{v_1}{1+\xi^*_1}\geq d^O\right\} \subsetneq \left\{\frac{v_1}{1+\lambda^*_1}\geq d^O\right\}. 
\end{align}

To make $G^{\rmCP}_1(\bxi^*_1)=\rho_1 \geq G_1(\blambda^*_1)$, there must exist $\sigma(1)\in \left\{2,\dots,s \right\}$, such that
\[
\left\{\frac{v_1}{1+\xi^*_1}\geq \frac{v_{\sigma(1)}}{1+\xi^*_{\sigma(1)}} \right\}\supsetneq \left\{\frac{v_1}{1+\lambda^*_1}\geq \frac{v_{\sigma(1)}}{1+\lambda^*_{\sigma(1)}} \right\},
\]
since otherwise,
\[
\left\{\frac{v_1}{1+\xi^*_1} \geq \max\left\{\max_{i\ne 1}\frac{v_i}{1+\xi^*_i}, d^O\right\}\right\} \subsetneq \left\{\frac{v_1}{1+\lambda^*_1} \geq \max\left\{\max_{i\ne 1}\frac{v_i}{1+\lambda^*_i}, d^O\right\}\right\}. 
\]
Since $\max\left\{\max_{i\ne 1}\frac{v_i}{1+\lambda^*_i},d^O\right\}\geq d^O$, 
we have $G_1(\blambda^*)>G^{\rmCP}_1(\bxi^*)$, contradicting that $G^{\rmCP}_1(\bxi^*_1)\geq G_1(\blambda^*_1)$.

This means that
\[
\frac{1+\xi^*_1}{1+\xi^*_{\sigma(1)}} < \frac{1+\lambda^*_1}{1+\lambda^*_{\sigma(1)}}. 
\]
Similarly, for $\sigma(1)$, we have $\sigma^2(1)\in \{1,\dots,s\}$, and moreover 
\[
\frac{1+\xi_{\sigma^{\ell-1}(1)}}{1+\xi_{\sigma^{\ell}(1)}} < \frac{1+\lambda_{\sigma^{\ell-1}(1)}}{1+\lambda_{\sigma^{\ell}(1)}}, ~~\text{for }\ell=0,1,\dots \text{ and } \sigma^\ell(1)\in \left\{1,\dots,s\right\}
\]
Since $s<\infty$, there must exist $\ell_0$, such that $1 = \sigma^{\ell_0}(1)$. Let us multiply the above inequalities from $0$ to $\ell_0$, we have
\[
1 = \prod^{\ell_0}_{\ell=0} \frac{1+\xi_{\sigma^{\ell-1}(1)}}{1+\xi_{\sigma^{\ell}}(1)} < \prod^{\ell_0}_{\ell=0} \frac{1+\lambda_{\sigma^{\ell-1}(1)}}{1+\lambda_{\sigma^{\ell}}(1)} = 1,
\]
leading to a contradiction. 

\end{proof}

\section{Omitted Proofs in Section~\ref{sec: hp}}

\subsection{Details on the Choice of Step Sizes}
In the algorithms presented in the main paper, we use $\epsilon = 1/\sqrt{T}$ as our step size. 
Generally, the results in \Cref{sec: hp} hold when $\epsilon$ satisfies the following conditions:
\begin{itemize}
	\item $\bar{v}\epsilon < 1, 2\gamma\epsilon < 1$ and $\underline{G}'\epsilon < 1$;
	\item $\lim_{T\rightarrow \infty}\epsilon = 0$ and $\lim_{T\rightarrow \infty}T\epsilon^{3/2} = \infty$.
\end{itemize}
One can check it when reading the following proofs. 

\subsection{Proof of \Cref{thm: dominance of hp}}

\thmdominanceofhp*

First, we show \Cref{line: lambda update} in \Cref{alg: hybrid coordinated pacing} is equivalent to \eqref{eqn: lambda update equivalent}.
\begin{align*}
    z'_{k, t} &= \bm{1}\left\{b^{I}_{k, t} \geq z_{k, t} \right\}\max \left(z_{k,t}, x_{k,t}d^{I}_{k, t}\right) \\
    &= \bm{1}\left\{b^{I}_{k, t} \geq x_{k, t}d^O_{t} \right\}\max \left(x_{k, t}d^O_{t}, x_{k,t}d^{I}_{k, t}\right) \\
    &= \bm{1}\left\{b^{I}_{k, t} \geq x_{k, t}d^O_{t} \right\}x_{k, t}d_{k, t}.
\end{align*}
When $x_{k, t} = 1$, as bidder $k$ wins both the internal election and the real auction, we have $b^{O}_{k, t} \geq d^O_{t}$ and $b^{I}_{k, t} \geq d^I_{k, t}$. Hence, 
\begin{align*}
    z'_{k, t} = \bm{1}\left\{b^{I}_{k, t} \geq d^O_{t} \right\}d_{k, t} = \bm{1}\left\{b^{I}_{k, t} \geq d_{k, t} \right\}d_{k, t}. 
\end{align*}
When $x_{k, t} = 0$, either $b^{O}_{k, t} < d^O_{t}$ or $b^{I}_{k, t} < d^I_{k, t}$ holds. If $b^{I}_{k, t} \geq d^I_{k, t}$, we have $b^{I}_{k, t} \leq b^{O}_{k, t} < d^O_{t} \leq d_{k, t}$, so in either case
\begin{align*}
    z'_{k, t} = 0 = \bm{1}\left\{b^{I}_{k, t} \geq d_{k, t} \right\}d_{k, t}. 
\end{align*}

As the update of $\lambda_k$ coincides with the subgradient descent scheme of an individual adaptive pacing strategy, we can apply the results in \citet{balseiro2019learning} to characterize the convergence of the sequence of pseudo multipliers. In what follows we use the standard norm notation $\|\yy\|_p \coloneqq \left(\sum_{k=1}^K|y_k|^p\right)^{1/p}$ for a vector $\yy \in \bbR^K$.
\begin{lemma}[\citet{balseiro2019learning}]
\label[lemma]{lem: lambda convergence}
Suppose that \Cref{asm: strongly-monotone} holds. There exists a unique $\blambda^*$ that satisfies \eqref{eqn: lambda star definition} and
$$\lim_{T\to \infty} \frac{1}{T}\sum_{t = 1}^T \bbE \left[\|\blambda_{t} - \blambda^*\|_1\right] = 0.$$
\end{lemma}

The following \Cref{lem: mu convergence} establishes a similar convergence result for multipliers $\bmu$.
\begin{restatable}{lemma}{lemmuconvergence}
\label[lemma]{lem: mu convergence}
    For each $k$, there exists a unique $\mu^*_k \leq \lambda^*_k$ that satisfies \eqref{eqn: mu star definition} and for each $k\in \calK$, $$\lim_{T\to \infty} \frac{1}{T} \sum_{t = 1}^T \bbE \left[|\mu_{k, t} - \mu^*_k|\right] = 0.$$
\end{restatable}

\Cref{lem: lambda convergence} and \Cref{lem: mu convergence} ensure the existence and uniqueness of $(\bmu^*, \blambda^*)$, so that we can bound the expected payoffs with a well-defined $\EXPutilityHP(\mu^*_k, \blambda^*_k)$. 

\begin{restatable}{lemma}{lemconvergenceofhp}
\label[lemma]{lem: convergence of hp}
\begin{align}
\label{eqn: convergence of hp}
\liminf_{\substack{T\rightarrow\infty\\B_k=\rho_kT}}\frac{1}{T}(\utilityHP - \EXPutilityHP(\mu^*_k, \blambda^*))\geq 0.
\end{align}
\end{restatable}

\Cref{lem: convergence of hp} establishes that the expected performance of Algorithm \ref{alg: hybrid coordinated pacing} can be asymptotically lower bounded by $\EXPutilityHP(\mu^*_k, \blambda^*_k)$. The proof of \Cref{lem: convergence of hp} consists of two parts. We first use the fact that budgets are not depleted too early as shown in the proof of \Cref{lem: mu convergence}. Then we use that the expected utility functions are Lipschitz continuous as argued in \Cref{cor: lipschitz}.

With \Cref{lem: convergence of hp} in hand, it only remains to show that $\EXPutilityHP(\mu^*_k, \blambda^*) \geq \EXPutility(\blambda^*)$. Since $\lambda^*_k \geq \mu^*_k \geq 0$ by \Cref{lem: mu convergence}, we first have
\begin{align*}
    \EXPutilityHP(\mu^*_k, \blambda^*) &\coloneqq  T \cdot \bbE_{\bv, d^O}\left[\bm{1}\left\{\frac{v_k}{1 + \lambda_k^*} \geq d^I_k \land \frac{v_k}{1 + \mu_k^*} \geq d^O \right\}(v_k - d^O)\right] \\
    &\geq T \cdot \bbE_{\bv, d^O}\left[\bm{1}\left\{\frac{v_k}{1 + \lambda_k^*} \geq d^I_k \land \frac{v_k}{1 + \lambda_k^*} \geq d^O \right\}(v_k - d^O)\right] \\
    &= \EXPutilityHP(\lambda^*_k, \blambda^*).
\end{align*}
Then we decompose $\EXPutilityHP(\lambda^*_k, \blambda^*)$ into two parts, $\EXPvalueHP(\lambda^*_k, \blambda^*)$ and $\EXPcostHP(\lambda^*_k, \blambda^*)$, where $\EXPvalueHP(\mu_k, \blambda)$ is defined by
\begin{align}
\label{eqn: hp exp value definition}
    \EXPvalueHP(\mu_k, \blambda) = \bbE_{\bv, d^O}\left[\bm{1}\left\{\frac{v_k}{1 + \lambda_k} \geq d^I_k \land \frac{v_k}{1 + \mu_k} \geq d^O \right\}v_k\right],
\end{align}
and $\EXPcostHP(\mu_k, \blambda)$ defined by \eqref{eqn: hp exp cost definition}. We also decompose the benchmark $\EXPutility(\blambda^*)$ into $\EXPvalue(\blambda^*)$ and $\EXPcost(\blambda^*)$, where $\EXPvalue(\blambda)$ is defined by
\begin{align}
\label{eqn: ip exp value definition}
    \EXPvalue(\blambda) = \bbE_{\bv, d^O}\left[\bm{1}\left\{\frac{v_k}{1 + \lambda_k} \geq d_k \right\}v_k\right],
\end{align}
and $\EXPcost(\blambda)$ is defined by \eqref{eqn: ip exp cost definition}. Observe that $\EXPvalueHP(\lambda^*_k, \blambda^*) = \EXPvalue(\blambda^*)$.

Therefore, we have
\begin{align*}
    \EXPutilityHP(\mu^*_k, \blambda^*) & \geq \EXPutilityHP(\lambda^*_k, \blambda^*) = T\cdot \EXPvalueHP(\lambda^*_k, \blambda^*) - T \cdot \EXPcostHP(\lambda^*_k, \blambda^*) \\
    & = T\cdot \EXPvalue(\blambda^*) - T \cdot \EXPcostHP(\lambda^*_k, \blambda^*) \\
    & \overset{(\text a)}{\geq} T\cdot \EXPvalue(\blambda^*) - T\cdot \EXPcost(\blambda^*) = \EXPutility(\blambda^*),
\end{align*}
where (a) follows from  \Cref{lem: strictly pareto}, which also indicates that the equality strictly holds with at most one bidder.

We defer the proofs of key lemmas to the next subsection. It is worth mentioning that the proof techniques differ from those of the similar results in previous work in the following ways: (1) First, the expected utility and expenditure functions are more complicated as they involve two multiplier vectors. (2) Second, Algorithm \ref{alg: hybrid coordinated pacing} imposes Euclidean projection with different intervals on $\mu_{k}$ in different rounds, i.e., $[0, \lambda^{k, t+1}]$ on $\mu_{k, t+1}$, so it is subtle to deal with the positive projection errors.

\subsection{Proof of Lemmas in \Cref{thm: dominance of hp}}

\lemmuconvergence*
\begin{proof}
% [Proof of \Cref{lem: mu convergence}]

\textbf{Step 1: Existence and uniqueness of $\bmu^*$.} By \Cref{lem: lambda convergence}, $\blambda^*$ defined by \eqref{eqn: lambda star definition} exists and is unique. By \Cref{lem: strictly pareto}, we have 
\begin{align}
    \label{eqn: G hp smaller than G}
    \EXPcostHP(\lambda^*_k, \blambda^*) \leq \EXPcost(\blambda^*).
\end{align} 
As $\EXPcostHP(\mu_k, \blambda)$ is continuous and strictly decreases in $\mu_k$ by \Cref{asm: hp} and $\EXPcost(\blambda^*) \leq \rho_k$ by the definition of $\blambda^*$, we can find a unique $\mu^*_k \in [0, \lambda^*_k]$ that satisfies $\mu^*_k \geq 0 \perp \EXPcostHP(\mu^*_k, \blambda^*) \leq \rho_k$.

\textbf{Step 2: bounding the ending time.} Let $\tau_k = \sup\{t \leq T: B_{k, t} \geq \bar{v}_k\}$ be the last period in which bidder $k$'s remaining budget is larger than her maximum value, and let $\tau = \min_k \tau_k$. 
We denote by $P_{k,t}$ the projection error introduced by \Cref{line: mu smaller than lambda} of \Cref{alg: hybrid coordinated pacing}:
$$P_{k, t} \coloneqq \mu_{k,t} - \epsilon(\rho_k-z_{k,t}) - P_{[0, \lambda_{k,t+1}]} \left(\mu_{k,t} - \epsilon(\rho_k-z_{k,t})\right),$$ and then $\mu_{k,t+1} = \mu_{k,t} - \epsilon(\rho_k-z_{k,t}) - P_{k, t}$. Summing up from $t = 1$ to $\tau_k$, we have
$$\sum_{t = 1}^{\tau_k} \left(z_{k,t} - \rho_k\right) = \frac{\mu_{k, \tau_k+1} - \mu_{k, 1}}{\epsilon} + \frac{1}{\epsilon}\sum_{t = 1}^{\tau_k} P_{k, t} \leq \frac{\mu_{k, \tau_k+1}}{\epsilon} + \frac{1}{\epsilon}\sum_{t = 1}^{\tau_k} P_{k, t}^+.$$

When $\tau_k < T$, bidder $k$ hits her budget constraint in round $\tau_k + 1$, so we should have
$\sum_{t = 1}^{\tau_k} z_{k,t} \geq \rho_k T - \bar{v}_k$. Together with $\mu_{k, t} \leq \lambda_{k, t} \leq \bar{\mu}_k$, we obtain
\begin{equation}
\label{eqn: bound ending time}
    T-\tau_k \leq \frac{\bar{v}_k}{\rho_k} + \frac{\bar{\mu}_k}{\epsilon\rho_k} + \frac{1}{\epsilon\rho_k}\sum_{t = 1}^{T} P_{k, t}^+.
\end{equation}
Note that inequality \eqref{eqn: bound ending time} also trivially holds when $\tau_k = T$ since all the terms on the right-hand side are positive.

Now we bound $P_{k, t}^+$ on the right-hand side of \eqref{eqn: bound ending time}. We first show that there is no positive projection error introduced by \Cref{line: lambda update} of \Cref{alg: hybrid coordinated pacing} whenever $\epsilon\bar{v}_k \leq 1$.
\begin{align*}
    \lambda_{k, t} - \epsilon\left(\rho - z_{k, t}'\right) \overset{(\text a)}{\leq} & \lambda_{k, t} + \epsilon\left(b^I_{k, t} - \rho_k\right) \overset{(\text b)}{\leq} \lambda_{k, t} + \frac{\epsilon \bar{v}_k}{1+\lambda_{k, t}} - \epsilon\rho_k \\
    \overset{(\text c)}{\leq} & \bar{\mu}_k + \frac{\epsilon \bar{v}_k}{1+\bar{\mu}_k} - \epsilon\rho_k \overset{(\text d)}{\leq} \bar{\mu}_k,
\end{align*}
where (a) follows from $z_{k, t}' = \bm{1}\left\{b^{I}_{k, t} \geq d_{k, t} \right\}d_{k, t}$; (b) holds by $b^I_{k, t} \leq v_k/(1+\lambda_{k, t}) \leq \bar{v}_k/(1+\lambda_{k, t})$; (c) holds since $\lambda_{k, t} + \epsilon \bar{v}_k/
(1+\lambda_{k, t})$ increases in $\lambda_{k, t}$ whenever $\epsilon\bar{v}_k \leq 1$; and (d) holds since $\bar{\mu}_k \geq \bar{v}_k / \rho_k$. 
As a result,  $\lambda_{k, t+1} \geq \lambda_{k, t} - \epsilon(\rho_k-z_{k,t}')$, and we have
\begin{align*}
    P_{k, t}^+ & =\left(\mu_{k,t} - \epsilon(\rho_k-z_{k,t}) - \lambda_{k,t+1}\right)^+ \\
    & \leq \left(\mu_{k,t} - \lambda_{k, t} + \epsilon\left(z_{k,t} - z_{k,t}'\right)\right)^+.
\end{align*}
Since $\mu_{k,t} \leq \lambda_{k, t}$, $P_{k, t}^+ > 0$ is only possible when $d^O_t = z_{k, t} > z_{k,t}' = 0$. In this case, bidder $k$ is the final winner with $b^O_{k, t} \geq d^O_{t}$ but she could not have won if using $b^I_{k, t} < d^O_{ t}$ as her final bid. We have
\begin{align}
    P_{k, t}^+ & \leq \left(\mu_{k,t} - \lambda_{k, t} + \epsilon d^O_{t}\right)^+ \leq \epsilon d^O_{t} \leq \epsilon b^O_{k,t} = \frac{\epsilon v_{k,t}}{1+\mu_{k,t}} \leq \epsilon \bar{v}_k, \label{eqn: bound proj error value}
\end{align}
which also implies that a positive error requires $\lambda_{k, t} - \mu_{k, t} < \epsilon d^O_{t}$. Then we can bound the probability as follows, 
\begin{align}
    \bbP \left[P_{k, t} > 0\right] & \leq \bbP \left[d^O_{t} \in \left(b^I_{k, t}, b^O_{k, t}\right]\right] \leq \bbP \left[d^O_{t} \in \left(\frac{v_{k, t}}{1 + \mu_{k, t} + \epsilon d^O_{t}}, \frac{v_{k, t}}{1 + \mu_{k, t}}\right]\right], \notag \\
    & \leq \frac{\epsilon d^O_{t}v_{k, t}}{(1 + \mu_{k, t})(1 + \mu_{k, t} + \epsilon d^O_{t})}\overline{h} \leq \epsilon \bar{v}^2_k \overline{h}, \label{eqn: bound proj error probability}
\end{align}
where we denote by $\overline{h}$ by the upper bound on the density funciton of $d^O_t$.

Combining \eqref{eqn: bound proj error value} and \eqref{eqn: bound proj error probability}, we have
\begin{equation}
\label{eqn: bound proj error exp}
    \bbE_{d^O_{t}} \left[P_{k, t}^+ \mid \calH_{k, t} \right] \leq \epsilon^2 \bar{v}^3_k \overline{h}.
\end{equation}
Taking expectations over \eqref{eqn: bound ending time} we have
\[
    \bbE \left[T-\tau_k\right] \leq \frac{\bar{v}_k}{\rho_k} + \frac{\bar{\mu}_k}{\epsilon\rho_k} + \frac{ T \epsilon \bar{v}^3_k \overline{h}}{\rho_k},
\]
and thus
\begin{equation}
\label{eqn: bound expected ending time}
   \bbE \left[T-\tau\right] \leq \frac{\bar{v}}{\underline{\rho}} + \frac{\bar{\mu}}{\epsilon\underline{\rho}} + \frac{ T \epsilon \bar{v}^3 \overline{h}}{\underline{\rho}}, 
\end{equation}
where $\bar{v} = \max_k \bar{v}_k$, 
$\underline{\rho} = \min_k \rho_k$ and $\bar{\mu} = \max_k \bar{\mu}_k$.

\textbf{Step 3: bounding the mean absolute errors.} We consider the update formula of $\mu_{k, t}$ to prove that the sequence converges to $\mu^*_k$.
\begin{align}
    \left(\mu_{k, t+1} - \mu^*_k\right)^2 & \overset{(\text a)}{\leq} \left(\mu_{k, t} - \epsilon \left(\rho_k - z_{k, t}\right) - P_{k, t}^+ - \mu^*_k\right)^2 \notag \\
    & =\left(\mu_{k, t} - \mu^*_k\right)^2 - 2 \epsilon \left(\rho_k - z_{k, t}\right)\left(\mu_{k, t} - \mu^*_k\right) + \epsilon^2 \left(\rho_k - z_{k, t}\right)^2 \notag \\
    & + (P_{k, t}^+)^2 - 2 P_{k, t}^+\left(\mu_{k, t} - \epsilon \left(\rho_k - z_{k, t}\right) - \mu^*_k\right) \notag \\
    & \overset{(\text b)}{\leq} (\mu_{k, t} - \mu^*_k)^2 - 2 \epsilon (\rho_k - z_{k, t})(\mu_{k, t} - \mu^*_k) + \epsilon^2\bar{v}^2_k + (P_{k, t}^+)^2 + 2P_{k, t}^+\bar{\mu}_k, \label{eqn: mu convergence step 1}
\end{align}
where (a) follows from a standard contraction property of the Euclidean projection operator and (b) is due to $(\rho_k - z_{k, t})^2 \leq \bar{v}^2_k$ and $\mu_{k, t} - \epsilon (\rho_k - z_{k, t}) > \lambda_{k, t+1} \geq 0$ when $P_{k, t}^+ > 0$. 

Define $\delta_{k,t} = \bbE \left[(\mu_{k, t} - \mu^*_k)^2\bm{1}\left\{t \leq \tau \right\}\right]$. Taking expectations on \eqref{eqn: mu convergence step 1} we obtain
\begin{align}
    \delta_{k,t+1} & \overset{(\text a)}{\leq} \delta_{k,t} - 2 \epsilon \bbE\left[(\rho_k - \EXPcostHP(\mu_{k, t}, \blambda_t))(\mu_{k, t} - \mu^*_k)\right] + \epsilon^2\bar{v}^2_k + \bbE \left[(P_{k, t}^+)^2\right] + 2 \bar{\mu}_k\bbE \left[P_{k, t}^+\right] \notag \\
    & \overset{(\text b)}{\leq} \delta_{k,t} + 2 \epsilon \bbE \left[(\EXPcostHP(\mu_{k, t}, \blambda_t) - \rho_k)(\mu_{k, t} - \mu^*_k)\right] + \epsilon^2\underbrace{\left(\bar{v}^2_k + \epsilon\bar{v}^4_k \overline{h} + 2 \bar{\mu}_k\bar{v}^3_k \overline{h}\right)}_{R} \label{eqn: mu convergence step 2}
\end{align}
where (a) holds since $\bm{1}\left\{t \leq \tau \right\}$ monotonically decreases with $t$ and the expectation of $z_{k, t}$ is $G_{k}(\mu_{k, t}, \blambda_t)$ conditioned on $\mu_{k, t}$ and $\blambda_t$; (b) holds by inequalities \eqref{eqn: bound proj error value} and \eqref{eqn: bound proj error probability}. For the second term in \eqref{eqn: mu convergence step 2}, we have
\begin{align}
\label{eqn: mu convergence step 3}
    \left(\EXPcostHP(\mu_{k, t}, \blambda_t) - \rho_k\right)\left(\mu_{k, t} - \mu^*_k\right) & =\left(\EXPcostHP(\mu_{k, t}, \blambda_t) - \EXPcostHP(\mu^*_{k}, \blambda^*) + \EXPcostHP(\mu^*_{k}, \blambda^*) - \rho_k\right)\left(\mu_{k, t} - \mu^*_k\right) \notag \\
    & \overset{(\text a)}{\leq} (\EXPcostHP(\mu_{k, t}, \blambda_t) - \EXPcostHP(\mu^*_{k}, \blambda^*))(\mu_{k, t} - \mu^*_k) \notag \\
    & \overset{(\text b)}{\leq} - \underline{G}' \left(\mu_{k, t} - \mu^*_k\right)^2 + \bar{v}_k^2\bar{f} |\mu_{k, t} - \mu^*_k| \|\blambda_{t} - \blambda^*\|_1 \notag \\
    & \overset{(\text c)}{\leq} - \underline{G}' \left(\mu_{k, t} - \mu^*_k\right)^2 + \frac{\underline{G}'}{2}\left(\mu_{k, t} - \mu^*_k\right)^2 + \frac{\bar{v}_k^4\bar{f}^2}{2\underline{G}'}\|\blambda_{t} - \blambda^*\|^2_2 \notag \\
    & =- \frac{\underline{G}'}{2}(\mu_{k, t} - \mu^*_k)^2 + \frac{\bar{v}_k^4\bar{f}^2}{2\underline{G}'}\|\blambda_{t} - \blambda^*\|^2_2, 
\end{align}
where (a) holds by $\rho_k - \EXPcostHP(\mu^*_{k}, \blambda^*) \geq 0$, $\mu_{k, t} \geq 0$ and $\mu^*_k(\rho_k - \EXPcostHP(\mu^*_{k}, \blambda^*)) = 0$; (b) follows from \Cref{lem: lipschitz hp} and $\bar{f}$ denotes the upper bound on $f_k$; (c) holds by AM-GM inequality. Together with \eqref{eqn: mu convergence step 2} one has $$\delta_{k, t+1} \leq \left(1 - \epsilon \underline{G}'\right)\delta_{k,t} + \frac{\epsilon \bar{v}_k^4\bar{f}^2}{\underline{G}'}\eta_t + \epsilon^2 R,$$ where $\eta_t = \bbE \left[\|\blambda_{t} - \blambda^*\|_2^2\bm{1}\left\{t \leq \tau \right\}\right]$. The recursion gives
\begin{align*}
    \delta_{k, t} \leq \left(1 - \epsilon \underline{G}'\right)^{t-1}\bar{\mu}^2_k + \sum_{s=1}^{t-1}\left(1 - \epsilon \underline{G}'\right)^{t-1-s}\left(\frac{\epsilon \bar{v}_k^4\bar{f}^2}{\underline{G}'}\eta_s + \epsilon^2 R\right).
\end{align*}
By the result of \citet{balseiro2019learning}, we have
\begin{align*}
    \eta_t \leq K\bar{\mu}^2_k\left(1 - 2\gamma \epsilon\right)^{t-1} + \frac{K\bar{v}^2_k}{2\gamma}\epsilon.
\end{align*}

Thus we have
\begin{align*}
    \delta_{k, t} & \leq \left(1 - \epsilon \underline{G}'\right)^{t-1}\bar{\mu}^2_k + \frac{\epsilon K\bar{\mu}^2_k\bar{v}_k^4\bar{f}^2}{\underline{G}'}\sum_{s=1}^{t-1}\left(1 - \epsilon \underline{G}'\right)^{t-1-s}\left(1 - 2\gamma \epsilon\right)^{s-1}\\
    & + \left(\frac{K\bar{v}^6_k\bar{f}^2}{2\gamma \underline{G}'} + R\right)\epsilon^2\sum_{s=1}^{t-1} \left(1 - \epsilon \underline{G}'\right)^{t-1-s} \\
    & \leq \left(1 - \epsilon R_1\right)^{t-1}\bar{\mu}^2_k + \epsilon R_2\sum_{s=1}^{t-1}\left(1 - \epsilon R_1\right)^{t-2} + R_3\epsilon^2\sum_{s=0}^{t-2} \left(1 - \epsilon R_1\right)^{s} \\
    & \leq \left(1 - \epsilon R_1\right)^{t-1}\bar{\mu}^2_k + \epsilon R_2(t-1)\left(1 - \epsilon R_1\right)^{t-2} + \epsilon\frac{R_3}{R_1}, \\
\end{align*}
where we denote $\min\{\underline{G}', 2\gamma\}$ by $R_1$, $\frac{K\bar{\mu}^2_k\bar{v}_k^4\bar{f}^2}{\underline{G}'}$ by $R_2$ and $\left(\frac{K\bar{v}^6_k\bar{f}^2}{2\gamma \underline{G}'} + R\right)$ by $R_3$.

Consider the sum of square roots of $\delta_{k, t}$,
\begin{align}
    \sum_{t=1}^T \delta_{k, t}^{1/2} & \overset{(\text a)}{\leq} \sum_{t=1}^T \left(1 - \epsilon R_1\right)^{(t-1)/2}\bar{\mu}_k + \sum_{t=1}^T \epsilon^{1/2}\sqrt{\frac{R_3}{R_1}} + \sum_{t=1}^T \epsilon^{1/2} \sqrt{R_2(t-1)}  \left(1 - \epsilon R_1\right)^{(t-2)/2} \notag \\
    & \overset{(\text b)}{\leq} \frac{2\bar{\mu}_k}{\epsilon R_1} + T\epsilon^{1/2}\sqrt{\frac{R_3}{R_1}} + \epsilon^{1/2} \sqrt{R_2} \sum_{t=1}^{T-1} t \left(1 - \epsilon R_1\right)^{(t-1)/2} \notag \\
    & \overset{(\text c)}{\leq} \frac{2\bar{\mu}_k}{\epsilon R_1} + T\epsilon^{1/2}\sqrt{\frac{R_3}{R_1}} + \frac{4\sqrt{R_2}}{\epsilon^{3/2}R_1^2},\label{eqn: bound mean square roots}
\end{align}
where (a) holds since $\sqrt{x + y} \leq \sqrt{x} + \sqrt{y}$ for $x, y \geq 0$, (b) follows from $\sum_{t=1}^T \left(1 - \epsilon R_1\right)^{(t-1)/2} \leq \sum_{t=0}^{\infty} \left(1 - \epsilon R_1\right)^{t/2} = \frac{1}{1 - \sqrt{1 - \epsilon R_1}} \leq \frac{2}{\epsilon R_1}$, and (c) follows from $\sum_{t=1}^{T-1} t \left(1 - \epsilon R_1\right)^{(t-1)/2} \leq \frac{1}{\left(1 - \sqrt{1 - \epsilon R_1}\right)^2} \leq \frac{4}{\epsilon^2 R_1^2}$.

Combining \eqref{eqn: bound expected ending time} and \eqref{eqn: bound mean square roots}, we obtain
\begin{align*}
    \sum_{t = 1}^T \bbE \left[|\mu_{k, t} - \mu^*_k|\right] & = \sum_{t=1}^T \bbE \left[|\mu_{k, t} - \mu^*_k|\bm{1}\left\{t \leq \tau\right\} + |\mu_{k, t} - \mu^*_k|\bm{1}\left\{t > \tau\right\}\right] \\
     & \overset{(\text a)}{\leq} \sum_{t=1}^T\delta_{k, t}^{1/2} + \bar{\mu}_k\bbE \left[T-\tau \right] \\
    & \leq \frac{2\bar{\mu}}{\epsilon R_1} + T\epsilon^{1/2}\sqrt{\frac{R_3}{R_1}} + \frac{4\sqrt{R_2}}{\epsilon^{3/2}R_1^2} + \frac{\bar{v}\bar{\mu}}{\underline{\rho}} + \frac{\bar{\mu}^2}{\epsilon\underline{\rho}} + \frac{ T \epsilon \bar{v}^3 \bar{\mu} \overline{h}}{\underline{\rho}}, 
\end{align*}
where (a) holds by $\bbE \left[|y|\right] \leq \sqrt{\bbE \left[y^2\right]}$. Dividing by $T$, the result then follows because each term goes to zero when $\epsilon\sim T^{-1/2}$.
\end{proof}

\lemconvergenceofhp*
\begin{proof}
In the proof we consider an alternate framework in which bidders ignore their budgets, i.e., each bidder $k$ can post $b^I_{k, t} = \frac{v_{k, t}}{1+\lambda_{k, t}}, b^O_{k, t} = \frac{v_{k, t}}{1+\mu_{k, t}}$ and make payments even after depleting her budget. Let $\tau = \sup\{t \leq T: \forall k, B_{k, t} \geq \bar{v}_k\}$ be the last period in which the remaining budget of every bidder is larger than her maximum value. The performance of both the original and the alternate frameworks coincide until time $\tau_k$. Let $u_{k, t}$ be bidder $k$'s payoff at round $t$ in the alternate framework. We have
\begin{equation}
\label{eqn: convergence step 0}
    \utilityHP \overset{(\text a)}{\geq} \sum_{t=1}^{\tau} \bbE\left[u_{k, t}\right] \overset{(\text b)}{\geq} \sum_{t=1}^{T} \bbE\left[u_{k, t}\right] - \bar{v}_k \bbE \left[T-\tau\right], 
\end{equation}
where (a) holds by discarding all auctions after $\tau$ and (b) holds since $u_{k, t} \leq \bar{v}_k$.  

We next give a lower bound on the expected utility per auction.
\begin{align*}
    \bbE \left[u_{k, t}\right]  &\overset{(\text a)}{=} \bbE \left[\frac{1}{T}\EXPutilityHP(\mu_{k, t}, \blambda_{t})\right] \\
    &\overset{(\text b)}{\geq}  \frac{1}{T}\EXPutilityHP(\mu^*_{k}, \blambda^*) - K\bar{v}_k^2 \bar{f} \bbE \left[\|\blambda_{t} - \blambda^*\|_1\right] - \bar{v}_k^2\left(\bar{h} + \bar{f}\right) \bbE \left[\left|\mu_{k, t} - \mu^*_k\right|\right],
\end{align*}
where (a) holds by using the independence of $\bv_{t}$ and $d^O_t$ from the multipliers and taking conditional expectations; and (b) holds by the Lipschitz continuity of $\EXPutilityHP/T$ from \Cref{cor: lipschitz}. Summing over $t=1, \ldots, T$, we obtain
\begin{align}
    \sum_{t=1}^{T}\bbE \left[u_{k, t}\right] \geq \EXPutilityHP(\mu^*_{k}, \blambda^*) - \underbrace{K\bar{v}_k^2 \bar{f}}_{C_1} \sum_{t=1}^{T}\bbE \left[\left\|\blambda_{t} - \blambda^*\right\|_1\right] -\underbrace{\bar{v}_k^2\left(\bar{h} + \bar{f}\right)}_{C_2}\sum_{t=1}^{T} \bbE \left[\left|\mu_{k, t} - \mu^*_k\right|\right] , \label{eqn: convergence step 1}
\end{align}

By the result from \citet{balseiro2019learning} and \Cref{eqn: bound mean square roots} in \Cref{lem: mu convergence}, there exist positive constants $C_3$ through $C_7$ such that
\begin{align}
    \sum_{t=1}^{T}\bbE \left[\|\blambda_{t} - \blambda^*\|_1\right] & \leq \frac{C_3}{\epsilon} + C_{4}\epsilon^{1/2}T. \label{eqn: convergence step 2} \\
    \sum_{t=1}^{T} \bbE \left[\left|\mu_{k, t} - \mu^*_k\right|\right] & \leq \frac{C_5}{\epsilon} + C_6\epsilon^{1/2}T + \frac{C_7}{\epsilon^{3/2}}. \label{eqn: convergence step 3}
\end{align}
Note that the sequences of multipliers in \eqref{eqn: convergence step 2} and \eqref{eqn: convergence step 3} are those in the alternate framework, so we do not involve $\bm{1}\left\{t < \tau\right\}$ as we did in \Cref{lem: mu convergence}.

Putting \eqref{eqn: convergence step 0}, \eqref{eqn: convergence step 1}, \eqref{eqn: convergence step 2}, \eqref{eqn: convergence step 3} and \eqref{eqn: bound expected ending time} together, we obtain
\begin{align*}
    \utilityHP \geq \EXPutilityHP(\mu^*_{k}, \blambda^*) - \frac{C_1C_3 + C_2C_5}{\epsilon} - (C_1C_4 + C_2C_6)\epsilon^{1/2}T - \frac{C_2C_7}{\epsilon^{3/2}} - \frac{\bar{v}^2}{\underline{\rho}} - \frac{\bar{v}\bar{\mu}}{\epsilon\underline{\rho}} - \frac{ T \epsilon \bar{v}^4 \overline{h}}{\underline{\rho}},
\end{align*}
which concludes the proof. In particular, a selection of step size of order $\epsilon\sim T^{-1/2}$ guarantees a long-run average convergence rate of order $T^{-1/4}$.
\end{proof}

\section{Omitted Proofs in \Cref{sec: discussion}} 

\subsection{Proof of \Cref{lmm: pareto optimality of hp}}

\lmmparetooptimalityofhp*

\begin{proof}
As it has been shown in \Cref{lmm: convergence of CP} that the performance of CP with respect to each bidder $k$, $\Pi^{\rm{CP}}_k$, converges to $\EXPutilityCP(\bxi^*)$. It suffices for us to show $\sum_{k\in\calK}\EXPutilityCP(\bxi^*)$ is an upper bound of $\bbE[\pi^{\rm{H}}]$. 

Rewrite the form of $\pi^{\rm{H}}$ for a realization of values and the highest bids outside the coalition $(\bv,\dd^O)$: 

 \begin{equation}
 \label{eqn: hindsight optimal welfare}
 \begin{aligned}
 \pi^{\mathrm{H}}\left(\bv, \dd^O\right) \coloneqq \max_{\bx \in\{0,1\}^{K\times T}} & \sum_{t=1}^{T} \sum_{k\in \calK} x_{k, t}(v_{k, t} - d^O_{t}), \\
 \text { s.t. } & \sum_{t=1}^{T} x_{k, t} d^O_{ t} \leq T \rho_{k}, \forall k \in \calK.
 \end{aligned}
 \end{equation}

 Introduce Lagrangian multipliers for each budget constraint. The Lagrangian of optimization problem \eqref{eqn: hindsight optimal welfare} is
\begin{align*}
L(\bv,\dd^O;\bmu)&=\sum^T_{t=1}\xx^\top_t(\bv_t-d^O_{t}\bm{e})+	\bmu^\top\sum^T_{t=1}\left(-d^O_{t}\xx_t+ T \boldsymbol{\rho}\right)\\
&=\sum^T_{t=1}\xx^\top_t(\bv_t-d^O_{t}\left(\bm{e}+\bmu\right))+\brho^\top\bmu. 
\end{align*}
where $\bm{e}$ is all-one vector. 

The dual problem is therefore given by
\begin{align*}
\inf_{\bmu\geq 0}\sup_{\xx} L(\bmu,\xx)=\inf_{\bmu \geq 0}\left\{\max_{k\in\calK}\left(v_{k,t}-(1+\mu_k)d^O_t\right)^{+}+\brho^\top\bmu\right\} \eqqcolon \inf_{\bmu\geq 0}\phi(\bmu),
\end{align*}
which serves an upper bound of $\pi^{\rm{H}}$. Take expectations on both sides,
\begin{align*}
\bbE\left[\pi^{\rm{H}}\right]\leq \bbE\left[\inf_{\bmu\geq0}\phi(\bmu)\right]\leq \inf_{\bmu\geq 0}\bbE\left[\phi(\bmu)\right]=\inf_{\bmu \geq 0}\left\{\bbE\left[\max_{k\in\calK}\left(v_{k,t}-(1+\mu_k)d^O_t\right)^+\right]+\brho^\top\bmu\right\} \eqqcolon \inf_{\bmu \geq 0}\Phi(\bmu). 
\end{align*}

Denote the expected expenditure per auction of each bidder $k$ to be 
\[
\EXPSumCost(\bmu) = \bbE\left[\bm{1}\left\{k\in \argmax_{k'\in\calK}\left(v_{k',t}-(1+\mu_{k'})d^O_t\right)\right\}\bm{1}\left\{v_{k,t}\geq(1+\mu_{k})d^O_t\right\} d^O_t\right]. 
\]
Then by the Karush-Kuhn-Tucker condition, the optimal dual solution $\bmu^{**}$satisfies the complementarity condition, which is 
\[
\EXPSumCost(\bmu^{**})\leq \rho_k \perp \mu^{**}_k\geq 0, \forall{k \in \calK}
\]
and the corresponding expected utility of bidder $k$ is 
\[
\EXPSumUtility(\bmu^{**}) = \bbE\left[\bm{1}\left\{k\in \argmax_{k'\in\calK}\left(v_{k',t}-(1+\mu^{**}_{k'})d^O_t\right)\right\}\bm{1}\left\{v_{k',t}\geq(1+\mu^{**}_{k'})d^O_t\right\}\left(v_{k',t}- d^O_t\right)\right]. 
\]
When every bidder has the same distribution and budget, we have $\mu^{**}_k = \mu^{**}_1$ for any $k \in \calK$, and $\xi^*_k=\xi^*_1$, where $\bxi^*$ are the equilibria parameters defined by \eqref{eqn: xi definition}. Recall the forms of $\EXPcostCP$ and $\EXPutilityCP$ in this case,
\begin{equation}
\label{eqn: symmetric CP}
\begin{aligned}
    \EXPcostCP(\bxi^*) &= \bbE\left[\prod_{k' \in \calK}\bm{1}\left\{v_{k,t}\geq v_{k',t}\right\}\bm{1}\left\{v_{k,t}\geq (1+\xi^*_k)d^O_t\right\} d^O_t\right],\\
    \EXPutilityCP(\bxi^*) &=\bbE\left[\prod_{k' \in \calK}\bm{1}\left\{v_{k,t}\geq v_{k',t}\right\}\bm{1}\left\{v_{k,t}\geq (1+\xi^*_k)d^O_t\right\} \left(v_{k,t}-d^O_t\right)\right]. 
\end{aligned}
\end{equation}
We show $\EXPSumCost(\bmu)$ and $\EXPSumUtility(\bmu)$ have the same form as \eqref{eqn: symmetric CP}, when bidders are homogeneous. Then by the uniqueness of $\bxi^*$, we have $\bxi^*=\bmu^{**}$. Moreover, we can conclude that CP maximizes the sum of bidders' utilities in expectation. 

To show the two expected expenditures per auction and the total expected utilities  have the same form, it suffices to show that the sets where the statements of the indicator functions hold are equal. Specifically, notice that the set $\left\{k \in \argmax_{k'\in \calK}\left(v_{k',t}-\left(1+\mu_{k'}\right)\right)\right\}$ can be rewritten as $\cap_{k'\in\calK}\left\{v_{k,t}-(1+\mu_k)d^O_t\geq v_{k',t}-(1+\mu_{k'})d^O_t\right\}$. Moreover, when $\mu_k=\mu_1$,

\[v_{k,t}-(1+\mu_k)d^O_t\geq v_{k',t}-(1+\mu_{k'})d^O_t \Leftrightarrow v_{k,t}\geq v_{k',t}.\]
This completes the proof for CP. To show the results of HP, just notice that at equilibrium, the expected expenditures per auction and the total expected utilities of HP has the same form as those of CP. 
\end{proof}

\subsection{Proof of \Cref{thm:misreport-budget-value-cp}}

\thmmisreportbudgetvaluecp*

\begin{proof}
Let the expenditure rate of a member be $\rho$. By symmetry, all members' equilibrium parameters are the same: $\bxi^*=\xi^*\bm{e}$ for some $\xi^*\in[0,\bar{\xi}]$. When bidder $k$ misreports a rate $\rho'<\rho$, denote her corresponding equilibrium parameter to be $\xi^*_k$ and that of other member $i$ to be $\xi^*_1$. By a slight abuse of notation, we let $G_i^\rmCP$ be the function of $(\xi_1,\xi_k)$, the first parameter for other bidders, and the second for bidder $k$. Without loss of generality, let $k\ne 1$, and we only consider the truthful bidder $1$ as other truthful bidders behave the same.  We  write down the form of $G^{\rmCP}_i$ and $U^{\rmCP}_i$: 
\begin{align}
G^{\rmCP}_k(\xi^*,\xi^*) &= \bbE\left[\bm{1}\left\{v_k\geq \max_{i\ne k}v_i\right\}\bm{1}\left\{v_k\geq (1+\xi^*)d^O\right\}d^O\right] = G^{\rmCP}_1(\xi^*,\xi^*) \leq \rho  , \\
G^{\rmCP}_1(\xi^*_1,\xi^*_k) &=\bbE\left[\Pi_{\substack{i\ne 1,\\i\ne k}}\bm{1}\left\{v_1\geq v_i\right\}\bm{1}\left\{v_1\geq \frac{1+\xi^*_1}{1+\xi^*_k}v_k\right\}\bm{1}\left\{v_1 \geq (1+\xi^*_1)d^O\right\}d^O\right] \leq \rho,\\
G^{\rmCP}_k(\xi^*_1,\xi^*_k) &= \bbE\left[\bm{1}\left\{v_k\geq \frac{1+\xi^*_k}{1+\xi^*_1}\max_{i\ne k}v_i\right\}\bm{1}\left\{v_k \geq (1+\xi^*_k)d^O\right\}d^O\right]\leq\rho' <\rho,
\end{align}
and,
\begin{align}
U^{\rmCP}_k(\xi^*,\xi^*) &= \bbE\left[\bm{1}\left\{v_k\geq \max_{i\ne k}v_i\right\}\bm{1}\left\{v_k\geq (1+\xi^*)d^O\right\}\left(v_k-d^O\right)\right]=U^{\rmCP}_1(\xi^*,\xi^*), \\
U^{\rmCP}_1(\xi^*_1,\xi^*_k) &=\bbE\left[\Pi_{\substack{i\ne 1,\\i\ne k}}\bm{1}\left\{v_1\geq v_i\right\}\bm{1}\left\{v_1\geq \frac{1+\xi^*_1}{1+\xi^*_k}v_k\right\}\bm{1}\left\{v_1 \geq (1+\xi^*_1)d^O\right\}\left(v_1-d^O\right)\right],\\
U^{\rmCP}_k(\xi^*_1,\xi^*_k) &= \bbE\left[\bm{1}\left\{v_k\geq \frac{1+\xi^*_k}{1+\xi^*_1}\max_{i\ne k}v_i\right\}\bm{1}\left\{v_k \geq (1+\xi^*_k)d^O\right\}\left(v_k-d^O\right)\right].
\end{align}

We consider all the possible cases the solutions of NCPs may be. We show that for the only reasonable cases that do not lead to contradiction with the definition of NCP and strong monotonicity, we have
\begin{equation}
\label{eqn: condition 1}
\begin{aligned}
\left\{v_k\geq \frac{1+\xi^*_k}{1+\xi^*_1}\max_{i\ne k}v_i\right\} &\subseteq \left\{v_k\geq \max_{i\ne k}v_i\right\}, \\
\left\{v_k\geq (1+\xi^*_k)d^O\right\} &\subseteq \left\{v_k\geq (1+\xi^*)d^O\right\}, 
\end{aligned}
\end{equation}
therefore $U^{\rmCP}_k(\xi^*_1,\xi^*_k) \leq U^{\rmCP}_k(\xi^*,\xi^*)$. 

\paragraph{When $\xi^*=\xi^*_1=\xi^*_k=0$,} all the members have sufficient budgets, after misreporting, bidder $k$'s budget is still big enough. In this case, her obtained utility will not decrease. 

\paragraph{When $\xi^*\ne 0$,} then $\EXPcostCP(\xi^*,\xi^*)=\rho$. First we have $\xi^*_1\ne \xi^*_k$: 
\begin{itemize}
    \item for $\xi^*_1 = \xi^*_k > \xi^*$, we have $\rho = G^{\rmCP}_1(\xi^*_1,\xi^*_k)< \EXPcostCP(\xi^*,\xi^*) = \rho$, leading to a contradiction;
    \item for $\xi^*_1 = \xi^*_k < \xi^*$, we have $\rho' \geq G^{\rmCP}_k(\xi^*_1,\xi^*_k)=G^{\rmCP}_1(\xi^*_1,\xi^*_k)>\EXPcostCP(\xi^*,\xi^*) = \rho$, leading to a contradiction. 
\end{itemize}
Then we argue  that $\xi^*_k \ne 0$, as otherwise it must be the case that $\xi^*_1\ne 0$ and 
\begin{align*}
\EXPcostCP(\xi^*,\xi^*) = G^{\rmCP}_1(\xi^*_1,\xi^*_k)=\rho,\\
G^{\rmCP}_k(\xi^*_1,\xi^*_k)\leq \rho' < \rho. 
\end{align*}
By the strong monotonicity of $\bm{G}^{\rmCP}$, we have
\[
(\xi^*_k-\xi^*)(\EXPcostCP(\xi^*,\xi^*)-G^{\rmCP}_k(\xi^*_1,\xi^*_k))> 0, 
\]
and
\[\xi^*_k > \xi^*,\]
leading to a contradiction. 

When $\xi^*_1 = 0$, it must be the case that  
\[
\xi^*_k > \xi^*_1=0, \text{~and~} \xi^* \geq \xi^*_1. 
\]
and, 
\begin{align*}
G^{\rmCP}_1(\xi^*_1,\xi^*_k)\leq \EXPcostCP(\xi^*,\xi^*) =\rho,\\
G^{\rmCP}_k(\xi^*_1,\xi^*_k) = \rho' < \rho. 
\end{align*}
By the strong monotonicity of $\bm{G}^{\rmCP}$, 
\[
(\xi^*_k - \xi^*)(G^{\rmCP}_k(\xi^*,\xi^*)-G^{\rmCP}_k(\xi^*_1,\xi^*_k))> (K-1)(\xi^*-\xi^*_1)(G^{\rmCP}_1(\xi^*,\xi^*)-G_1(\xi^*_1,\xi^*_k))\geq 0,
\]
which means $\xi^*_k > \xi^*$. 
In this case, 
\begin{align*}
    \left\{v_k\geq \frac{1+\xi^*_k}{1+\xi^*_1}\max_{i\ne k}v_i\right\} &\subseteq \left\{v_k\geq \max_{i\ne k}v_i\right\},\\ 
    \left\{v_k\geq (1+\xi^*_k)d^O\right\} &\subseteq \left\{v_k\geq (1+\xi^*)d^O\right\}. 
\end{align*}
\eqref{eqn: condition 1} holds. 

Finally we discuss the case when $\xi^*,\xi^*_1,\xi^*_k >0$, and 
\begin{align}
G^{\rmCP}_1(\xi^*_1,\xi^*_k) &= \EXPcostCP(\xi^*,\xi^*) =\rho,\\
G^{\rmCP}_k(\xi^*_1,\xi^*_k) &= \rho' < \rho. 
\end{align}

First by the strong monotonicity of $\bm{G}^{\rmCP}$, we have $\xi^*_k > \xi^*$. Then consider the following possible cases:
\begin{itemize}
    \item $\xi^*_1 \geq \xi^*_k > \xi^*$: $(1+\xi^*_1)/(1+\xi^*_k) > 1$, and 
\begin{equation}
\label{eqn: contradiction 2}
\begin{aligned}
    \left\{v_1\geq \frac{1+\xi^*_1}{1+\xi^*_k}v_k\right\} &\subseteq \left\{v_1\geq v_k\right\},\\ 
    \left\{v_1\geq (1+\xi^*_1)d^O\right\} &\subseteq \left\{v_1\geq (1+\xi^*)d^O\right\},
\end{aligned}
\end{equation}
which means that $G_1^{\rm{CP}}(\xi^*_1,\xi^*_k)\leq G_1^{\rm{CP}}(\xi^*,\xi^*)$. 
We state the following lemma:
\begin{restatable}{lemma}{lmmcontradictionone}
\label{lmm: contradiction 1}
When $\bm{G}^{\rmCP}$ is strongly monotone, then for any $\xi^*_1 > \xi^*$ and $(1+\xi^*_1)/(1+\xi^*_k)>1$,
we have  
\[
G^{\rmCP}_1(\xi^*_1,\xi^*_k) < G^{\rmCP}_1(\xi^*,\xi^*). 
\]
\end{restatable}
By \Cref{lmm: contradiction 1}, $G^{\rmCP}_1(\xi^*_1,\xi^*_k) < G^{\rmCP}_1(\xi^*,\xi^*)$, leading to a contradiction. 
    \item $\xi^*_k > \xi^*_1$: $(1+\xi^*_k)/(1+\xi^*_1) > 1$, and
\begin{equation} 
\label{eqn: reasonable 1}
\begin{aligned}
    \left\{v_k\geq \frac{1+\xi^*_k}{1+\xi^*_1}\max_{i\ne k}v_i\right\} &\subseteq \left\{v_k\geq \max_{\substack{i\ne k}}v_i\right\},\\ 
    \left\{v_k\geq (1+\xi^*_k)d^O\right\} &\subseteq \left\{v_k\geq (1+\xi^*)d^O\right\}.
\end{aligned}
\end{equation}
\eqref{eqn: condition 1} holds. 
\end{itemize}

\paragraph{When $\xi^*=0$, and one of $\xi^*_1, \xi^*_k$ is non-zero,} first we argue that $\xi^*_1 \ne \xi^*_k$, as it contradicts that $G^{\rmCP}_1(\xi^*_1,\xi^*_k) \ne G^{\rmCP}_k(\xi^*_1,\xi^*_k)$. 

If $\xi^*_k = 0$, then $G^{\rmCP}_1(\xi^*_1,\xi^*_k)=\rho \geq \EXPcostCP(\xi^*,\xi^*)$, and $\xi^*_1 > \xi^* =\xi^*_k =0$, we have $(1+\xi^*_1)/(1+\xi^*_k) > 1$, and \eqref{eqn: contradiction 2} holds. 
By \Cref{lmm: contradiction 1}, $G^{\rmCP}_1(\xi^*_1,\xi^*_k) < G^{\rmCP}_1(\xi^*,\xi^*)$, leading to a contradiction. 

If $\xi^*_1 = 0$, then $\EXPcostCP(\xi^*_1,\xi^*_k)=\rho'$, and $\xi^*_k > \xi^*_1 = \xi^*$. In this case, $(1+\xi^*_1)/(1+\xi^*_k) < 1$, \eqref{eqn: reasonable 1} and therefore \eqref{eqn: condition 1} holds.

When $\xi^*_1, \xi^*_k \ne 0$, consider the two cases:
\begin{itemize}
    \item if $\xi^*_1 > \xi^*_k > \xi^* = 0$, then $(1+\xi^*_1)/(1+\xi^*_k) > 1$, \eqref{eqn: contradiction 2} holds. 
By \Cref{lmm: contradiction 1},  $G^{\rmCP}_1(\xi^*_1,\xi^*_k) < G^{\rmCP}_1(\xi^*,\xi^*) \leq \rho$, leading to a contradiction. 
    \item if $\xi^*_k > \xi^*_1 >\xi^* = 0$, then $(1+\xi^*_1)/(1+\xi^*_k) < 1$, \eqref{eqn: reasonable 1} and therefore \eqref{eqn: condition 1} holds. 
\end{itemize}

\end{proof}

\subsection{Proof of \Cref{thm:misreport-budget-value-ip}}

\thmmisreportbudgetvalueip*

\begin{proof}
Let the expenditure rate of a member be $\rho$. By symmetry, all members' equilibrium parameters are the same: $\blambda^*=\lambda^*\bm{e}$ for some $\lambda^*\in[0,\bar{\lambda}]$. When bidder $k$ misreports a rate $\rho'<\rho$, denote her corresponding equilibrium parameter to be $\lambda^*_k$ and that of another member $i$ to be $\lambda^*_1$. By a slight abuse of notation, we let $G_i^\rmCP$ be the function of $(\lambda_1,\lambda_k)$, the first parameter for other bidders, and the second for bidder $k$. Without loss of generality, let $k\ne 1$, and we only consider the truthful bidder $1$ as other truthful bidders behave the same.  We  write down the form of $G_i$ and $V_i$: 
\begin{align*}
G_k(\lambda^*,\lambda^*) &= \bbE\left[\bm{1}\left\{v_k\geq \max_{i\ne k}v_i\right\}\bm{1}\left\{v_k\geq (1+\lambda^*)d^O\right\}\max \left\{\frac{\max_{i\ne k}v_i}{1+\lambda^*},d^O\right\}\right]=G_1(\lambda^*,\lambda^*) \leq \rho, \\
G_1(\lambda^*_1,\lambda^*_k) &=\bbE\left[\Pi_{\substack{i\ne 1,\\i\ne k}}\bm{1}\left\{v_1\geq v_i\right\}\bm{1}\left\{v_1\geq \frac{1+\lambda^*_1}{1+\lambda^*_k}v_k\right\}\bm{1}\left\{v_1 \geq (1+\lambda^*_1)d^O\right\} \max \left\{\frac{\max_{\substack{i\ne 1\\i\ne k}}v_i}{1+\lambda^*_1}, \frac{v_k}{1+\lambda^*_k}, d^O\right\}\right] \leq \rho,\\
G_k(\lambda^*_1,\lambda^*_k) &= \bbE\left[\bm{1}\left\{v_k\geq \frac{1+\lambda^*_k}{1+\lambda^*_1}\max_{i\ne k}v_i\right\}\bm{1}\left\{v_k \geq (1+\lambda^*_k)d^O\right\}\max \left\{\frac{\max_{i\ne k}v_i}{1+\lambda^*_1},d^O\right\}\right]=\rho' <\rho,
\end{align*}
and,
\begin{align*}
V_k(\lambda^*,\lambda^*) &= \bbE\left[\bm{1}\left\{v_k\geq \max_{i\ne k}v_i\right\}\bm{1}\left\{v_k\geq (1+\lambda^*)d^O\right\}v_k\right]=V_1(\lambda^*,\lambda^*), \\
V_1(\lambda^*_1,\lambda^*_k) &=\bbE\left[\Pi_{\substack{i\ne 1,\\i\ne k}}\bm{1}\left\{v_1\geq v_i\right\}\bm{1}\left\{v_1\geq \frac{1+\lambda^*_1}{1+\lambda^*_k}v_k\right\}\bm{1}\left\{v_1 \geq (1+\lambda^*_1)d^O\right\} v_1\right],\\
V_k(\lambda^*_1,\lambda^*_k) &= \bbE\left[\bm{1}\left\{v_k\geq \frac{1+\lambda^*_k}{1+\lambda^*_1}\max_{i\ne k}v_i\right\}\bm{1}\left\{v_k \geq (1+\lambda^*_k)d^O\right\}v_k\right].
\end{align*}

We consider all the possible cases the solutions of NCPs may be. We show that for the only reasonable cases that do not contradict the definition of NCP and strong monotonicity, we have
\begin{equation}
\label{eqn: condition 2}
\begin{aligned}
\left\{v_k\geq \frac{1+\lambda^*_k}{1+\lambda^*_1}\max_{i\ne k}v_i\right\} &\subseteq \left\{v_k\geq \max_{i\ne k}v_i\right\}, \\
\left\{v_k\geq (1+\lambda^*_k)d^O\right\} &\subseteq \left\{v_k\geq (1+\lambda^*)d^O\right\}, 
\end{aligned}
\end{equation}
therefore $V_k(\lambda^*_1,\lambda^*_k) \leq V_k(\lambda^*,\lambda^*)$. 

\paragraph{When $\lambda^*=\lambda^*_1=\lambda^*_k=0$,} all the members have sufficient budgets, after misreporting, bidder $k$'s budget is still big enough. In this case, her obtained value will not decrease. 

\paragraph{When $\lambda^*\ne 0$,} then $\EXPcost(\lambda^*,\lambda^*)=\rho$. First we have $\lambda^*_1\ne \lambda^*_k$: 
\begin{itemize}
    \item for $\lambda^*_1 = \lambda^*_k > \lambda^*$, we have $\rho = G_1(\lambda^*_1,\lambda^*_k)< \EXPcost(\lambda^*,\lambda^*) = \rho$, leading to a contradiction;
    \item for $\lambda^*_1 = \lambda^*_k < \lambda^*$, we have $\rho' \geq G_k(\lambda^*_1,\lambda^*_k)=G_1(\lambda^*_1,\lambda^*_k)>\EXPcost(\lambda^*,\lambda^*) = \rho$, leading to a contradiction. 
\end{itemize}
Then we argue  that $\lambda^*_k \ne 0$, as otherwise it must be the case that $\lambda^*_1\ne 0$ and 
\begin{align*}
\EXPcost(\lambda^*,\lambda^*) = G_1(\lambda^*_1,\lambda^*_k)=\rho,\\
G_k(\lambda^*_1,\lambda^*_k)\leq \rho' < \rho. 
\end{align*}
By the strong monotonicity of $\bm{G}$, we have
\[
(\lambda^*_k-\lambda^*)(\EXPcost(\lambda^*,\lambda^*)-G_k(\lambda^*_1,\lambda^*_k))> 0, 
\]
and
\[\lambda^*_k > \lambda^*,\]
leading to a contradiction. 

When $\lambda^*_1 = 0$, it must be the case that  
\[
\lambda^*_k > \lambda^*_1=0, \text{~and~} \lambda^* \geq \lambda^*_1. 
\]
By the strong monotonicity of $\bm{G}$, 
\[
(\lambda^*_k - \lambda^*)(G_k(\lambda^*,\lambda^*)-G_k(\lambda^*_1,\lambda^*_k))> (\lambda^*-\lambda^*_1)(G_1(\lambda^*,\lambda^*)-G_1(\lambda^*_1,\lambda^*_k))\geq 0,
\]
which means $\lambda^*_k > \lambda^*$. 
In this case, 
\begin{align*}
    \left\{v_k\geq \frac{1+\lambda^*_k}{1+\lambda^*_1}\max_{i\ne k}v_i\right\} &\subseteq \left\{v_k\geq \max_{i\ne k}v_i\right\},\\ 
    \left\{v_k\geq (1+\lambda^*_k)d^O\right\} &\subseteq \left\{v_k\geq (1+\lambda^*)d^O\right\}. 
\end{align*}
\eqref{eqn: condition 2} holds. 

Finally we discuss the case when $\lambda^*,\lambda^*_1,\lambda^*_k >0$, and 
\begin{align}
G_1(\lambda^*_1,\lambda^*_k) &= \EXPcost(\lambda^*,\lambda^*) =\rho,\\
G_k(\lambda^*_1,\lambda^*_k) &= \rho' < \rho. 
\end{align}

First by strong monotonicity of $\bm{G}$, we have $\lambda^*_k > \lambda^*$. Then consider the following possible cases:
\begin{itemize}
    \item $\lambda^*_1 \geq \lambda^*_k > \lambda^*$: $(1+\lambda^*_1)/(1+\lambda^*_k) > 1$, and 
\begin{equation}
\label{eqn: contradiction 3}
\begin{aligned}
    \left\{v_1\geq \frac{1+\lambda^*_1}{1+\lambda^*_k}v_k\right\} &\subseteq \left\{v_1\geq v_k\right\},\\ 
    \left\{v_1\geq (1+\lambda^*_1)d^O\right\} &\subseteq \left\{v_1\geq (1+\lambda^*)d^O\right\},\\
    \frac{\max_{\substack{i\ne 1\\i\ne k}}v_i}{1+\lambda^*_1} &< \frac{\max_{\substack{i\ne 1\\i\ne k}}v_i}{1+\lambda^*},\\
    \frac{v_k}{1+\lambda^*_k} &< \frac{v_k}{1+\lambda^*}.
\end{aligned}
\end{equation}
$G_1(\lambda^*_1,\lambda^*_k) \leq G_1(\lambda^*,\lambda^*)$. We prove a similar lemma as \Cref{lmm: contradiction 1},
\begin{lemma}
\label{lmm: contradiction 2}
When $\bm{G}$ is strongly monotone, then for any $\lambda^*_1 > \lambda^*$ and $(1+\lambda^*_1)/(1+\lambda^*_k)>1$, we have 
\[
G_1(\lambda^*_1,\lambda^*_k) < G_1(\lambda^*,\lambda^*)
\]
\end{lemma}
By \Cref{eqn: contradiction 3} and \Cref{lmm: contradiction 2}, 
$G_1(\lambda^*_1,\lambda^*_k) < G_1(\lambda^*,\lambda^*)$, leading to a contradiction. 
    \item $\lambda^*_k > \lambda^*_1$: $(1+\lambda^*_k)/(1+\lambda^*_1) > 1$, and 
\begin{align*}
    \left\{v_k\geq \frac{1+\lambda^*_k}{1+\lambda^*_1}\max_{i\ne k}v_i\right\} &\subseteq \left\{v_k\geq \max_{\substack{i\ne k}}v_i\right\},\\ 
    \left\{v_k\geq (1+\lambda^*_k)d^O\right\} &\subseteq \left\{v_k\geq (1+\lambda^*)d^O\right\}.
\end{align*}
\eqref{eqn: condition 2} holds. 
\end{itemize}

\paragraph{When $\lambda^*=0$, and one of $\lambda^*_1, \lambda^*_k$ is non-zero,} first we argue that $\lambda^*_1 \ne \lambda^*_k$, as it contradicts that $G_1(\lambda^*_1,\lambda^*_k) \ne G_k(\lambda^*_1,\lambda^*_k)$. 

If $\lambda^*_k = 0$, then $G_1(\lambda^*_1,\lambda^*_k)=\rho \geq G_1(\lambda^*,\lambda^*)$, and $\lambda^*_1 > \lambda^* =\lambda^*_k =0$, we have $(1+\lambda^*_1)/(1+\lambda^*_k) > 1$, and
\begin{align*}
    \left\{v_1\geq \frac{1+\lambda^*_1}{1+\lambda^*_k}v_k\right\} &\subseteq \left\{v_1\geq v_k\right\},\\ 
    \left\{v_1\geq (1+\lambda^*_1)d^O\right\} &\subseteq \left\{v_1\geq (1+\lambda^*)d^O\right\},\\
    \frac{\max_{\substack{i\ne 1\\i\ne k}}v_i}{1+\lambda^*_1} &< \frac{\max_{\substack{i\ne 1\\i\ne k}}v_i}{1+\lambda^*},\\
    \frac{v_k}{1+\lambda^*_k} &\leq \frac{v_k}{1+\lambda^*}.
\end{align*}
By \Cref{lmm: contradiction 2}, 
$G_1(\lambda^*_1,\lambda^*_k) < G_1(\lambda^*,\lambda^*)$, leading to a contradiction. 

If $\lambda^*_1 = 0$, then $\EXPcost(\lambda^*_1,\lambda^*_k)=\rho'$, and $\lambda^*_k > \lambda^*_1 = \lambda^*$. In this case, $(1+\lambda^*_1)/(1+\lambda^*_k) < 1$, and \eqref{eqn: condition 2} holds.

When $\lambda^*_1, \lambda^*_k \ne 0$, consider the two cases:
\begin{itemize}
    \item if $\lambda^*_1 > \lambda^*_k > \lambda^* = 0$, then \eqref{eqn: contradiction 3} holds,   leading to a contradiction. 
    \item if $\lambda^*_k > \lambda^*_1 >\lambda^* = 0$, then $(1+\lambda^*_1)/(1+\lambda^*_k) < 1$, and \eqref{eqn: condition 2} holds. 
\end{itemize}

\end{proof}

\subsection{An Experimental Example that IP is Not Truthful on Utilities}

Consider a coalition of two bidders. Consider independent random varaibles $X_1,X_2$, distributed uniformly in $[0,1]$, and $Y_1,Y_2$, distributed uniformly in $[1,2]$, respectively. For $i\in{1,2}$, the random variable of bidder $i$'s value is $Z_i=1/3X_i+2/3Y_i$. The competing bid outside the coalition is uniformly distributed in $[0,1]$. We set the true expenditure ratio of each bidder to be $0.5$, and let bidder $1$ misreports a ratio of $0.49$. 

We run IP the same as we run the example in \Cref{sec: asymmetric example}. And we repeated the experiment $100$ times. For every experiment, all bidders have higher utilities when bidder $1$'s misreports: bidder $1$'s increases by about $0.003$, and bidder $2$'s by about $0.006$. 

As we further check the changes of bidders' obtained values and the equilibrium parameters, we can find out that bidder $1$'s obtained value decreases by about $0.007$. Denote the equilibrium parameter when bidder $1$ truthfully reports by $\mu^*$, the equilibrium parameter of bidder $1$'s by $\mu^*_1$ and of bidder $2$'s by $\mu^*_2$, respectively, when bidder $1$ misreports. It shows that 
$\mu^*_1>\mu^*_2>\mu^*$, which is consistent with the analysis above. 

\subsection{Proofs of Lemmas in Proving \Cref{thm:misreport-budget-value-cp} and \ref{thm:misreport-budget-value-ip}}

We give the proof of \Cref{lmm: contradiction 1}, as  the proof of \Cref{lmm: contradiction 2} is similar. 

\lmmcontradictionone*

\begin{proof}
We prove this by contradiction. Suppose $G^{\rmCP}_1(\xi^*_1,\xi^*_k)=G^{\rmCP}_1(\xi^*,\xi^*)$, then 
\begin{equation*}
\begin{aligned}
0=&G^{\rmCP}_1(\xi^*,\xi^*) - G^{\rmCP}_1(\xi^*_1,\xi^*_k) = \bbE\left[\prod_{\substack{i\in \calK, \\i\ne 1,i\ne k}}\mathbf{1}\left\{v_1\geq v_i\right\}\right.\\
\cdot&\left.\mathbf{1}\left\{\left(\left(v_1\geq v_k\right)\wedge\left((1+\xi^*)d^O\leq v_1 < (1+\xi^*_1)d^O\right)\right) \vee \left(\left(v_k\leq v_1 <\frac{1+\xi_1^*}{1+\xi^*_k}v_k\right)\wedge \left(v_1\geq (1+\xi^*)d^O\right)\right)\right\}\right],
\end{aligned}
\end{equation*}
which means 
\begin{equation}
\label{eqn: universe set}
\begin{aligned}
\bigcap_{\substack{i\in \calK, \\i\ne 1,i\ne k}}\left\{v_1\geq v_i\right\}\bigcap&
\left(
\left(
\left\{v_1\geq v_k\right\}\bigcap\left\{(1+\xi^*)d^O\leq v_1 < (1+\xi^*_1)d^O\right\}\right)\right.\\
& \phantom{((}\left.\bigcup  \left(\left\{v_k\leq v_1 <\frac{1+\xi_1^*}{1+\xi^*_k}v_k\right\}\bigcap \left\{v_1\geq (1+\xi^*)d^O\right\}\right)\right)
\end{aligned}
\end{equation}
has zero measure, then for $\xi\in (\xi^*,\xi^*_1)$, 
the set
\begin{equation*}
\begin{aligned}
\bigcap_{\substack{i\in \calK, \\i\ne 1,i\ne k}}\left\{v_1\geq v_i\right\}\bigcap&
\left(
\left(
\left\{v_1\geq v_k\right\}\bigcap\left\{(1+\xi)d^O\leq v_1 < (1+\xi^*_1)d^O\right\}\right)\right.\\
& \phantom{((}\left.\bigcup  \left(\left\{v_k\leq v_1 <\frac{1+\xi_1^*}{1+\xi^*_k}v_k\right\}\bigcap \left\{v_1\geq (1+\xi)d^O\right\}\right)\right)
\end{aligned}
\end{equation*}
is a subset of \eqref{eqn: universe set}, and therefore has measure zero. This means that
\[
G^{\rmCP}_1(\xi^*,\xi^*)-G^{\rmCP}_1(\xi^*_1,\xi^*_k) = 0=G^{\rmCP}_1(\xi,\xi)-G^{\rmCP}_1(\xi^*_1,\xi^*_k), 
\]
equivalent to
\[
G^{\rmCP}_1(\xi^*,\xi^*) = G^{\rmCP}_1(\xi,\xi). 
\]

This contradicts that $\bm{G}^{\rmCP}$ is strongly monotone. 
\end{proof}

\section{Auxiliary Results}

\begin{lemma}
\label{lem: lipschitz cp}
$\EXPcost(\blambda)$, $\EXPvalue(\blambda)$, $\EXPcostCP(\bxi)$ are Lipschitz continuous. 
\end{lemma}
\begin{proof}

We denote by $L_k$ the cumulative distribution function of $d^I_k$:
\begin{align}
\label{eqn: cdf of d i}
L_k(x; \blambda) = \prod_{i\in \calK: i\ne k}F_i((1+\lambda_i)x).
\end{align}
We omit $\blambda$ when there is no confusion. In what follows, we denote by $\overline{f}$  the upper bound on $f_k$ and by $\overline{h}$  the upper bound on $h_k$. We use $\bar{F}_k=1-F_k$ to denote the complementary cumulative distribution function of $v_k$.

\begin{enumerate}
    \item We write down two expressions of $\EXPcost(\blambda)$ and calculate its derivatives with respect to $\lambda_k$ and $\lambda_i$ for $i\ne k$.
    \begin{align}
        \EXPcost(\blambda) &= \int x \bar{F}_k((1+\lambda_k)x) \,d\left(L_k(x)H(x)\right) \label{eqn: G first expression} \\
        &=\int L_k(x)H(x)\left((1+\lambda_k)x f_k((1+\lambda_k)x)-\bar{F}((1+\lambda_k)x)\right)\,dx. \label{eqn: G second expression} 
    \end{align}
    Using the first expression \eqref{eqn: G first expression}, we have
    \begin{align*}
    \frac{\partial \EXPcost(\blambda)}{\partial \lambda_k} = -\int x^2 f((1+\mu_k)x)\,d(L_k(x)H(x)),
    \end{align*}
    and thus,
    \begin{align*}
        \left|\frac{\partial \EXPcost(\blambda)}{\partial \lambda_k}\right|\leq \bar{v}\bar{f}. 
    \end{align*}
    Using the second expression \eqref{eqn: G second expression}, we have for $i\neq k$, 
    \begin{align*}
        \frac{\partial \EXPcost(\blambda)}{\partial \lambda_i} = \int \frac{t}{\left(1+\lambda_k\right)^2}f_i\left(\frac{1+\lambda_i}{1+\lambda_k}t\right)\prod_{j\ne i,k} F_j\left(\frac{1+\lambda_j}{1+\lambda_k}t\right)H\left(\frac{t}{1+\lambda_k}\right)\left(tf_{k}(t)-\bar{F}_k(t)\right)\,dt,
    \end{align*}
    and thus,
    \begin{align*}
        \left|\frac{\partial \EXPcost(\blambda)}{\partial \lambda_i}\right|\leq \bar{v}\bar{f}(\bar{v}\bar{f}+1).  
    \end{align*}

    \item We write down the expression of $\EXPvalue(\blambda)$ and calculate its derivatives with respect to $\lambda_k$ and $\lambda_i$ for $i\ne k$.
    \begin{align*}
        \EXPvalue(\blambda)=\int x H\left(\frac{x}{1+\lambda_k}\right)L_k\left(\frac{x}{1+\lambda_k}\right)\,dF_k(x).
    \end{align*}
    Therefore
    \begin{align*}
        \frac{\partial \EXPvalue(\blambda)}{\partial \lambda_k}=-\int\left(\frac{x}{1+\lambda_k}\right)^2\left(h\left(\frac{x}{1+\lambda_k}\right)L_k\left(\frac{x}{1+\lambda_k}\right)+H\left(\frac{x}{1+\lambda_k}\right)l_k\left(\frac{x}{1+\lambda_k}\right)\right)\,dF_k(x),
    \end{align*}
    where $l_k$ is the density function of $L_k$. 
    Since 
    \[
    l_k\left(x\right) = \sum_{i\ne k}(1+\lambda_i) f_i\left((1+\lambda_i)x\right)\prod_{j\ne i,k}F_j((1+\lambda_i)x), 
    \]
    we have
    \[
    \left|\frac{\partial \EXPvalue(\blambda)}{\partial \lambda_k}\right| \leq \bar{v}^2 \bar{h} + (K-1)\bar{v}^2\bar{f}.
    \]
    And for $i\neq k$, 
    \begin{align*}
        \frac{\partial \EXPvalue(\blambda)}{\partial \lambda_i} = \int \frac{x^2}{1+\lambda_k} H\left(\frac{x}{1+\lambda_k}\right)f_i\left(\frac{1+\lambda_i}{1+\lambda_k}x\right)\prod_{j\ne i,k}F_j\left(\frac{1+\lambda_j}{1+\lambda_k}x\right)\, dF_k(x),
    \end{align*}
    \[
    \left| \frac{\partial \EXPvalue(\blambda)}{\partial \lambda_k}\right| \leq \bar{v}^2 \bar{f}. 
    \]

    \item We write down the expression of $\EXPcostCP(\bxi)$ and calculate its derivatives with respect to $\xi_k$ and $\xi_i$ for $i\ne k$.
    \begin{align*}
    \EXPcostCP(\bxi) &= \int \int_{x\geq y} \bar{F}_k((1+\xi_k)x)y\,dL_k(x)\,dH(y)+\int \bar{F}_k((1+\xi_k)y)L_k(y)y\,dH(y)\\
    &=\int y \int_{x\geq y} (1+\xi_k)L_k(x) f_k((1+\xi_k)x)\,dx\,dH(y), 
    \end{align*}Therefore, we have
    \begin{align*}
        \frac{\partial \EXPcostCP(\bxi)}{\partial \xi_k} = -\int \int_{x \geq y}x f_k((1+\xi_k)x)y \, dL_k(x)\,dH(y) - \int y^2 f\left((1+\xi_k)y\right)L_k(y)\,dH(y).
    \end{align*}
    \[
    \left|\frac{\partial \EXPcostCP(\bv)}{\partial \xi_k} \right| \leq 2 \bar{v}^2\bar{f}.  
    \]
    And for $i \neq k$, 
\begin{align*}
    \frac{\partial \EXPcostCP(\bxi)}{\partial \xi_i} = \int y \int_{t \geq \left(1+\xi_k\right)y}t f_i\left(\frac{1+\xi_i}{1+\xi_k}t\right)f_k\left(t\right) \prod_{j\ne i,k}F_j\left(\frac{1+\xi_j}{1+\xi_k}t\right)\,dx \,dH(y),
\end{align*}
\[
    \left|\frac{\partial \EXPcostCP(\bxi)}{\partial \xi_i} \right| \leq  \bar{v}^2\bar{f}^2.  
\]
    
\end{enumerate}
\end{proof}

\begin{lemma}
\label{lem: lipschitz hp}
 $\EXPcostHP(\mu_k,\blambda)$ and $\EXPvalueHP(\mu_k,\blambda)$ are Lipschitz continuous. Moreover, if for all $k\in \calK$, bidder $k$'s density function has a lower bound $\underline{f} > 0$ over the interval $[0, \bar{v}_k]$, and density function $h$ also has a lower bound $\underline{h} > 0$ over the interval $[0, \bar{v}]$, then there exists constants $\underline{G}' > 0$ such that $$\frac{\partial \EXPcostHP(\mu_k, \blambda)}{\partial \mu_k} \leq - \underline{G}' < 0.$$
\end{lemma}

\begin{proof}
In hybrid coordinated pacing algorithms, we still denote by $L_k$ the cumulative distribution function of $d^I_k$:
\begin{align}
\label{eqn: cdf of d i hp}
L_k(x; \mu_k,\blambda) = \prod_{i\in \calK: i\ne k}F_i((1+\lambda_i)x).
\end{align}
Note that $L_k$ is independent of $\mu_k$. We omit $\mu_k$ and $\blambda$ when there is no confusion. In what follows, we denote by $\overline{f}$ the upper bound on $f_k$ and by $\overline{h}$ the upper bound on $h_k$. We use $\bar{F}_k$ to denote the omplementary cumulative distribution function of $v_k$.

We first consider bidder $k$'s expected expenditure under \Cref{alg: hybrid coordinated pacing} and calculate its derivatives with respect to $\mu_k$, $\lambda_k$ and $\lambda_i$ for $i\ne k$.
\begin{align}
\EXPcostHP(\mu_k,\blambda)& =\int_y y \int_{x\geq \frac{1+\mu_k}{1+\lambda_k}y}\bar{F}_k((1+\lambda_k)x)\,dL_k(x)\,dH(y) \notag \\
& +\int_y y \bar{F}_k((1+\mu_k)y)L_k\left(\frac{1+\mu_k}{1+\lambda_k}y\right)\,dH(y) \label{eqn: G hp first expression}\\
& = \int_y y\int_{x\geq \frac{1+\mu_k}{1+\lambda_k}y}(1+\lambda_k)f_k((1+\lambda_k)x)L_k(x)\,dx \,dH(y). \label{eqn: G hp second expression}
\end{align}
\begin{enumerate}
    \item Using the second expression \eqref{eqn: G hp second expression} we obtain that
    \[
    \frac{\partial \EXPcostHP(\mu_k, \blambda)}{\partial \mu_k} = -\int_y y^2 f_k\left(\left(1+\mu_k\right)y\right)L_k\left(\frac{1+\mu_k}{1+\lambda_k}y\right)\,dH(y),
    \]
    and thus,
    \[
    \left|\frac{\partial \EXPcostHP(\mu_k, \blambda)}{\partial \mu_k}\right|\leq \bar{v}_k^2 \bar{f}. 
    \]
    Let $\underline{v} = \min_{i\in \calK} \bar{v}_i$ and $\bar{\mu} = \max_{i\in \calK} \bar{\mu}_i$. When $y \leq \underline{v}/(1+\bar{\mu})$ one has $$\frac{(1+\lambda_i)(1+\mu_k)}{1+\lambda_k}y \leq (1+\lambda_i)y \leq \underline{v} \leq \bar{v}_i, $$ and consequently, 
    \begin{align}
        L_k\left(\frac{1+\mu_k}{1+\lambda_k}y\right) = & \prod_{i\in \calK: i\ne k}F_i\left(\frac{(1+\lambda_i)(1+\mu_k)}{1+\lambda_k}y\right) \notag \\
        \geq & \left(\frac{(1+\lambda_i)(1+\mu_k)}{1+\lambda_k}y\underline{f}\right)^{K-1} \geq \left(\frac{y\underline{f}}{1+\bar{\mu}}\right)^{K-1}. \label{eqn: L hp bound}
    \end{align}
    We can upper bound the derivative $\partial \EXPcostHP(\mu_k, \blambda)/\partial \mu_k$ as follows:
    \begin{align*}
    \label{eqn: G derivative in mu}
        \frac{\partial \EXPcostHP(\mu_k, \blambda)}{\partial \mu_k} \overset{(\text a)}{\leq} & - \underline{f}\underline{h}\int_{0}^{\bar{v}_k/(1+\mu_k)} y^2L_k(\frac{1+\mu_k}{1+\lambda_k}y)\,dy \\
        \overset{(\text b)}{\leq} & - \underline{f}\underline{h}\int_{0}^{\underline{v}/(1+\bar{\mu})} y^2\left(\frac{y\underline{f}}{1+\bar{\mu}}\right)^{K-1}\,dy \\
        = & - \frac{\underline{v}^{K+2}\underline{f}^{K}\underline{h}}{(K+2)(1+\bar{\mu})^{2K}},
    \end{align*}
    where (a) follows $f_k \geq \underline{f}$ and $h_k \geq \underline{h}$ when $(1+\mu_k)y \leq \bar{v}_k$; (b) follows from \eqref{eqn: L hp bound} and $\bar{v}_k/(1+\mu_k) \geq \underline{v}/(1+\bar{\mu})$.
    \item Using the first expression \eqref{eqn: G hp first expression} we obtain that
    \begin{align*}
        \frac{\partial \EXPcostHP(\mu_k, \blambda)}{\partial \lambda_k}= - \int_y y \int_{x \geq \frac{1+\mu_k}{1+\lambda_k}y}xf_k\left((1+\lambda_x) x\right) \, dL_k(x)\, dH(y), 
    \end{align*}
    and thus,
    \[
    \left|\frac{\partial \EXPcostHP(\mu_k, \blambda)}{\partial \lambda_k}\right|\leq \bar{v}_k^2 \bar{f}. 
    \]
    \item Using the second expression \eqref{eqn: G hp second expression} we obtain that for $i\neq k$, 
    \begin{align*}
        \frac{\partial \EXPcostHP(\mu_k, \blambda)}{\partial \lambda_i}=\int_y y\int_{x\geq \frac{1+\mu_k}{1+\lambda_k}y}(1+\lambda_k)f_k((1+\lambda_k)x)\frac{\partial L_k(x)}{\partial \lambda_i}\,dx \,dH(y),
    \end{align*}
    where the derivative of $L_k(x)$ satisfies
    \begin{align*}
    \frac{\partial L_k(x; \mu_k,\blambda)}{\partial \lambda_i} = xf_i((1+\lambda_i)x)\prod_{i\in \calK: j\ne i, k}F_j((1+\lambda_j)x) \leq \bar{v}_k\bar{f}.
    \end{align*}
    Therefore, we have
    \[
    \left|\frac{\partial \EXPcostHP(\mu_k, \blambda)}{\partial \lambda_k}\right|\leq \bar{v}_k^2\bar{f}. 
    \]
\end{enumerate}

We next consider $\EXPvalueHP(\mu_k, \blambda)$ and its derivatives.

\begin{align}
    \EXPvalueHP(\mu_k,\blambda) = \int_x x L_k\left(\frac{x}{1+\lambda_k}\right)H\left(\frac{x}{1+\mu_k}\right)\,dF_k(x). \label{eqn: V hp expression}
\end{align}

\begin{enumerate}
    \item We have that 
    \begin{align*}
        \frac{\partial \EXPvalueHP(\mu_k, \blambda)}{\partial \mu_k} = - \int_x \left(\frac{x}{1+\mu_k}\right)^2 L_k\left(\frac{x}{1+\lambda_k}\right)h\left(\frac{x}{1+\mu_k}\right)\,dF_k(x), 
    \end{align*}
    and thus, 
    \[
    \left|\frac{\partial \EXPvalueHP(\mu_k, \blambda)}{\partial \mu_k}\right| \leq \bar{v}_k^2\bar{h}. 
    \]
    \item We have
    \begin{align*}
        \frac{\partial \EXPvalueHP(\mu_k, \blambda)}{\partial \lambda_k} = - \int_x \left(\frac{x}{1+\lambda_k}\right)^2 l_k\left(\frac{x}{1+\lambda_k}\right)H\left(\frac{x}{1+\mu_k}\right)\,dF_k(x), 
    \end{align*}
    and thus, 
    \[
    \left|\frac{\partial \EXPvalueHP(\mu_k, \blambda)}{\partial \lambda_k}\right| \leq (K-1)\bar{v}_k^2\bar{f}. 
    \]
    \item For $i\neq k$, the derivative with respect to $\lambda_i$ is 
    \begin{align*}
        \frac{\partial \EXPvalueHP(\mu_k, \blambda)}{\partial \lambda_i}=\int \frac{x^2}{1+\lambda_k}f_i\left(\frac{1+\lambda_i}{1+\lambda_k}x\right)\prod_{j\ne i,k} F_j \left(\frac{1+\lambda_j}{1+\lambda_k}x\right)H\left(\frac{x}{1+\mu_k}x\right)\, dF_k(x),
    \end{align*}
    and thus,
    \[
    \left|\frac{\partial \EXPvalueHP(\mu_k, \blambda)}{\partial \lambda_i}\right|\leq \bar{v}_k^2\bar{f}. 
    \]
\end{enumerate}

\end{proof}

\begin{corollary}
\label{cor: lipschitz}
$\EXPutility(\blambda)$, $\EXPutilityCP(\blambda)$ and $\EXPutilityHP(\mu_k,\blambda)$ are Lipschitz continuous.
\end{corollary}

\begin{proof}
The results directly follow from \Cref{lem: lipschitz cp} and \Cref{lem: lipschitz hp} since 
\begin{align*}
    \frac{1}{T}\EXPutility(\blambda) &= \EXPvalue(\blambda)-\EXPcost(\blambda);\\
    \frac{1}{T}\EXPutilityCP(\bxi) &= \EXPvalueCP(\bxi)-\EXPcostCP(\bxi);\\
    \frac{1}{T}\EXPutilityHP(\mu_k,\blambda) &=\EXPvalueHP(\mu_k,\blambda) - \EXPcostHP(\mu_k,\blambda). 
\end{align*}
\end{proof}

\begin{lemma}
\label{lem: strictly pareto}
Suppose that \Cref{asm: hp} holds. We have
\begin{align*}
    \EXPcostHP(\lambda^*_k, \blambda^*) \leq \EXPcost(\blambda^*), \forall k,
\end{align*}
and the equality strictly holds for at most one bidder.
\end{lemma}

\begin{proof}

In what follows, we use $L_k(x)$ to denote the cumulative distribution function of $d^I_k$ defined by \eqref{eqn: cdf of d i hp} and $l_k(x)$ is the corresponding density function. We analyze how much $\EXPcostHP(\lambda_k, \blambda)$ and $\EXPcost(\blambda)$ differ by their expressions.
\begin{align*}
    \EXPcostHP(\lambda_k, \blambda) & \overset{(\text a)}{=} \int_x (1+\lambda_k)f_k((1+\lambda_k)x)L_k(x) \left(\int_{y\leq x}y\,dH(y)\right) \,dx \\
    & \overset{(\text b)}{=} \int_x (1+\lambda_k)xf_k((1+\lambda_k)x)L_k(x)H(x)\,dx \\
    & - \int_x (1+\lambda_k)f_k((1+\lambda_k)x)L_k(x) \left(\int_{y\leq x}H(y)\,dy\right) \,dx \\
    & \overset{(\text c)}{=} \int_x \left((1+\lambda_k)xf_k((1+\lambda_k)x) - \bar{F}_k((1+\lambda_k)x)\right)L_k(x)H(x)\,dx \\
    & - \int_x \bar{F}_k((1+\lambda_k)x)l_k(x) \left(\int_{y\leq x}H(y)\,dy\right) \,dx \\
    & \overset{(\text d)}{=} \EXPcost(\blambda) - \underbrace{\int_x \bar{F}_k((1+\lambda_k)x)l_k(x) \left(\int_{y\leq x}H(y)\,dy\right) \,dx}_{\Delta_k(\blambda)}.
\end{align*}
where (a) follows from \eqref{eqn: G hp second expression} and changes the order of integration; (b) and (c) use integration by parts; and (d) follows from \eqref{eqn: G second expression}. Observe that $\Delta_k(\blambda)$ is always non-negative. Moreover, we claim that when \Cref{asm: hp} holds, $\Delta_k(\blambda^*)$ must be strictly positive for all but one bidder, by which the proof can be completed. 

To prove the claim, we first show that for every bidder $k\in \calK$, 
\begin{align}
\label{eqn: everyone can beat outside with positive prob}
    H\left(\frac{\hat{v}_k}{1+\lambda^*_k}\right) > 0,
\end{align}
where $\hat{v}_k \coloneqq\sup\{v_k : \bar{F}_k(v_k) > 0\}$. Otherwise, $H(\hat{v}_k/(1+\lambda^*_k)) = 0$ implies that bidder $k$ can, almost surely, never win the auction by bidding $v_k/(1+\lambda^*_k)$ or lower when bidders are independently bidding. Note that $\lambda^*_k < \bar{\mu}_k$ because
\begin{align*}
    \EXPcost\left(\bar{\mu}_k, \blambda^*_{-k}\right) = \bbE \left[\bm{1}\left\{\frac{v_k}{1+\bar{\mu}_k}\geq d_k\right\}d_k\right] \overset{(\text a)}{\leq} \bbE \left[\bm{1}\left\{\frac{v_k}{1+\bar{v}_k/\rho_k}\geq d_k\right\}d_k\right] \leq \frac{v_k}{1+\bar{v}_k/\rho_k} < \rho_k,
\end{align*}
where (a) follows from $\bar{\mu}_k \geq \bar{v}_k/\rho_k$. Thus, $\EXPcost((\lambda_k, \blambda^*_{-k})) \equiv 0$ on $[\lambda^*_k, \bar{\mu}_k]$ contradicts the strong monotonicity assumption. 

Let $k'=\arg\min_k \hat{v}_k/(1+\lambda^*_k)$ and let $\underline{d}^O \coloneqq \inf\{d^O: H(d^O) > 0\}$. Note that by \eqref{eqn: everyone can beat outside with positive prob} and the continuity of $H$, we must have $\underline{d}^O < \hat{v}_{k'}/(1+\lambda^*_{k'})$. For every $k \neq k'$, when $x \in (\underline{d}^O, \hat{v}_{k'}/(1+\lambda^*_{k'}))$, we have \begin{enumerate}
    \item $\bar{F}_k((1+\lambda^*_k)x) > 0$;
    \item $\int_{y\leq x}H(y)\,dy > 0$.
\end{enumerate}
Suppose for contradiction that for $k \neq k'$, $\Delta_k(\blambda^*) = 0$. Then $l_k(x) = 0$ holds almost everywhere on $[\underline{d}^O, \hat{v}_{k'}/(1+\lambda^*_{k'})]$, i.e., the competing bid $d^I_k$ lies in this interval with zero probability. However, $b_{k'} =  v_{k'}/(1+\lambda^*_{k'})$ belongs to $[\underline{d}^O, \hat{v}_{k'}/(1+\lambda^*_{k'})]$ with positive probability by the definition of $\hat{v}_{k'}$. This implies there is almost always another bid higher than $b_{k'}$ so that $b_{k'}$ almost never becomes $d^I_k$, which further implies that bidder $k'$ actually can, almost surely, never win in IP by bidding $v_{k'}/(1+\lambda^*_{k'})$ or lower. Again by a similar argument, this contradicts the strong monotonicity assumption. Therefore, $\Delta_k(\blambda^*) > 0$ holds at least for all bidders except bidder $k'$.
\end{proof}
\section{More Experiment Results on Synthetic Data}
\label{sec: app-exp}

\begin{figure*}[t]
\begin{subfigure}{\textwidth}
    \centering
    \includegraphics[width=\linewidth]{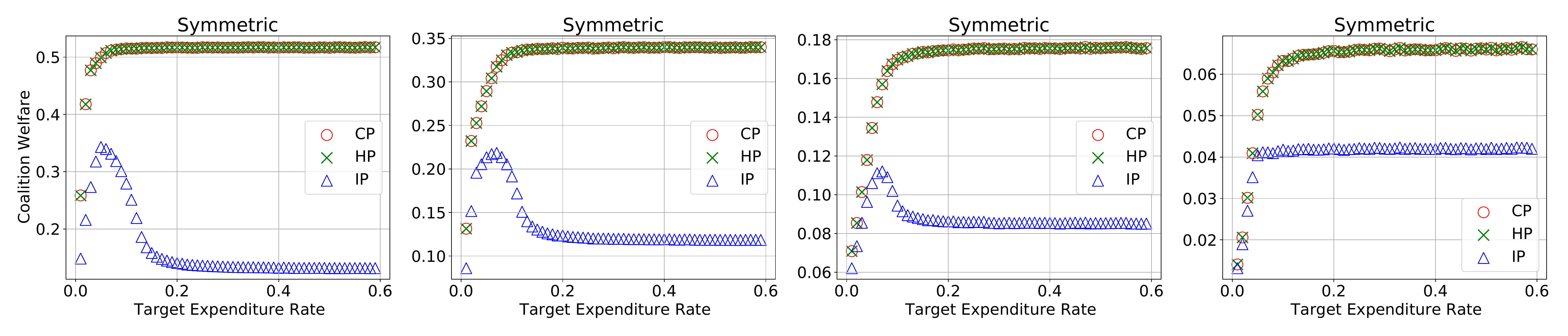}
    \caption{Symmetric cases: coalition welfare with different means of the highest bid outside the coalition: $0.2,0.4,0.6,0.8$.}
    \label{fig:Appendix.Symmetric}
\end{subfigure}
\begin{subfigure}{\textwidth}
    \centering
    \includegraphics[width=\linewidth]{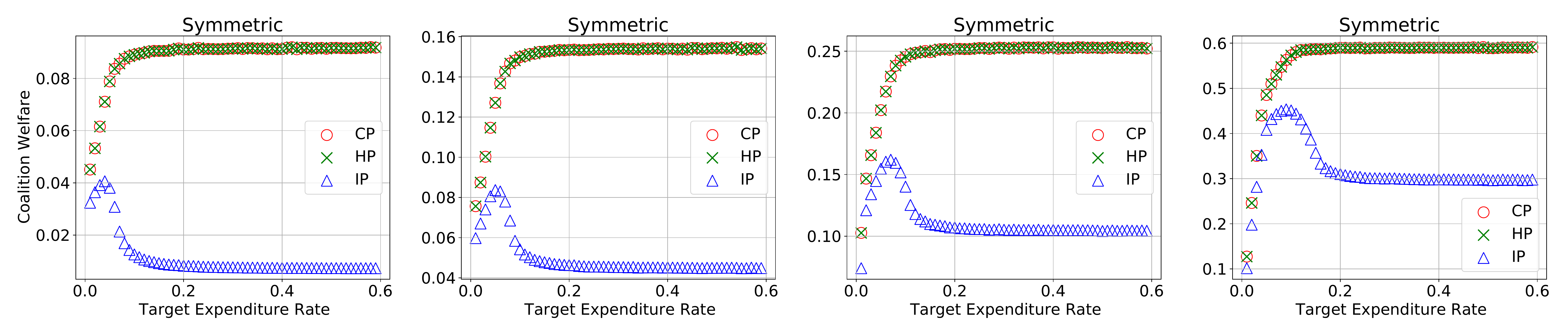}
    \caption{Symmetric cases: coalition welfare with different variances of bidders' values: $0.02, 0.1, 0.2, 0.5$. }
    \label{fig:Appendix.Symmetricvariance}
\end{subfigure}
\begin{subfigure}{\textwidth}
    \centering
    \includegraphics[width=1\linewidth]{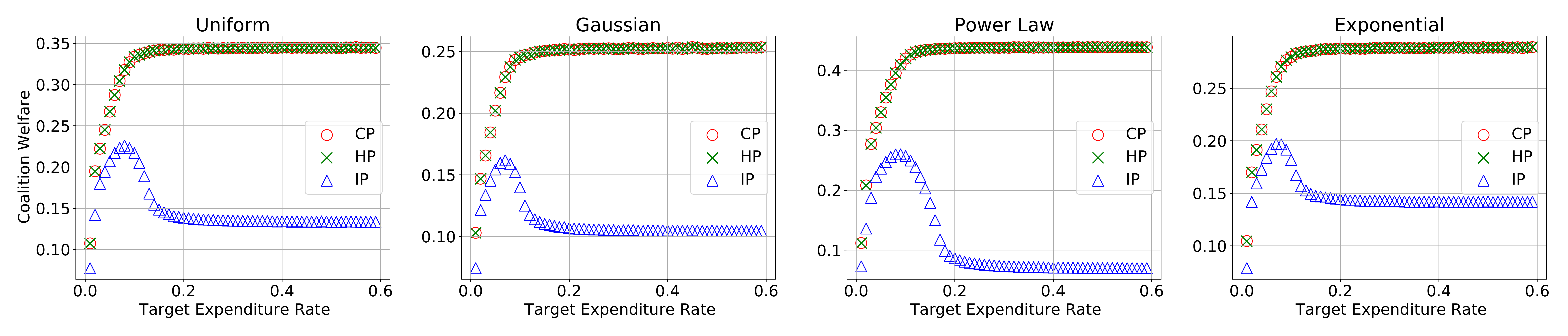}
    \caption{Symmetric cases: coalition welfare with different members' value distributions: Uniform, Gaussian, power law distribution, and exponential.}
    \label{fig:Appendix.Symmetricdistri1}
\end{subfigure}
\caption{Experiment Results on Symmetric Synthetic Data}
\label{fig:symmetric}
\end{figure*}

We present more experiment results on synthetic data here. All the following cases have the same number of coalition members: $K=5$.

\Cref{fig:symmetric} show the results on symmetric cases, where all $5$ members have the same value distribution and the same budget. Specifically, \Cref{fig:Appendix.Symmetric} shows the coalition welfare when the distribution of the highest bid outside the coalition $d^O$ varies. From left to right, $d^O$ is sampled from Gaussian distribution with $\mu$ equals to $0.2$, $0.4$, $0.6$, $0.8$ respectively. With the mean increases, bidders' power outside the coalition becomes stronger. 

In \Cref{fig:Appendix.Symmetricvariance}, the coalition members' value distributions are all Gaussian distributions with $\mu = 0.5$ but they have different standard deviations $\sigma$. From left to right, $\sigma$ equals to $0.02$, $0.1$, $0.2$, $0.5$ respectively.

The four subfigures in \Cref{fig:Appendix.Symmetricdistri1} compare the coalition welfare when members' values are sampled from different distributions: uniform distribution, Gaussian distribution, power law distribution, and exponential distribution.

\Cref{fig:Appendix.Asymmetric} shows four asymmetric cases. Different from above, the target expenditure rates for coalition members have the same expectation rather than being completely the same. 
From left to right, coalition member's values are sampled from totally the same distributions, same distributions with different means, and completely different distributions, respectively. 

In \Cref{fig:Appendix.Asymetric1}, each member's value is sampled from Gaussian distribution with $\mu=0.5$, $\sigma=0.2$. 
In \Cref{fig:Appendix.Asymetric2}, the five members' value are respectively sampled from Gaussian distributions with $\mu$ equals to $0.1$, $0.3$, $0.5$, $0.7$, $0.9$.
In \Cref{fig:Appendix.Asymetric3}, the five members' value are respectively sampled from uniform distributions on $[0, 0.2]$, $[0.2, 0.4]$, $[0.4, 0.6]$, $[0.6, 0.8]$, $[0.8, 1]$.
In \Cref{fig:Appendix.Asymetric4}, the value functions are more complicated. We choose one uniform distribution, one Gaussian distribution, two power law distributions, and one exponential distribution. We adjust the parameters so that the value expectation of all members is $0.5$, which is the same as the others in \Cref{fig:Appendix.Asymmetric}.

To verify if the averages of the multipliers indeed converge to the profile defined by the non-linear complementarity conditions \eqref{eqn: lambda star definition}, \eqref{eqn: xi definition} and \eqref{eqn: mu star definition}, we check the variances of the averages and the empirical value of $(\sum^T_{t=1}\lambda_{k,t}/T)(\rho_k-\EXPcost(\sum^T_{t=1}\blambda_{t}/T))$ ($(\sum^T_{t=1}\xi_{k,t}/T)(\rho_k-\EXPcostCP(\sum^T_{t=1}\bxi_{t}/T))$ and $(\sum^T_{t=1}\mu_{k,t}/T)(\rho_k-\EXPcostHP(\sum^T_{t=1}\blambda_{t}/T, \sum^T_{t=1}\mu_{k,t}/T))$).  To approximate $\EXPcost(\sum^T_{t=1}\blambda_{t})$, we use the empirical average expenditure $\sum^T_{t=1}z_{k,t}/T$. The same goes for $\EXPcost(\sum^T_{t=1}\bxi_{t})$ and $\EXPcost(\sum^T_{t=1}\bxi_{t})$. In the experiments, the variances of the last 100 rounds of $\sum\lambda_{k,t}/T$ ($\sum\xi_{k,t}/T$ and $\sum\mu_{k,t}/T$) are around $10^{-4}$, indicating the convergence of multipliers. The values of the last 100 rounds of $(\sum\lambda_{k,t}/T)(\rho_k-\sum z_{k,t}/T)$ ($(\sum\xi_{k,t}/T)(\rho_k-\sum z_{k,t}/T)$ and $(\sum\mu_{k,t}/T)(\rho_k-\sum z_{k,t}/T)$) are less than $10^{-3}$, indicating the convergence to the point defined by \eqref{eqn: lambda star definition}, \eqref{eqn: xi definition} and \eqref{eqn: mu star definition}. 

\begin{figure*}[t]
    \begin{subfigure}{0.246\linewidth}
        \centering
        \includegraphics[width = \linewidth]{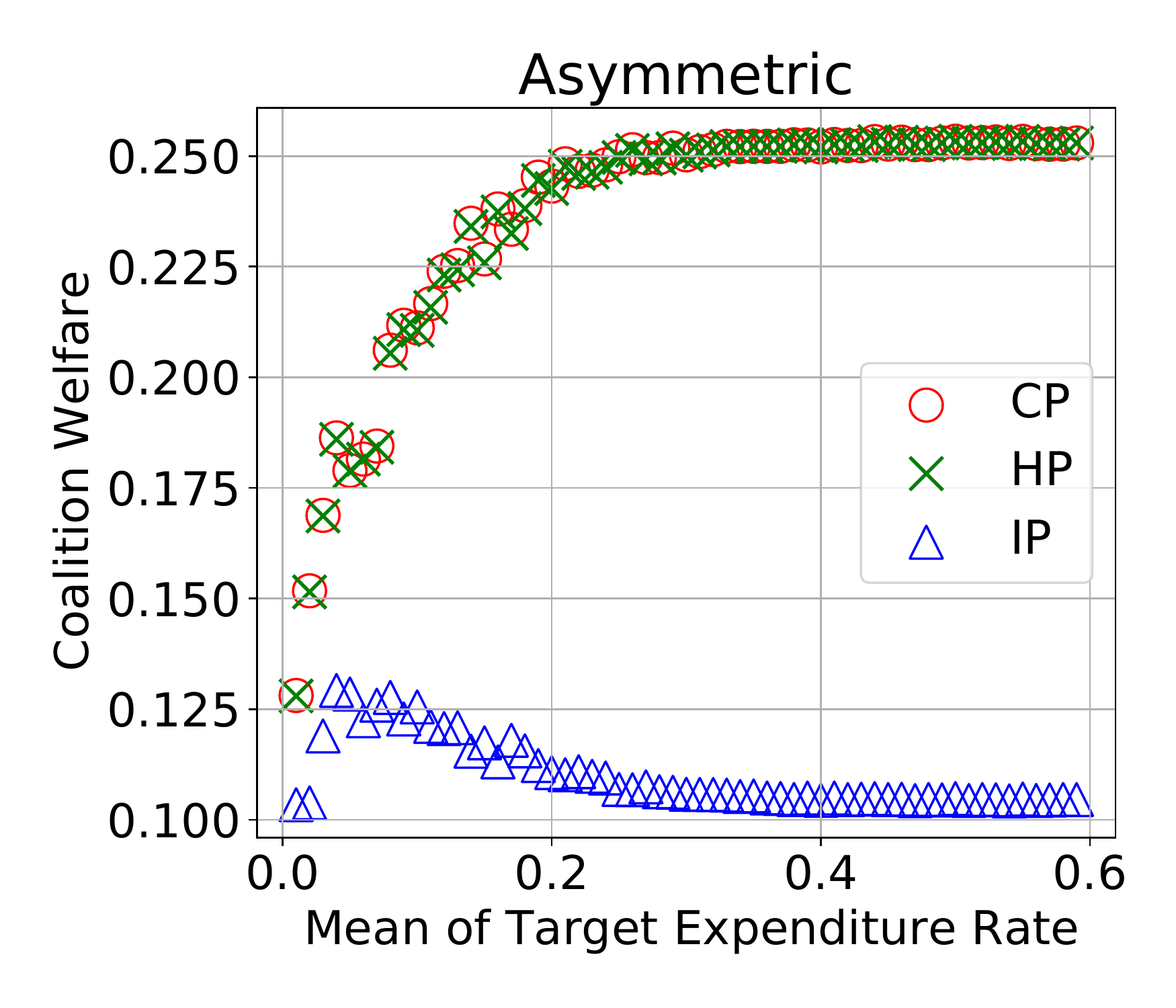}
        \caption{Gaussian with same means. }
        \label{fig:Appendix.Asymetric1}
    \end{subfigure}
    \begin{subfigure}{0.246\linewidth}
        \centering
        \includegraphics[width = \linewidth]{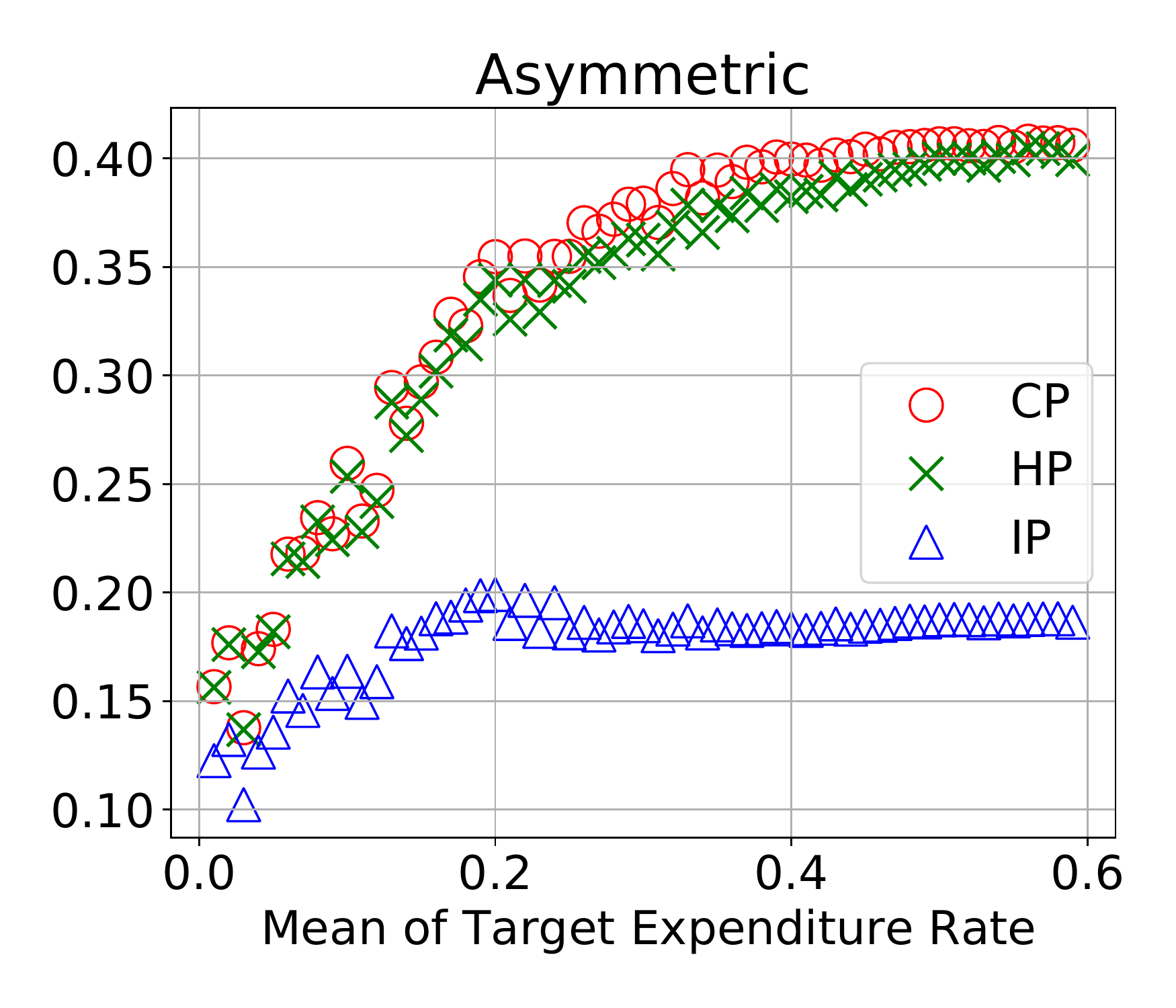}
        \caption{Gaussian with different means. }
        \label{fig:Appendix.Asymetric2}
    \end{subfigure}
    \begin{subfigure}{0.246\linewidth}
        \centering
        \includegraphics[width = \linewidth]{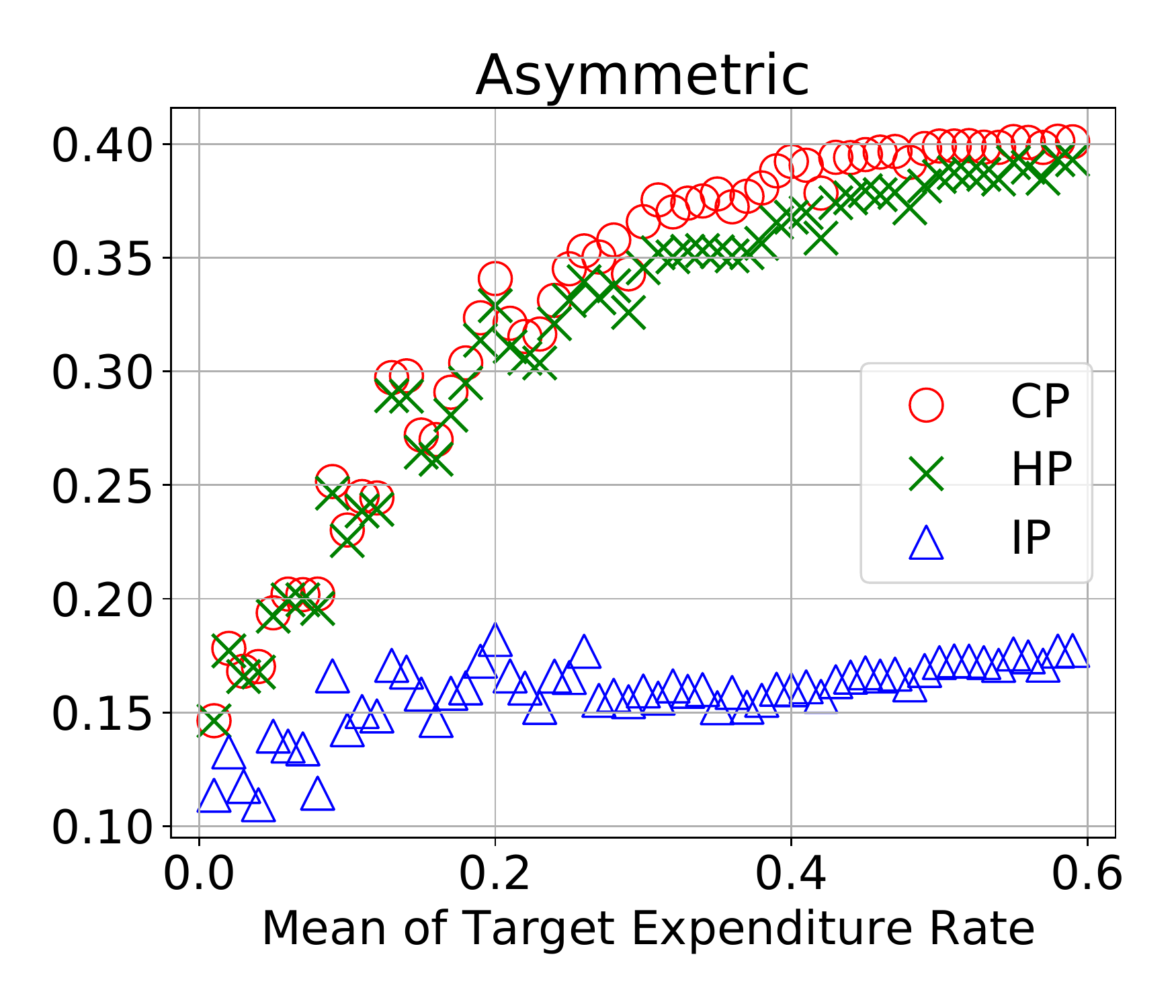}
        \caption{Uniform with different means. }
        \label{fig:Appendix.Asymetric3}
    \end{subfigure}
    \begin{subfigure}{0.246\linewidth}
        \centering
        \includegraphics[width = \linewidth]{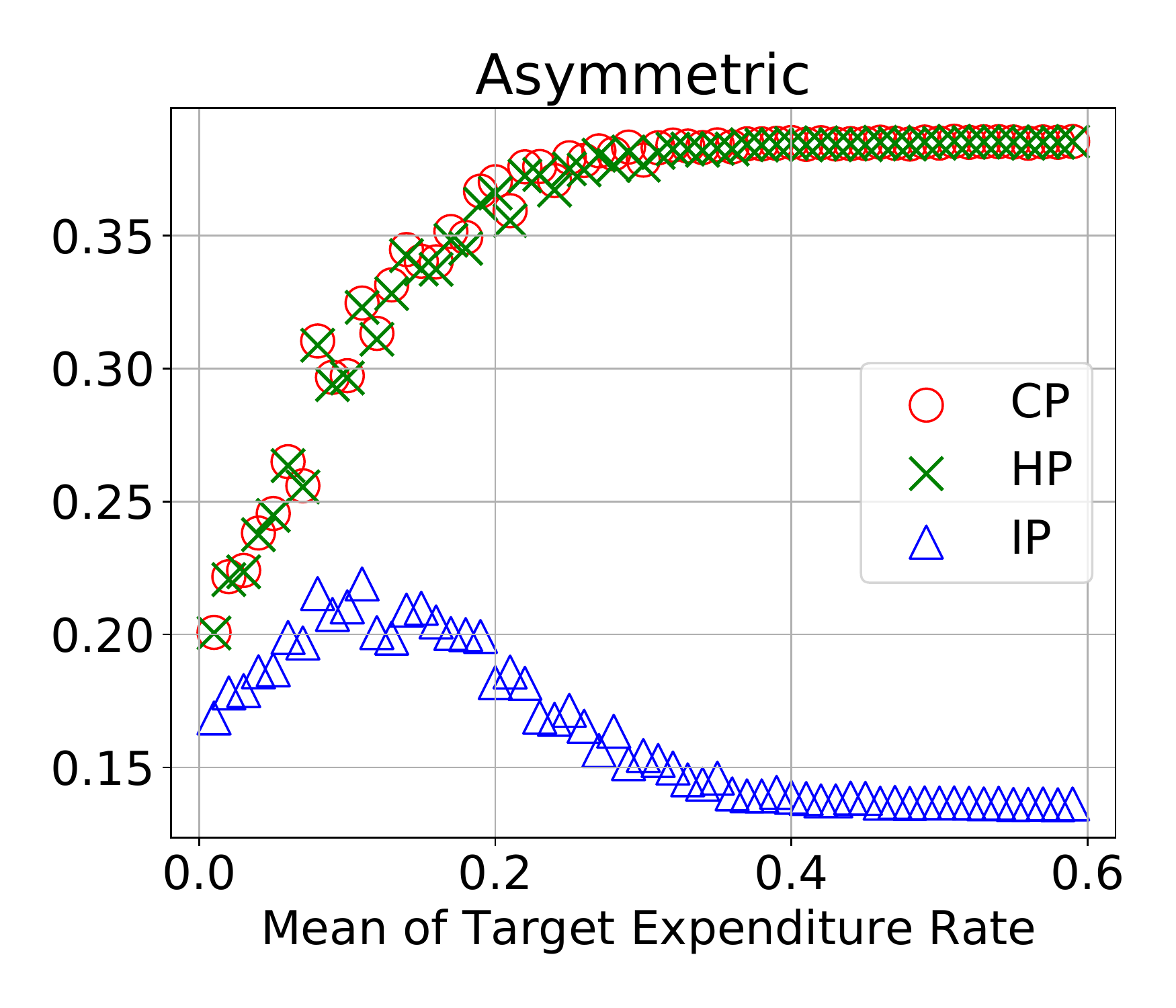}
        \caption{Different value distributions. }
        \label{fig:Appendix.Asymetric4}
    \end{subfigure}
    \caption{Experiment Results on Asymmetric Synthetic Data}
    \label{fig:Appendix.Asymmetric}
\end{figure*}

%%%%%%%%%%%%%%%%%%%%%%%%%%%%%%%%%%%%%%%%%%%%%%%%%%%%%%%%%%%%%%%%%%%%%%%%%%%%%%%
%%%%%%%%%%%%%%%%%%%%%%%%%%%%%%%%%%%%%%%%%%%%%%%%%%%%%%%%%%%%%%%%%%%%%%%%%%%%%%%

\end{document}